\pdfoutput=1
\documentclass[11pt, a4paper]{article}

\usepackage[sc]{mathpazo}

\usepackage{amsmath}
\usepackage{enumitem}
\usepackage{amssymb}
\usepackage{amsfonts}
\usepackage{textcomp}
\usepackage{algorithm}
\usepackage{authblk}
\usepackage[noend]{algpseudocode}
\usepackage{amsthm}
\usepackage{xspace}
\usepackage{fullpage}
\usepackage[english]{babel}
\usepackage[utf8]{inputenc}
\usepackage[classfont = bold, funcfont=roman]{complexity}
\usepackage[pagebackref=true, colorlinks=true, citecolor=DarkGreen, linkcolor=Navy]{hyperref}
\usepackage{mathtools}
\usepackage[svgnames]{xcolor}
\usepackage{bussproofs}
\usepackage{subfig}
\usepackage{thmtools}
\usepackage{thm-restate}
\usepackage[capitalise, noabbrev, nameinlink]{cleveref}
\usepackage{bm}

\usepackage[scaled=.92]{helvet}

\usepackage[activate={true,nocompatibility},
selected=true,
tracking=true,
factor=1100,
babel=true
]{microtype}

\usepackage
[labelformat=default, labelsep=colon, textformat=simple, font={small,sf}, labelfont=bf ]{caption}

\usepackage{sectsty}
\allsectionsfont{\sffamily}

\makeatletter
 \renewenvironment{abstract}{\small \begin{center}{\bfseries \sffamily\abstractname\vspace{-.5em}\vspace{\z@}}\end{center}\quotation}
 {\endquotation}
\makeatother

\theoremstyle{plain}
\newtheorem{theorem}{Theorem}[section]
\newtheorem{lemma}[theorem]{Lemma}
\newtheorem{corollary}[theorem]{Corollary}

\theoremstyle{definition}
\newtheorem{remark}[theorem]{Remark}

\newtheorem{definition}[theorem]{Definition}
\newtheorem{invariant}[theorem]{Invariant}

\DeclareMathOperator*{\supp}{supp}

\DeclareMathOperator{\dom}{dom}
\DeclareMathOperator{\depth}{depth}

\newcommand{\eqcomma}{\enspace ,}
\newcommand{\eqperiod}{\enspace .}

\newcommand{\set}[1]{\{#1\}}

\newcommand\restrict[2]{{
  #1 \lceil_{#2}
  }}
\usepackage{bbm}

\newcommand{\calD}{\mathcal{D}}

\newcommand{\calP}{\mathcal{P}}
\newcommand{\calQ}{\mathcal{Q}}
\newcommand{\calR}{\mathcal{R}}

\newcommand{\calT}{\mathcal{T}}
\newcommand{\calU}{\mathcal{U}}

\newcommand{\N}{\mathbb{N}}

\DeclareMathOperator{\closure}{closure}

\newcommand{\row}{\operatorname{row}}
\newcommand{\col}{\operatorname{col}}

\newcommand{\cl}{K}

\DeclareMathOperator{\sign}{sign}

\newcommand{\centers}{C}
\newcommand{\chosen}[1][\sigma]{\centers_{{#1}}}

\newcommand{\spaceRho}{\ensuremath R}
\newcommand{\spaceRhoReg}{\ensuremath \spaceRho^{\mathrm{reg}}}

\DeclareMathOperator{\tseitin}{Tseitin}

\crefname{prop}{Property}{Properties}
\Crefname{prop}{Property}{Properties}

\crefname{inv}{Invariant}{Invariants}
\Crefname{inv}{Invariant}{Invariants}

\crefname{invariants}{Invariants}{Invariants}
\Crefname{invariants}{Invariants}{Invariants}

\crefname{step}{Step}{Steps}
\Crefname{step}{Step}{Steps}

\crefname{case}{Case}{Cases}
\Crefname{case}{Case}{Cases}

\newcommand{\factorN}{D}

\title{On bounded depth proofs for Tseitin\\formulas on the grid; revisited\thanks{Supported by the Approximability and Proof Complexity project
    funded by the Knut and Alice Wallenberg Foundation. Kilian Risse
    was supported by the Swiss National Science Foundation through
    project \mbox{200021-184656} “Randomness in Problem Instances and
    Randomized Algorithms” and the Postdoc.Mobility fellowship
    \mbox{P500-2 235298}. Most of this work was done while the second
    author was affiliated with KTH Royal Institute of Technology and
    EPFL. This is the full-length version of a paper with the same
    title that appeared in the \emph{Proceedings of the 63rd Annual
      IEEE Symposium on Foundations of Computer Science (FOCS '22)}.}}

\author[1]{Johan H{\aa}stad}

\author[2]{Kilian Risse}

\affil[1]{KTH Royal Institute of Technology}

\affil[2]{Lund University}

\date{\today}

\begin{document}
\maketitle

\begin{abstract}
  We study Frege proofs using depth-$d$ Boolean formulas for the
  Tseitin contradiction on $n \times n$ grids.  We prove that if each
  line in the proof is of size $M$ then the number of lines is
  exponential in $n/(\log M)^{O(d)}$. This strengthens a recent result
  of Pitassi et al.\,\cite{PRT21}. The key technical step is a
  multi-switching lemma extending the switching lemma of H{\aa}stad
  \cite{jhtseitin} for a space of restrictions related to the Tseitin
  contradiction.
  
  The strengthened lemma also allows us to improve the lower bound for
  standard proof size of bounded depth Frege refutations from
  exponential in $\tilde \Omega (n^{1/59d})$ to exponential in
  $\tilde \Omega (n^{1/d})$. This strengthens the bounds given in the
  preliminary version of this paper \cite{jhkr}.
\end{abstract}
 
\pagenumbering{roman}
\thispagestyle{empty}
\newpage

\tableofcontents
\newpage

\pagenumbering{arabic}
\setcounter{page}{1}

\section{Introduction}

Mathematicians like proofs, formal statements where each line follows
by simple reasoning rules from previously derived lines.  Each line
derived in this manner, assuming that the reasoning steps are sound,
can give us some insight into the initial assumptions of the proof. A
particularly interesting consequence is contradiction. Deriving an
obviously false statement allows us to conclude that the initial
assumptions, also called axioms, are contradictory. We continue the
study of Frege proofs of contradiction where each line in the proof is
a Boolean formula of depth $d$. This subject has a long tradition, so
let us start with a very brief history.

A very basic proof system is resolution: each line of such a proof
simply consists of a disjunction of literals. The derivation rules of
resolution are also easy to understand and simple to implement, but the
proof system nevertheless gives rise to reasonably short proofs for
some formulas.
It is far from easy to give lower bounds for the size of proofs in
resolution but it has been studied for a long time and by now many
strong bounds are known.  An early paper by Tseitin \cite{tseitin}
defined an important class of contradictions based on graphs that is
central to this and many previous papers.  For each edge there is a
variable and the requirement is that the parity of the variables
incident to any given node sum to a particular bit which is called the
charge of that node. If the sum of the charges is one modulo two
this is a contradiction. For a subsystem of resolution, called regular
resolution, Tseitin proved exponential lower bounds on refutations of
these formulas. After this initial lower bound it took almost another
two decades before the first strong lower bound for general resolution
was obtained by Haken \cite{haken}, whose lower bound applied to the
pigeonhole principle (PHP). Many other resolution lower bounds
followed, but as we are not so interested in resolution and rather
intend to study the more powerful proof system with formulas of
larger, though still bounded, depth $d$ on each line, let us turn to
such proof systems.

The study of proofs with lines limited to depth $d$ dates back several
decades. A pioneering result was obtained by Ajtai
\cite{Ajtai94Complexity} who showed that the PHP cannot be proved in
polynomial size for any constant depth $d$. Developments continued in
the 1990s and polynomial size proof were ruled out for values of $d$
up to $O(\log \log n)$ for both the PHP
\cite{PBI93ExponentialLowerBounds,KPW95FregePHPExp} as well as the
Tseitin contradiction defined over complete
\cite{UF96SimplifiedLowerBounds} and expander graphs
\cite{BenSasson02HardExamples}.

These developments followed previous work where the
computational power of the class of circuits\footnote{When the depth
  is small, there is no major difference between circuits and formulas
  so the reader should feel free to ignore this difference.}  of depth
$d$ was studied \cite{sipser, fss, yao, jhswitch, razborov88,
  smolensky}.  It is not surprising that it is easier to understand
the computational power of a single circuit rather than to reason
about a sequence of formulas giving a proof.  This manifested itself
in that while the highest value of $d$ for which strong bounds were
known for size of proofs remained at $O(\log \log n)$, the results for
circuit size extended to almost logarithmic depth.

This gap was (essentially) closed in two steps.  First Pitassi et
al.\,\cite{pitassi16frege} proved super-polynomial lower bounds for $d$
up to $o(\sqrt {\log n})$ and then H{\aa}stad \cite{jhtseitin}
extended this to depth $\Theta (\frac {\log n} {\log \log n})$ which,
up to constants, matches the result for circuits.

The key technique used in most of the described results is the use of
restrictions.  These set most of the variables to constants which
simplifies the circuit or formulas studied.  If done carefully one can
at the same time preserve the contradiction refuted or the function
computed.  Of course one cannot exactly preserve the contradiction and
to be more precise a contradiction with parameter $n$ before the
restriction turns into a contradiction of the same type but with a
smaller parameter, $n/T$, after the restriction.

The simplification under a restriction usually takes place in the form
of a switching lemma.  This makes it possible to convert depth $d$
formulas to formulas of depth $d-1$.  A sequence of restrictions is
applied to reduce the depth to (essentially) zero making the circuit
or formula straightforward to analyze.  The balance to be struck is to
find a set of restrictions that leave a large resulting contradiction
but at the same time allows a switching lemma to be proved with good
parameters.

In proof complexity the most commonly studied measure is the total
size of a proof.  There are two components to this size, the number of
reasoning steps needed and the size of each line of the proof.  In
some cases, such as resolution, each line is automatically bounded in
size and hence any lower bound for proof size is closely related to
the number of proof steps.  In some other situation the line sizes
may grow and an interesting question is whether this can be avoided.

This line of investigation for Frege proofs with bounded depth
formulas was recently initiated by Pitassi, Ramakrishnan, and
Tan~\cite{PRT21}. They consider the Tseitin contradiction defined over
the grid of size $n \times n$, a setting where strong total size lower
bounds for Frege refutations of bounded depth had previously been
given by Håstad \cite{jhtseitin}. If each line of the refutation is
limited to size $M$ and depth $d$, then Pitassi et al.~\cite{PRT21}
showed that the Frege proof must consist of at least
$\exp(n/ 2^{O ( d \sqrt {\log M})})$ many lines. For most interesting
values of $M$ this greatly improves the bounds implied by the results
for total proof size. In particular if $M$ is a polynomial the lower
bounds are of the form $\exp(n^{1-o(1)})$, as long as
$d=o (\sqrt {\log n})$, in contrast to the total size lower bounds of
the form $\exp(n^{\Omega(1/d)})$. Pitassi et al.~\cite{PRT21} rely on
the restrictions introduced by Håstad~\cite{jhtseitin} but analyze
them using the methods of Pitassi et al.~\cite{pitassi16frege}.

We study the same Tseitin contradiction on the grid and improve the
lower bounds to $\exp(n/ (\log M)^{O(d)})$, a bound conjectured by
Pitassi et al.~\cite{PRT21}. Note that if the size of each line is
bounded by~$M = O\bigl(n^{\polylog(n)}\bigr)$ and the depth
is~$d = o\bigl(\frac{\log n}{\log \log n}\bigr)$, then the length
lower bound is of the form~$\exp\bigl(n^{1-o(1)}\bigr)$. For this
setting of parameters the bound is essentially optimal since there is
a resolution upper bound of size $2^{O(n)}$.

For other settings of parameters we cannot match the lower bound. We
do believe, though, that the bound obtained is tight for a wide range
of parameters. While we cannot match the lower bounds with actual
proofs we can at least represent the intermediate results of a natural
proof by formulas of the appropriate size. We discuss this in more
detail below.

\subsection{Overview of proof techniques}

The structure of the proof of our main result follows the approach of
\cite{PRT21} but relies on proving much sharper variants of the
switching lemma.

In a standard application of a switching lemma to proof complexity one
picks a restriction and demands that switching happens to all depth
two formulas in the entire proof.  Each formula switches successfully
with high probability and by an application of a union bound it is
possible to find a restriction to get them all to switch
simultaneously.

The key idea of \cite{PRT21} is that one need not consider all
formulas in the proof at the same time.  Rather one can focus on the
sub-formulas of a given line.  It is sufficient to establish that
these admit what is called an $\ell$-common partial decision tree of
small depth.  This is a decision tree with the property that at each
leaf, each of the formulas can be described by a decision tree of
depth $\ell$.  It turns out that this is enough to analyze the proof
and establish that a short proof cannot derive contradiction. The key
property is that it is sufficient to only look at the constant number
of formulas involved in each derivation step and analyze each such
step separately.

The possibility to compute a set of formulas by an $\ell$-common
partial decision tree after having been hit by a restriction is
exactly what is analyzed by what has become known as a
``multi-switching lemma'' as introduced by \cite{jhmultiswitch,imp}.
This concept was introduced in order to analyze the correlation of
small circuits of bounded depth with parity but turns out to also be
very useful in the current context.

Even though there is no general method, it seems like
when it is possible to prove a standard switching
lemma there is good hope to also prove
a multi-switching lemma with similar parameters.  This happens when
going from \cite{jhswitch} to \cite{jhmultiswitch}
and when going from \cite{pitassi16frege} to \cite{PRT21}.
We follow the same approach here and this paper
very much builds on \cite{jhtseitin}.
We need a slight modification of the space of
restrictions and changes to some steps of the proof, but a large fraction
of the proof remains untouched.  Let us briefly touch on the necessary 
changes.

The switching lemma of Håstad \cite{jhtseitin} has a failure
probability to not switch to a decision tree of depth $s$ of the form
$(As)^{\Omega(s)}$ where $A$ depends on other parameters.  As a first
step one needs to eliminate the factor $s$ in the base of the
exponent.  This triggers the above mentioned change in the space of
restrictions.  This change enables us to prove a standard switching
lemma with stronger parameters and, as a warm-up, we give this proof
in the current paper.  This results in an improvement of the lower
bound for total proof size from $\exp( \tilde \Omega (n^{1/58d}))$ to
$\exp( \tilde \Omega (n^{1/d}))$. We believe that this lower bound is
tight up to poly-logarithmic factors in the exponent.

The high level idea of the proof of the multi-switching lemma is that
for each of the formulas analyzed we try to construct a decision tree
of depth $\ell$.  If this fails then we take the long branch in the
resulting decision tree and instead query these variables in the
common decision tree.  

\subsection{Constructing small proofs}

Let us finally comment on a possible upper bound; how to construct
efficient refutations.  If we are allowed to reason with linear
equations modulo two then the Tseitin contradiction has efficient 
refutations.
In particular on the grid we can sum all equations in a single
column giving an equation containing $O(n)$ variables that
must be satisfied.  Adding the corresponding equation for the
adjacent column maintains an equation of the same size and 
we can keep adding equations from adjacent columns until we have
covered the entire grid.  We derive a contradiction and
we never use an equation containing more than $O(n)$
variables.

If we consider resolution then it is possible to represent a
parity of size $m$ as a set of clauses.  Indeed, 
looking at the equation $\sum_{i=1}^m x_i=0$ we can 
replace this by the $2^{m-1}$ clauses of full width where
an odd number of variables appear in negative form.
Now replace each parity in the above proof by its corresponding
clauses.  It is not difficult to check that Gaussian elimination
can be simulated by resolution.  Given linear equation 
$L_1=b_1$ and $L_2=b_2$ with $m_1$, and $m_2$ variables
respectively, and both containing the variable $x$
we want to derive all clauses representing $L_1 \oplus L_2 = b_1 \oplus b_2$.
We have $2^{m_1-1}$ clauses representing the first linear
equation and the $2^{m_2-1}$ clauses representing the second linear
equation.  Now we can take each pair of clauses and resolve
over $x$ and this produces a good set of clauses. If $L_1$
and $L_2$ do not have any other common variables
we are done.  If they do contain more common variables then
additional resolution steps are needed but these are not difficult to
find and we leave it to the reader to
figure out this detail.  We conclude that Tseitin 
on the grid allows resolution proofs of length $2^{O(n)}$.

Let us consider proofs that contain formulas of depth $d$
and let us see how to represent a parity.  Given $\sum_{i=1}^mx_i=0$
we can divide the variables in to 
groups of size $(\log M)^{d-1}$ and
write down formulas of depth $d$ and size $M$ that
represent the parity and the negation of the parity of each group.
Assume that the output gate of each of these formulas is an or.
We now use the above clause representation of the parity of
the groups and get a set of $2^{m/(\log M)^{d-1}}$ formulas
of size $mM/(\log M)^{d-1}$ that represent the linear
equations  This means that we can represent each line
in the parity proof by about $2^{n/(\log M)^{d-1}}$ lines
of size about $M$.  We do not know how to syntactically translate
a Gaussian elimination step to some proof steps in this representation and thus
we do not actually get a proof, only a representation of
the partial results.

\subsection{Organization}

Let us outline the contents of this paper.  We start in
\cref{sec:prelim} with some preliminaries. In \cref{sec:restrictions}
we define the set of restrictions used in the current paper which are
almost the same as in \cite{jhtseitin}. Next we show how to derive our
two main theorems assuming the new switching lemmas in
\cref{sec:main}. In \cref{sec:singleswitch-overview} we provide some
further preliminaries and explain the proof idea of the standard
switching lemma. The full proof of the standard switching lemma is
given in \cref{sec:singleswitch} and the extension to a
multi-switching lemma is presented in \cref{sec:multi-switch}. We end
with some conclusions in \cref{sec:conclusion}.

\section{Preliminaries}\label{sec:prelim}

Logarithms are denoted by $\log$ and are always with respect to the
base $2$. For integers $n \ge 1$ we introduce the shorthand
$[n] = \set{1, \ldots, n}$ and sometimes identify singletons $\set{u}$
with the element $u$. We identify \emph{false} (\emph{true}) with $0$
(with $1$) and let the binary \emph{or} and \emph{and} connective be
denoted by~$\lor$ and~$\land$ while the unary \emph{negation}
connective is denoted by~$\lnot$.

Since Frege systems over the basis $\lor$, $\land$ and $\lnot$ can
polynomially simulate each other \cite{cook79efficiency} it is not
essential what Frege system we use. We choose to work with
Schoenfield's system as previous work
has~\cite{UF96SimplifiedLowerBounds,pitassi16frege, jhtseitin, PRT21}.

\subsection{The Frege Proof System}
\label{sec:frege}

Schoenfield's Frege system works over the basis $\lor$ and $\lnot$. We
simulate a conjunction $A \land B$ by treating it as an abbreviation
for the formula $\lnot(\lnot A \lor \lnot B)$.

If $A$ is a formula over variables $p_1, \ldots, p_m$, and $\sigma$
maps the variables $p_1, \ldots, p_m$ to formulas $B_1, \ldots, B_m$,
then $\sigma(A)$ is the formula obtained from $A$ by replacing the
variable $p_i$ with $B_i = \sigma(p_i)$ for all $i$.
A \emph{rule} is a sequence of formulas written as
$A_1, \ldots, A_{k} \vdash A_0$.
If every truth assignment satisfying all of $A_1, \ldots, A_{k}$ also
satisfies $A_0$, then the rule is \emph{sound}.
A formula $C_0$ is inferred from $C_1, \ldots, C_k$ by the rule
$A_1, \ldots, A_k \vdash A_0$ if there is a function $\sigma$ mapping
the variables $p_1, \ldots, p_m$, over which $A_0, \ldots, A_k$ are
defined, to formulas $B_1, \ldots, B_m$ such that $C_i = \sigma(A_i)$
for all $i$.

The Frege system we consider consists of the rules
\begin{align*}
  &\vdash p \lor \lnot p &&\text{Excluded Middle,}\\
  p &\vdash q \lor p &&\text{Expansion rule,}\\
  p \lor p &\vdash p &&\text{Contraction rule,}\\
  p \lor (q \lor r) &\vdash (p \lor q) \lor r
  && \text{Association rule,}\\
  p \lor q, \lnot p \lor r &\vdash q \lor r
  && \text{Cut rule.}
\end{align*}

A \emph{Frege proof} of a formula $B$ from a formula
$A = C_1 \land \ldots \land C_m$ is a sequence of formulas
$F_1, F_2, \ldots, F_\ell$ such that $F_\ell = B$ and every formula
$F_i$ in the sequence is either one of $C_1, \ldots, C_m$ or inferred
from formulas earlier in the sequence by one of the above rules. Since
the above system is sound and complete a formula $B$ has a proof from
a formula $A$ if and only if $B$ is implied by $A$. A \emph{Frege
  refutation} of a formula $A$ is a Frege proof of $\bot$ constant
false.

The \emph{size of a formula} is the number of connectives in it and
the \emph{depth} of a formula $A$ is the maximum number of
alternations of $\lor$ and $\neg$ on any root-to-leaf path when $A$ is
viewed as a tree. The \emph{size of a Frege proof} is the sum of the
sizes of all formulas in the proof and the \emph{depth of a proof} is
the maximum depth of any formula in it.

\subsection{The Grid}
Throughout the paper we work over graphs~$G_n = (V,E)$ with~$n^2$
nodes which we call \emph{the grid}. However, in order to avoid
problems at the boundary, we in fact work over the \emph{2-dimensional
  torus}: each node $(i,j) \in V$ is indexed by two integers
$i,j \in [n]$ and an edge $\set{u,v}$ is in $E$ if and only if it
connects two adjacent nodes, that is, if one of the coordinates of $u$
and $v$ are identical and the other differs by 1 modulo $n$.

For a set~$U \subseteq V$ we say that a node~$v$ is at distance~$d$
from~$U$ if there is a node~$u \in U$ such that the shortest path
between~$u$ and~$v$ is of length~$d$.

\subsection{Tseitin Formulas}
The Tseitin formula $\tseitin(G, \alpha)$ defined for a graph $G$ and
a vector $\alpha \in \set{0,1}^{V(G)}$ claims that there is a
$\set{0,1}$-labeling of the edges of $G$ such that the number of
$1$-labeled edges incident to each node~$v$ is equal to the
\emph{charge} $\alpha_v$ modulo $2$. This is formalized with a Boolean
variable $x_e$ per edge $e\in E(G)$ and encoding the linear
constraints
\begin{align}
\sum_{e \ni v } x_e = \alpha_v \mod 2
\end{align}
for each node $v \in V(G)$ as a CNF formula. The main case we consider
is when $\alpha_v = 1$ for all nodes~$v$. Let us denote this formula
by $\tseitin(G)$. We use more general charges in intermediate steps
and hence the following lemma from \cite{jhtseitin} is useful. In
order to be self-contained we provide a proof in
\cref{sec:omitted-proofs}.

\begin{restatable}{lemma}{evenok}
  \label{lemma:evenok}
  Consider the Tseitin formula $\tseitin(G_n, \alpha)$ defined over
  the $n \times n$ grid. If $\sum_v \alpha_v$ is even, then
  $\tseitin(G, \alpha)$ is satisfiable and has $2^{r_n}$ solutions for
  a positive integer $r_n$ that only depends on $n$ and not on
  $\alpha$.
\end{restatable}

As a converse to the above lemma, if $\sum_v \alpha_v$ is odd, then by
summing all equations it is easy to see that such a system is
contradictory. In particular the Tseitin formulas with $\alpha_v = 1$
for all $v$ are contradictions for graphs with an odd number of
nodes. We note that all Tseitin formulas $\tseitin(G_n, \alpha)$ over
the grid graph can be written as a 4-CNF formula with 8 clauses of
width 4 for each node.

\subsection{Local Consistency of Assignments}
\label{sec:grid}

We are interested in solutions to subsystems of the Tseitin
formula~$\tseitin(G_n)$. From \cref{lemma:evenok} it follows that if
we drop the constraints of a single node, then we obtain a consistent
system of linear equations with many solutions. Denote by $X$ the
variables that $\tseitin(G_n)$ is defined over and say that a partial
assignment $\alpha\colon X \rightarrow \set{0,1,*}$ assigns a variable
$x \in X$ if $\alpha(x) \in \set{0,1}$.

The \emph{support} of a partial assignment $\alpha$, denoted by
$\supp(\alpha)$, is the set of nodes incident to assigned
variables. We say that $\alpha$ is \emph{complete} on a set of nodes
$U \subseteq V(G_n)$ if $\alpha$ assigns all variables incident to $U$
and no others.  Note that the support of an assignment complete on $U$
also includes the neighbors of $U$.

We consider partial assignments that give values to few
variables. More specifically we are interested in assignments that are
complete on a set~$U \subseteq V(G_n)$ of size at most
$|U| \leq 2n/3$. Note that such a $U$ cannot touch all rows or columns
of the grid. Denote by $U^c = V(G_n) \setminus U$ the complement of
$U$.

Since $|U| \leq 2n/3$, the sub-graph $G_n[U^c]$ induced by $U^c$ has a
\emph{giant component} that contains almost all nodes of the grid:
there are at least $n/3$ complete rows and columns in $U^c$ and all
the nodes of these rows and columns are connected. It is important to
control assignments on the other, \emph{small components}, of
$G_n[U^c]$ as they may fail to extend in a consistent manner to these
small components. For a set $U$ let the \emph{closure of $U$}, denoted
by $\closure(U) \subseteq V(G_n)$, consist of all nodes in $U$ along
with all the nodes in the small components of~$G_n[U^c]$. Note that
$\closure(U)^c$ contains exactly the set of nodes that are in the
giant component of $G_n[U^c]$.

\begin{definition}[local consistency]
  \label{def:consistent}
  A partial assignment $\alpha$ with $U=\supp(\alpha)$ is
  \emph{locally consistent} if it can be extended to an assignment
  complete on $\closure(U)$ such that all parity constraints on
  $\closure(U)$ are satisfied. We extend this notion to say that a
  pair of assignments is \emph{pairwise locally consistent} if they do
  not give different values to the same variable and the union of the
  two assignments is locally consistent.
\end{definition}

The following lemma from \cite{jhtseitin} is used throughout the
article.

\begin{restatable}{lemma}{extendone}
  \label{lemma:extendone}
  If~$\alpha$ is a locally consistent assignment
  satisfying~$\lvert\supp(\alpha)\rvert \leq n/2$, then for any
  variable~$x_e$ there is a locally consistent
  assignment~$\alpha' \supseteq \alpha$ with~$x_e$ in its domain.
\end{restatable}

For completeness we provide a proof in \cref{sec:omitted-proofs}.

\begin{definition}[local implication]
  \label{def:forced-consistency}
  Let $\alpha$ be a locally consistent assignment. A variable $x$ is
  \emph{locally implied} by $\alpha$ if there is a unique $b \in \set{0,1}$
  such that the partial assignment $\alpha \cup \set{x \mapsto b}$ is
  locally consistent.
\end{definition}

In particular if a locally consistent assignment $\alpha$ assigns a
variable $x$, then $x$ is locally implied by $\alpha$.

\subsection{Restrictions}
\label{sec:prel-restrictions}

Let $\tau\colon \set{x_1, \ldots, x_m} \rightarrow \set{0,1,*}$ be a
partial assignment and denote by $F$ a formula over the variables in
the domain of $\tau$. The formula \emph{$F$ restricted by $\tau$},
denoted by $\restrict{F}{\tau}$, is the formula obtained from $F$ by
replacing each variable~$x_i$ by $\tau(x_i)$ unless $\tau(x_i) = *$ in
which case we leave~$x_i$ untouched.
More generally, if $\sigma$ maps variables $x_1, \ldots, x_m$ to
formulas $A_1, \ldots, A_m$, then the formula \emph{$F$ restricted by
  $\sigma$}, denoted by $\restrict{F}{\sigma}$, is the formula
obtained from $F$ by replacing the variable $x_i$ with
$A_i = \sigma(x_i)$ for all $i \in [m]$.

\subsection{Decision Trees}
\label{sec:simple-dec-tree}

A \emph{decision tree} is a directed tree such that every node has
either out-degree 2 (an \emph{internal node}) or 0 (a \emph{leaf
  node}) and all nodes have in-degree $1$ except the designated
\emph{root node} which has in-degree $0$. Edges and leaves are labeled
$0$ or $1$ while internal nodes are labeled with a variable. A
\emph{branch} of a decision tree $T$ is a root-to-leaf path in $T$ and
the \emph{depth} of $T$, denoted by $\depth(T)$, is the length of the
longest branch in $T$.  Throughout we implicitly assume that internal
node labels (variables) of any branch are distinct.

A \emph{1-branch} (a \emph{0-branch}) is a branch with a leaf labeled
$1$ (labeled $0$), and a \emph{1-tree} (a \emph{0-tree}) is a decision
tree where all leaves are labeled 1 (labeled 0). Special cases of
$b$-trees are trees of depth 0. We sometimes write $T = b$ if $T$ is a
$b$-tree.

Given a Boolean assignment $\tau$ we can evaluate a decision tree $T$:
start at the root node $v$. If $v$ is a leaf, then output its
label. Otherwise $v$ is labeled by some variable $x$. Let
$b = \tau(x)$ and recurse on the node that the out-edge $b$ of
$v$ points to.

Since every branch~$B$ has a minimal partial assignment $\tau$ such
that any extension of $\tau$ traverses $B$, we interchangeably
identify a branch by the root-to-leaf path $B$, by $\tau$, and by the
unique leaf in $B$.

For a partial assignment $\alpha$ and a decision tree $T$ we obtain the
decision tree $T$ restricted by $\alpha$ by iteratively removing
internal nodes $v \in T$ labeled
$x_i \in \alpha^{-1}\bigl(\set{0,1}\bigr)$ and replacing them by the
node that the out-edge $\alpha(x_i)$ of $v$ points to.

Equivalently, if we view a decision tree $T$ as a set of branches,
then $T$ restricted by $\alpha$ consists of all branches $\tau \in T$
consistent with~$\alpha$ (in the standard sense: $\tau$ and $\alpha$
do not assign a variable to opposite value). These branches fit nicely
into a tree structure once all information about $\alpha$ is removed.

Let us stress the obvious: the restriction as defined above
\emph{always} produces a valid decision tree. Throughout the
manuscript we do \emph{not} use the above notion of a restriction. In
the following we define the restrictions used.

\paragraph{Decision Trees and Local Consistency.}
In the following we consider decision trees on the variables of the
Tseitin formula $\tseitin(G_n)$ defined over the $n\times n$ grid.

\begin{definition}[local consistency for branches]
  \label{def:local-consistency-dec-trees}
  Let $T$ be a decision tree on the variables of the Tseitin formula
  $\tseitin(G_n)$. A branch~$\tau$ in $T$ is \emph{locally consistent}
  if the partial assignment $\tau$ is locally consistent as an
  assignment (see \cref{def:consistent}) and $T$ is \emph{locally
    consistent} if all branches $\tau$ of $T$ are locally consistent.
\end{definition}

The following is a direct consequence of \cref{lemma:extendone}.

\begin{corollary}
  \label{lem:local-consistent-branch}
  Let $T$ be a decision tree on the variables of the Tseitin formula
  $\tseitin(G_n)$. If $\depth(T) \leq n/4$, then $T$ contains a
  locally consistent branch.
\end{corollary}

Throughout this article we assume that all considered decision trees
are of depth at most one fourth of the dimension of the grid we are
considering. We may thus assume that all decision trees have a locally
consistent branch.

We are going to maintain the even stronger property that all the
considered decision trees~$T$ are locally consistent, that is,
\emph{all branches} of $T$ are locally consistent. This is easy to
maintain during the creation of a decision tree: when extending a
decision tree at some leaf $\tau$ we simply disallow queries to a
variable $x$ if $x$ is locally implied by $\tau$.

We further intend to maintain this property when a decision tree $T$
is hit by a locally consistent assignment~$\alpha$. To this end we
trim $T$ aggressively during the restriction: denote by
$\restrict{T}{\alpha}$ the decision tree that consists of all the
branches $\tau \in T$ that are \emph{pairwise locally consistent} with
$\alpha$ as defined in \cref{def:consistent}. If there are indeed some
such branches $\tau$, then these fit again into a tree-structure once
any information about variables $x$ locally implied by $\alpha$ is
removed.

However, if there are \emph{no} branches $\tau$ pairwise locally
consistent with $\alpha$, then the above restriction fails to return a
decision tree. The following lemma, a direct consequence of
\cref{lemma:extendone}, states that if $\alpha$ is small and the tree
$T$ is of low depth, then the restriction $\restrict{T}{\alpha}$ does
not fail, that is, the restricted tree $\restrict{T}{\alpha}$ is
indeed a locally consistent decision tree.

\begin{corollary}
  \label{lem:local-consistent-branch-restricted}
  Let $T$ be a decision tree on the variables of the Tseitin formula
  $\tseitin(G_n)$ and denote by $\alpha$ a locally consistent
  assignment. If
  $\lvert\supp(\alpha)\rvert + 2\cdot\depth(T) \leq n/2$, then
  $\restrict{T}{\alpha}$ is a locally consistent decision tree.
\end{corollary}

\begin{definition}[functional equivalence of decision trees]
  \label{def:local-consistent-dec-tree}
  Denote by $T_1$ and $T_2$ two locally consistent decision trees of
  depth at most $n/8$ defined over the variables of the Tseitin
  formula $\tseitin(G_n)$. The decision trees $T_1$ and $T_2$ are
  \emph{functionally equivalent} if for every branch $\tau$ of $T_1$
  ending in a leaf labeled $b$ it holds that
  $\restrict{T_2}{\tau} = b$ and vice-versa.
\end{definition}

The assumption in \cref{def:local-consistent-dec-tree} on the depth of
$T_1$ and $T_2$ ensure that the restricted trees are well-defined.

\begin{lemma}
  \label{lem:loc-cons-b-tree}
  Let $T_1$ and $T_2$ be two locally consistent decision trees defined
  over the variables of the Tseitin formula $\tseitin(G_n)$ and denote
  by $\alpha$ a locally consistent assignment. Suppose
  $\restrict{T_1}{\alpha}$ is a $b$-tree. If $T_1$ and $T_2$ are
  functionally equivalent and
  $\lvert\supp(\alpha)\rvert + 2\bigl(\depth(T_1) +\depth(T_2)\bigr)
  \leq n/2$, then $\restrict{T_2}{\alpha}$ is a $b$-tree.
\end{lemma}
\begin{proof}
  Suppose that $\restrict{T_2}{\alpha}$ has a $\lnot b$-branch
  $\tau_2$. Since $T_1$ and $T_2$ are functionally equivalent it holds
  that $\restrict{T_1}{\alpha \cup \tau_2}$ is a $\lnot b$-tree, where
  we use~\cref{lem:local-consistent-branch-restricted}. This
  contradicts the assumption that $\restrict{T_1}{\alpha}$ is a
  $b$-tree.
\end{proof}

For decision trees~$T, T_1, T_2, \ldots, T_m$ we say that $T$
\emph{represents} $\bigvee_{i=1}^{m}T_i$ if for every branch $\tau$ of
$T$ ending in a leaf labeled $1$ it holds that there is an $i \in [m]$
such that $\restrict{T_i}{\tau} = 1$, and if $\tau$ ends in a leaf
labeled $0$, then for all $i \in [m]$ it holds that
$\restrict{T_i}{\tau} = 0$.

Recall that the key idea of \cite{PRT21} is that one need not consider
all formulas in the proof at the same time. Rather one can focus on
the sub-formulas of a given line and establish that these admit what
is called an $\ell$-common partial decision tree of small depth. This
is a decision tree with the property that at each leaf each of the
formulas can be described by a decision tree of depth $\ell$. The
formal definition follows.

\begin{definition}[common partial decision tree]
  Let $T_1, \ldots, T_m$ be decision trees over the variables of the
  Tseitin formula $\tseitin(G_n)$. A decision tree $\calT$ is said to
  be an \emph{$\ell$-common partial decision tree for
    $T_1, \ldots, T_m$ of depth $t$} if
  \begin{enumerate}
  \item the depth of~$\calT$ is bounded by $t$, and
  \item for every $T_i$ and branch $\tau \in \calT$ there are decision
    trees $T(i, \tau)$ of depth $\ell$ satisfying the following. Let
    $\calT_i$ be the decision tree obtained from $\calT$ by appending
    the trees $T(i, \tau)$ at the corresponding leaf $\tau$ of
    $\calT$. Then, if a branch $\tau' \in \calT_i$ ends in a leaf
    labeled $b$, it holds that $\restrict{T_i}{\tau'} = b$.
  \end{enumerate}
\end{definition}

Let $m_1, \ldots, m_M \in \N^+$ for some integer $M$. Consider
decision trees $T_i^j$ for $i \in [m_j]$ and $j \in [M]$. For
$j \in [M]$ let $T_j$ be a decision tree that represent
$\bigvee_{i=1}^{m_j}T^j_i$. An $\ell$-common partial decision tree
$\calT$ of depth $t$ \emph{represents} the sequence
$\big(\bigvee_{i=1}^{m_j}T^j_i\big)_{j=1}^M$ if it is an $\ell$-common
partial decision tree for $T^1, \ldots, T^M$ of depth $t$.

\subsection{Evaluations}
\label{sec:teval}

The concept of an \emph{evaluation} was introduced by Krajíček et
al.~\cite{KPW95FregePHPExp} and is a very convenient tool for proving
lower bounds on Frege proof size. The content of this section is
standard and we follow the presentation of Urquhart and
Fu~\cite{UF96SimplifiedLowerBounds} while using the notation of
Håstad~\cite{jhtseitin}. We need a generalization of previous notions
as introduced by Pitassi et al.~\cite{PRT21}.
\begin{definition}[evaluation]
  \label{def:evaluation}
  The set of formulas $\Gamma$ has a \emph{$t$-evaluation $\varphi$},
  mapping formulas from $\Gamma$ to locally consistent decision trees
  of depth at most $t$ if the following holds.
  \begin{enumerate}
  
  \item The mapping $\varphi$ assigns constants (variables) to the
    corresponding decision trees of depth $0$ (of depth $1$),\label[prop]{prop:depth-0}

  \item axioms are assigned to $1$-trees,\label[prop]{prop:1-tree}
  
  \item if $\varphi(F) = T$, then $\varphi(\lnot F)$ is the same
    decision tree as $T$ except that the leaf-labels are negated,
    and\label[prop]{prop:not}
  
  \item if $F = \bigvee_{i \in [s]} F_i$, then $\varphi(F)$ represents
    $\bigvee_{i \in [s]} \varphi(F_i)$.\label[prop]{prop:or}
  \end{enumerate}
\end{definition}

Eventually we will associate each line of a Frege proof with its own
$t$-evaluation. In order to argue about the proof we require that
these different $t$-evaluations are functionally equivalent as
explained in the following. Let us say that two formulas are
\emph{isomorphic} if they only differ in the order of the binary
$\vee$ connectives.

\begin{definition}[functional equivalence of evaluations]
  \label{def:consistency-evaluation}
  Consider a $t$-evaluation $\varphi$ defined over a set of formulas
  $\Gamma$ and similarly let $\varphi'$ be a $t$-evaluation defined
  over the set of formulas $\Gamma'$. The two $t$-evaluations
  $\varphi$ and $\varphi'$ are \emph{functionally equivalent} if all
  isomorphic formulas $F \in \Gamma$ and $F' \in \Gamma'$ satisfy that
  the decision trees $\varphi(F)$ and $\varphi'(F')$ are functionally
  equivalent.\end{definition}

We say that a \emph{Frege proof has a $t$-evaluation} if each line
$\nu$ in the proof has a $t$-evaluation $\varphi^\nu$ for all
sub-formulas occurring on $\nu$ and for all lines $\nu, \nu'$ it holds
that $\varphi^\nu$ and $\varphi^{\nu'}$ are functionally equivalent.

The following lemma is central. It states that if we have a
$t(k)$-evaluation for a Frege proof with $t(k) \leq n/16$, then all
lines in the proof are represented by $1$-trees. As constant false is
represented by a $0$-tree (\cref{def:evaluation}, \cref{prop:depth-0})
it is thus not possible to derive contradiction.  Hence any Frege
refutation is large, respectively long in the case of Frege proofs of
bounded line size.

\begin{restatable}{lemma}{noproof}\label{lem:noproof}
  Let $n,t \in \N$ such that $t \leq n/16$ and suppose that we have a
  Frege proof of a formula~$A$ from the Tseitin
  formula~$\tseitin(G_n)$ defined over the $n \times n$ grid. If this
  proof has a $t$-evaluation, then each line in the derivation is
  mapped to a 1-tree. In particular $A \neq \bot$, that is,
  contradiction cannot be derived.
\end{restatable}

The proof of \cref{lem:noproof} follows by standard arguments. For
completeness we provide a proof in \cref{sec:omitted-proofs}.

\section{Full Restrictions}
\label{sec:restrictions}
\label{sec:full-restriction}

In this section we introduce a space of restrictions that we use to
turn the Tseitin contradiction $\tseitin(G_{n_1})$ defined over the
$n_1 \times n_1$ grid into the Tseitin contradiction
$\tseitin(G_{n_2})$ over the smaller $n_2 \times n_2$ grid. Throughout
we assume that $n_1$ and $n_2$ are odd integers such that the
mentioned Tseitin formulas are indeed contradictions.

Let $\factorN = \lfloor n_1/n_2\rfloor$ and partition the columns (rows) of
the $n_1 \times n_1$ grid into $n_2$ almost equal sized sets
$\calQ = \set{Q_i \mid i \in [n_2]}$ (sets
$\calR = \set{R_i \mid i \in [n_2]}$) such that each set~$Q_i$
(set~$R_i$) contains $\factorN$ or $\factorN+1$ consecutive columns (rows). The
\emph{central columns} of $Q \in \calQ$ consist of columns $q \in Q$
such that $Q$ has at least $\factorN/8$ columns to the left and right of $q$
and similarly we let the \emph{central rows} of $R \in \calR$ consist
of rows $r \in R$ such that $R$ has at least $\factorN/8$ rows up and down of
$r$. For each set $Q \in \calQ$ (set $R \in \calR$) we designate
$\Delta = \lfloor \factorN/5\rfloor$ of the central columns in $Q$ (central
rows in $R$) to be the \emph{center columns} (\emph{center
  rows}). These center columns (center rows) are chosen evenly spaced
from the central columns (central rows) and hence there are always at
least $2$ central columns (central rows) between each consecutive pair
of center columns (center rows).

\begin{figure}
  \centering
  \includegraphics{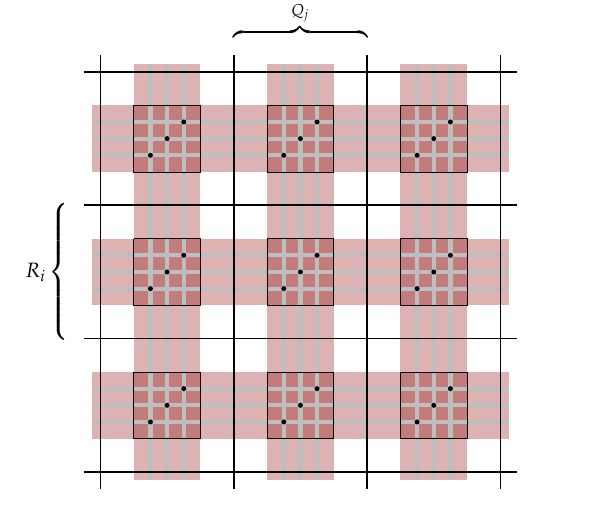}
  \caption{Centers and central areas: central columns and central rows
    are highlighted red while center columns and center rows are
    shaded gray}
  \label{fig:grid}
\end{figure}

The partitions $\calQ$ and $\calR$ naturally induce a partition of the
$n_1 \times n_1$ grid into $n_2^2$ \emph{sub-squares}:
sub-square~$(i,j)$ is defined as the sub-graph induced by the nodes in
$R_i \cap Q_j$.  The \emph{central area} of a sub-square $(i,j)$
consists of the nodes in the intersection of the central rows of $R_i$
and the central columns of $Q_j$. The $\ell$th \emph{center} of a
sub-square is the node in the intersection of the $\ell$th center row
of $R_i$ and the $\ell$th center column of $Q_j$. Each sub-square has
hence $\Delta$ centers. A schematic picture is given in
\cref{fig:grid}.

A restriction chooses one center per sub-square which, eventually,
will be the nodes of the smaller $n_2 \times n_2$ grid. For this to
make sense we need to explain (1) how to connect such centers by paths
and (2) how these paths correspond to variables in the smaller
instance.

Let us specify the paths used to connect centers in adjacent
sub-squares. Suppose we are given a center $c_i$ and a center $c_j$ in
the sub-square below. Since there are at least
$2 \cdot \lceil \factorN/8 \rceil \geq \Delta$ rows between the two
central areas we can designate for each center $c_i$ a unique row
$\row_i$ in the middle area.

To connect $c_i$ to $c_j$ we start at $c_i$, first go $1$ step to the
left and then straight down to~$\row_i$. This is complemented by
starting at~$c_j$, going $1$ step to the right, and then straight up
to~$\row_i$. The appropriate segment from~$\row_i$ completes the
path. An illustration is provided in \cref{fig:path}.

Connecting $c_i$ to a center $c_j$ in a sub-square to the right is
done in an analogous way: there is a unique column $\col_i$ associated
with the center $c_i$. The path consists of five non-empty
segments. The first segment consists of the vertical edge down from
$c_i$ while the last segment consists of the vertical edge up from
$c_j$. We add two horizontal segments connecting the first and last
segment to the designated column $\col_{i}$ and use the appropriate
middle segment from $\col_i$.

\begin{figure}
  \centering
  \includegraphics{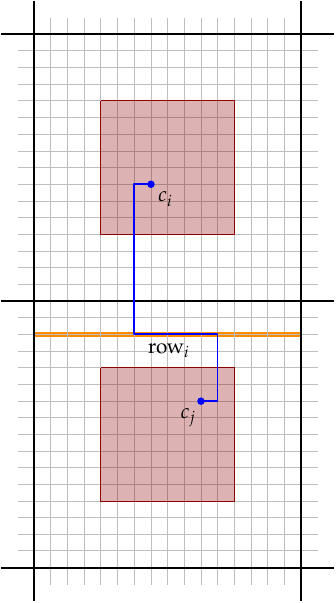}
  \caption{A path connecting $c_i$ to a center $c_j$ in the sub-square
    below with the central areas highlighted red and the designated
    row~$\row_i$ highlighted orange}
  \label{fig:path}
\end{figure}

This completes the discussion of the paths. Let us turn our attention
to how a path corresponds to a variable in the smaller Tseitin
instance. Recall that $n_1 > n_2$ are odd.

A restriction
$\sigma \in \Sigma(n_1,n_2) = \Sigma_{\calQ, \calR}(n_1,n_2)$ is
defined by one center in each sub-square (the so-called \emph{chosen
  centers} $\chosen$ of $\sigma$) and an assignment $\sigma_0$ to the
edges of the $n_1 \times n_1$ grid that satisfies the Tseitin formula
with 0 charges at the chosen centers and 1 charges at all other
nodes. Since the number of chosen centers is odd, by
\cref{lemma:evenok}, such an assignment exists.

Let us call a path that connects two chosen centers a \emph{chosen
  path} and note that the set of chosen paths is pair-wise
edge-disjoint. For each chosen path $P$ we introduce a new variable
$y_P$ and define the \emph{full restriction} $\sigma$ as
\begin{align}
  \sigma(x_e) =
  \begin{cases}
    \sigma_0(x_e)
    &\text{if $e$ is not  on a chosen path,}\\
    y_P
    &\text{if $e$ is on a chosen path $P$ and $\sigma_0(x_e) = 1$,}\\
    \neg y_P
    &\text{if $e$ is on a chosen path $P$ and $\sigma_0(x_e) = 0$.}\\
  \end{cases}
  \label{eq:def-sigma}
\end{align}
The value given by $\sigma_0$ to a variable that is not on a chosen
paths is called the \emph{final value}.
See \cref{fig:sigma} for an illustration.

We claim that $\restrict{\tseitin(G_{n_1})}{\sigma}$ is the Tseitin
contradiction $\tseitin(G_{n_2})$ defined over the smaller
$n_2 \times n_2$ grid. Let us check that under $\sigma$ all the axioms
of $\tseitin(G_{n_1})$ are either mapped to true (and can be removed)
or to an axiom of the
smaller instance $\tseitin(G_{n_2})$:
\begin{itemize}
\item The axioms of a node $v$ not on a chosen path are satisfied
  since $\sigma_0$ assigns an odd number of incident edges to 1.
  
\item The axioms of an interior node $v$ of a chosen path $P$ are
  reduced to tautologies: the axioms are true independent of the value
  of $y_P$ since flipping the value of $y_P$ changes the value of two
  variables incident to $v$.
  
\item The axioms of a chosen center $v \in \chosen$ turn into axioms
  of the smaller instance. Note that the charge is still $1$: the
  restriction $\sigma_0$ assigns an even number of edges incident to
  $v$ to $0$ and hence there is an even number of negated variables
  $\lnot x_P$ incident to $v$. Thus the constraint
  $\sum_{e \ni v} \sigma(x_e) = 1 \mod 2$ is equivalent to
  $\sum_{P \ni v} y_P = 1 \mod 2$.
\end{itemize}

Note that a full restriction $\sigma$ is really an \emph{affine
  restriction} in the vocabulary of Rossman et al.~\cite{rst} since
$\sigma$ not only assigns values to variables but also identifies old
variables with new variables or the negations thereof.

\begin{figure}
  \centering
  \includegraphics{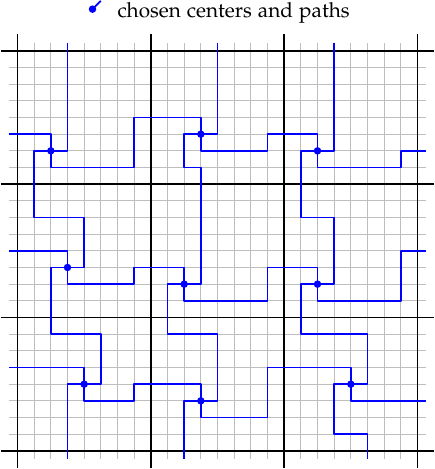}
  \caption{A full restriction $\sigma$}
  \label{fig:sigma}
\end{figure}

\subsection{A Distribution Over Full Restrictions}
\label{sec:distribution}

For odd integer~$k$
sample~$\sigma \sim \calD_k\big(\Sigma(n_1, n_2)\big)$ as
follows. Uniformly at random choose a set of~$k$ centers from the set
of all $k$-subsets of the~$\Delta n^2_2$ centers with the property
that every sub-square contains~$(1\pm 0.01)k/n^2_2$ centers. These are
the so-called \emph{alive} centers. In each sub-square the alive
center with the lowest numbered row becomes a \emph{chosen
  center}. Sample uniformly at random an assignment~$\sigma_0$ from
the space of solutions to the Tseitin formula with charges $0$ at the
chosen centers and $1$ at all other nodes. Define $\sigma$ from
$\sigma_0$ as in \cref{eq:def-sigma}.

\subsection{Differences to Previous Work}

We use a very similar space of random restriction compared to
\cite{jhtseitin} but make two changes. First, we make the number~$k$
of live centers \emph{independent} of the depth of the considered
decision trees: in \cite{jhtseitin} the number of live centers
is~$C n_2^2 s$ for~$s$ the depth of the considered decision trees
whereas we have~$Cn_2^2\log n$ live centers (independent of~$s$). This
change allows us to prove a multi-switching lemma. Second we define
full restrictions by designating rows and columns that are unique to a
single center instead of pairs of centers as done in \cite{jhtseitin}.

While the first change already appeared in the preliminary version of
this work \cite{jhkr} the second modification is here presented for
the first time. This change allows us to further strengthen the size
lower bound by a factor $2$ in the second exponent.
 
\subsection{Decision Trees and Full Restrictions}
\label{sec:dt}

In this section we define~$\restrict{T}{\sigma}$: the restriction of a
decision tree~$T$ by a full restriction~$\sigma$. Recall that for a
partial assignment~$\alpha$ the decision tree~$\restrict{T}{\alpha}$
only consists of branches pairwise locally consistent with~$\alpha$
(see \cref{sec:simple-dec-tree}). In the following we define a notion
of local consistency for full restrictions so that we can define
$\restrict{T}{\sigma}$ as for partial assignments.

While the initial decision tree~$T$ queries variables $x_e$, the
resulting decision tree $\restrict{T}{\sigma}$ queries path variables
$y_P$ defined on the smaller grid. The idea of pairwise local
consistency of a full restrictions~$\sigma \in \Sigma(n_1,n_2)$ and a
partial assignment~$\tau$ is not complicated: we want to ensure that
if a variable is assigned a constant by both $\tau$ and $\sigma$, then
these agree. Furthermore we require that the restriction induced by
$\tau$ on the smaller $n_2 \times n_2$ grid is locally consistent. The
following definition formalizes this notion.

\begin{definition}[pairwise local consistency for full restrictions]
  \label{def:local-consistency-full}
  Consider the Tseitin formula~$\tseitin(G_{n_1})$ defined over the
  $n_1 \times n_1$ grid, denote by $X$ the variables of
  $\tseitin(G_{n_1})$, let $\tau\colon X \rightarrow \set{0,1,*}$ be a
  partial assignment, and denote by $\sigma \in \Sigma(n_1,n_2)$ a
  full restriction. We say that $\tau$ and $\sigma$ are \emph{pairwise
    locally consistent} if the following holds.
  \begin{enumerate}
  \item For all variables $x \in X$ it holds that if $x$ is assigned
    to a constant by both $\tau$ and $\sigma$, then
    $\sigma(x) = \tau(x)$.\label[prop]{prop:sigma-const}
    
  \item Suppose $\sigma$ maps $x_1$ and $x_2$ to the same variable. If
    $x_1,x_2 \in \tau^{-1}\bigl(\set{0,1}\bigr)$, then
    \begin{enumerate}
    \item 
      $\tau(x_1) = \tau(x_2)$ if $\sigma(x_1) = \sigma(x_2)$,
    \item
      $\tau(x_1) = \neg \tau(x_2)$ if $\sigma(x_1) = \neg\sigma(x_2)$.
    \end{enumerate}
    \label[prop]{prop:def-consistency-beta-well-defined}
    
  \item If $\tau_\sigma$ denotes the minimal partial assignment such
    that for all $x \in \dom(\tau)$ it holds that
    \begin{align*}
      \tau_\sigma\bigl(\sigma(x)\bigr) = \tau(x)\eqcomma
    \end{align*}
    then we require that $\tau_\sigma$ is locally consistent with respect to
    the $n_2 \times n_2$ grid as defined in \cref{def:consistent}.
    \label[prop]{prop:def-consistency-beta}
  \end{enumerate}
\end{definition}

Note that \cref{prop:def-consistency-beta-well-defined} of the above
definition ensures that $\tau_\sigma$ is well defined, that is, it
ensures that there are no two edges on a chosen path $P$ such that
$\tau_\sigma$ assigns $0$ and $1$ to $y_P$.  Note that $\tau_\sigma$,
as defined in \cref{prop:def-consistency-beta} of
\cref{def:local-consistency-full}, can simply be thought of as the
restriction that $\tau$ induces on the smaller grid. In case a more
explicit description of $\tau_\sigma$ is sought: the partial
assignment $\tau_\sigma$ may equivalently be defined as
\begin{align*}
  \tau_\sigma(y_P)
  &=
  \begin{cases}
    \tau(x_e)
    &\text{if there is an edge $e \in P$ such that
      $x_e \in \tau^{-1}\bigl(\set{0,1}\bigr)$
      and $\sigma(x_e) = y_P$,}\\
    \lnot \tau(x_e)
    &\text{if there is an edge $e \in P$ such that
      $x_e \in \tau^{-1}\bigl(\set{0,1}\bigr)$
      and $\sigma(x_e) = \lnot y_P$,}\\
    * &\text{otherwise.}
  \end{cases}
\end{align*}

With the notion of pairwise local consistency for full
restrictions~$\sigma$ in place we are ready to define the restriction
of a decision tree~$T$ by a full restriction~$\sigma$, denoted by
$\restrict{T}{\sigma}$, as follows. Consider a branch $\tau \in T$
that is pairwise locally consistent with $\sigma$. The restricted
decision tree~$\restrict{T}{\sigma}$ contains the branch $\tau_\sigma$
as defined in \cref{prop:def-consistency-beta} of
\cref{def:local-consistency-full} for each such $\tau$. By definition
each such branch $\tau_\sigma$ is locally consistent with respect to
the smaller grid. Furthermore, these branches fit into a
tree-structure once duplicate queries to the same variable are
removed.

However, if there is \emph{no} branch $\tau\in T$ pairwise locally
consistent with $\sigma$, then the above restriction fails to return a
decision tree. For shallow trees we are guaranteed to obtain locally
consistent decision trees as summarized in the following.

\begin{lemma}
  \label{lem:local-consistent-branch-restricted-sigma}
  Let $T$ be a decision tree on the variables of the Tseitin formula
  $\tseitin(G_{n_1})$ and denote by $\sigma\in \Sigma(n_1,n_2)$ a full
  restriction. If $\depth(T) = t \leq n_2/4$, then
  $\restrict{T}{\sigma}$ is a locally consistent decision tree of
  depth at most $t$.
\end{lemma}

\begin{proof}
  We need to argue that there is a branch $\tau \in T$ pairwise
  locally consistent with $\sigma$. We construct $\tau$ inductively
  starting with $v$ being the root node of $T$. Denote by $\tau^v$ the
  partial assignment from the root to $v$ and let $\tau_\sigma^v$ be
  the assignment induced by $\tau^v$ on the smaller grid as defined in
  \cref{prop:def-consistency-beta} of
  \cref{def:local-consistency-full}. Suppose $v$ is an internal node
  labeled $x_e$.  If $\sigma$ assigns $x_e$ to a constant, then
  recurse on the node the out-edge of $v$ labeled $\sigma(x_e)$ points
  to. Otherwise $\sigma$ assigns $x_e$ to a variable.
  \begin{enumerate}
  \item If $\tau^v_\sigma$ assigns the variable $y_P$ that
    $\sigma(x_e)$ maps to, then recurse on the node the out-edge of
    $v$ labeled
    \begin{enumerate}
    \item $\tau^v_\sigma(y_P)$ points to, assuming
      $\sigma(x_e) = y_P$, or
    \item $\lnot \tau^v_\sigma(y_P)$ points to, assuming
      $\sigma(x_e) = \lnot y_P$.
    \end{enumerate}
  \item Else $\tau^v_\sigma$ does not assign $y_P$. Let
    $b \in \set{0,1}$ such that
    $\tau^v_\sigma \cup \set{y_P\mapsto b}$ is a locally consistent
    assignment on the $n_2 \times n_2$ grid. Recurse on the node the
    out-edge of $v$ labeled $b$ points
    to.\label[step]{item:extend-tau}
  \end{enumerate}
  If the node $v$ is a leaf, then we have found a branch
  $\tau = \tau^v$ in $T$ that is pairwise consistent with~$\sigma$.

  The above process may fail in \cref{item:extend-tau} if there is no
  $b \in \set{0,1}$ such that $\tau^v_\sigma \cup \set{y_P\mapsto b}$
  is locally consistent. Since $\depth(T) \leq n_2/4$ by
  \cref{lemma:extendone} there is always such a choice. The statement
  follows.
\end{proof}

\begin{lemma}
  \label{lem:equivalence-sigma}
  Let $T$ and $T'$ be two functionally equivalent decision trees on
  the variables of the Tseitin formula $\tseitin(G_{n_1})$ and denote
  by $\sigma \in \Sigma(n_1,n_2)$ a full restriction. If
  $\depth(T),\depth(T')\leq n_2/8$, then $\restrict{T}{\sigma}$ and
  $\restrict{T'}{\sigma}$ are locally equivalent.
\end{lemma}

\begin{proof}
  Consider a $b$-branch $\tau \in T$ pairwise locally consistent with
  $\sigma$. Since $\tau$ and $\sigma$ are pairwise consistent the
  decision tree $\restrict{T}{\sigma}$ contains the $b$-branch
  $\tau_\sigma$ as defined in \cref{prop:def-consistency-beta} of
  \cref{def:local-consistency-full}.

  As $T$ and $T'$ are functionally equivalent the decision
  tree~$\restrict{T'}{\tau}$ is a $b$-tree. Since for all
  $x \in \dom(\tau)$ it holds that
  $\tau_\sigma\bigl(\sigma(x)\bigr) = \tau(x)$, it follows that the
  decision tree
  $\restrict{\big(\restrict{T'}{\sigma}\big)}{\tau_\sigma}$ is also a
  $b$-tree. Here we rely on
  \cref{lem:local-consistent-branch-restricted} to argue that the
  decision tree~$\restrict{T'}{\sigma}$ restricted by~$\tau_\sigma$ is
  well defined. The statement follows.
\end{proof}

\subsection{Evaluations and Full Restrictions}
\label{sec:kevalbasic}

Given a $t$\nobreakdash-evaluation $\varphi$ for $\Gamma$ and a full
restriction~$\sigma$, we denote by $\restrict{\varphi}{\sigma}$ the
$t$\nobreakdash-evaluation for $\restrict{\Gamma}{\sigma}$ defined by
$\restrict{\varphi}{\sigma}(\restrict{F}{\sigma}) =
\restrict{\varphi(F)}{\sigma}$ for all $F \in \Gamma$.

Consider a Frege proof of depth $d$ and for a line $\nu$ in the proof
let us denote by $\Gamma^\nu$ the set of sub-formulas occurring on
line $\nu$. We intend to construct a sequence of full restrictions
$\sigma_1, \sigma_2, \ldots, \sigma_d$ with the following
property. Denote by $\sigma_k^*$ the concatenation of the first $k$
restrictions in the sequence and let $t(k)$ be a function growing
with~$k$ to be fixed -- it will depend on the application. From the
sequence of restrictions we require that all sub-formulas occurring in
the proof of depth at most $k$ have functionally equivalent
$t(k)$-evaluations after hitting them with the restriction
$\sigma_k^*$. In more detail, for every line $\nu$ we want a
$t(k)$-evaluation for the formulas in
\begin{align}
  \Gamma^\nu_k = \set{\restrict{F}{\sigma_k^*} \mid F \in \Gamma^\nu
  \wedge \depth(F) \le k}
\end{align}
and require that any pair of these $t(k)$-evaluations is functionally
equivalent. We construct these $t(k)$-evaluations by induction on
$k$. To ensure that the domain of the $t$-evaluations does not
decrease when we apply another restriction we rely on the following
lemma.

\begin{lemma}\label{lem:phiextend}
  Let $n_1, n_2, t \in \N$ such that $n_2 \leq n_1$ and
  $t \leq n_2/8$. Denote by $\varphi$ a $t$-evaluation defined over
  the set of formulas $\Gamma$, let $\varphi'$ be a $t$-evaluation
  defined over the set of formulas $\Gamma'$, and denote by
  $\sigma \in \Sigma(n_1,n_2)$ a full restriction. If $\varphi$ and
  $\varphi'$ are functionally equivalent, then
  $\restrict{\varphi}{\sigma}$ and $\restrict{\varphi'}{\sigma}$ are
  also functionally equivalent $t$-evaluations with domains
  $\dom\bigl(\restrict{\varphi}{\sigma}\bigr) =
  \restrict{\Gamma}{\sigma}$ and
  $\dom\bigl(\restrict{\varphi'}{\sigma}\bigr) =
  \restrict{\Gamma'}{\sigma}$.
\end{lemma}

\begin{proof}
  Fix a formula $F \in \Gamma$ and let $T = \varphi(F)$. By definition
  we have that
  $\restrict{\varphi}{\sigma}(\restrict{F}{\sigma}) =
  \restrict{T}{\sigma}$.
  By \cref{lem:local-consistent-branch-restricted-sigma} we see that $\restrict{T}{\sigma}$ is a locally consistent decision
  tree with respect to the $n'\times n'$ grid and it holds that
  $\depth\bigl(\restrict{T}{\sigma}\bigr) \leq t$.
  We need to check that $\restrict{T}{\sigma}$ satisfies
  \cref{prop:depth-0,prop:1-tree,prop:not,prop:or} of
  \cref{def:evaluation}.
  \Cref{prop:depth-0,prop:1-tree,prop:not} are immediate since the
  process of a restriction neither depends on the labels of the leaves
  nor does it change them.

  We are left to show \cref{prop:or}. Suppose
  $F = \bigvee_{i\in[m]} F_i$, let $T_i = \varphi(F_i)$ and consider a
  $b$-branch $\tau$ in $T$ that is pairwise locally consistent with
  $\sigma$. Since $\varphi$ is a $t$-evaluation it holds that if
  $b=0$, then $\restrict{T_i}{\tau} = 0$ for all $i \in [m]$ and if
  $b=1$, then there is an $i \in [m]$ such that
  $\restrict{T_i}{\tau} = 1$.

  Since $\tau$ and $\sigma$ are locally consistent the decision tree
  $\restrict{T}{\sigma}$ contains the $b$-branch $\tau_\sigma$ as
  defined in \cref{prop:def-consistency-beta} of
  \cref{def:local-consistency-full}. Recall that restricting first by
  $\sigma$ and then $\tau_\sigma$ sets the variables in $\dom(\tau)$
  to the same constants as $\tau$ does, that is, for $x\in \dom(\tau)$
  we have that $\tau_\sigma\bigl(\sigma(x)\bigr) = \tau(x)$. Hence if
  $\restrict{T_i}{\tau}$ is a $b$-tree, then the decision tree
  $\restrict{T_i}{\sigma}$ under the restriction $\tau_\sigma$ is also
  a $b$-tree. For the last statement we relied on
  \cref{lem:local-consistent-branch-restricted}.  This yields
  \cref{prop:or}.
  
  The claimed functional equivalence follows from
  \cref{lem:equivalence-sigma}. This establishes the claim.
\end{proof}

The important step of the argument is to use a switching lemma to
extend the domain of the $t(k)$-evaluation from $\Gamma^\nu_k$ to
$\Gamma^\nu_{k+1}$. We give that argument in the next section.

\section{Proofs of the Main Theorems}\label{sec:main}

We first reprove the main theorem of \cite{jhtseitin} with improved
parameters.

\begin{theorem}\label{thm:main-size}
  For $d = O\bigl(\frac {\log n}{\log {\log n}}\bigr)$ it holds
  that any depth-$d$ Frege refutation of the Tseitin formula
  $\tseitin(G_n)$ with odd charges at all nodes of the $n \times n$
  grid requires size
  \begin{align*}
    \exp\Bigl(\Omega\bigl(n^{1/d}/\log^4 n\bigr)\Bigr) \eqperiod
  \end{align*}
\end{theorem}

As outlined in \cref{sec:kevalbasic} we construct a $t$-evaluation for
all sub-formulas occurring in a short and shallow Frege proof. By
\cref{lem:noproof} we then conclude that all shallow Frege proofs of
the Tseitin contradiction must be long. For the total size lower bound
we in fact do not create distinct $t$-evaluations per line but rather
a single one, used on each line. Such a $t$-evaluation is clearly
functionally equivalent and hence satisfies our needs. In order to
extend the $t$-evaluation to larger depth we use the following
switching lemma.

\begin{restatable}[Switching Lemma]{lemma}{singleSwitch}
  \label{lemma:switch}
  There are absolute constants~$A, C, n_0 > 0$ such that for
  integer~$n \geq n_0$ the following holds.
  Let~$k, m, n', s, t \in \N^+$ satisfy~$n/n' \geq A t \log^4 n$,
  $k = n'^2(1\pm 0.01)C\log n'$ be odd, and~$t \leq s \leq
  n'/32$. Then for any decision trees~$T_1, \ldots, T_m$ of depth at
  most~$t$ querying edges of the~$n \times n$ grid it holds that
  if~$\sigma \sim\calD_k\bigl(\Sigma(n,n')\bigr)$, then the
  probability that~$\bigvee_{i = 1}^m \restrict{T_i}{\sigma}$ cannot
  be represented by a decision tree of depth~$s$ is bounded by
  \begin{align*}
    \left(\frac{At\log^4 n'}{n/n'}\right)^{s/64} \eqperiod
  \end{align*}
\end{restatable}

The proof of \cref{lemma:switch} is given in
\cref{sec:singleswitch}. Let us verify that \cref{thm:main-size}
indeed follows from \cref{lemma:switch}.

\begin{proof}[Proof of \cref{thm:main-size}]
  Suppose we have a refutation of
  size~$N \leq \exp \bigl(n^{1/d}/c_1 \log^4 n\bigr)$ for some large
  constant~$c_1 > 0$. Denote by~$\Gamma$ the set of sub-formulas
  occurring in this alleged proof. We proceed by induction
  on~$i= 0, 1, 2, \ldots, d-1$ as follows.

  Assume by induction that we are given a sequence of
  restrictions~$\sigma_1, \sigma_2, \ldots, \sigma_{i-1}$, denote
  by~$\sigma_{i-1}^*$ the composition thereof, and suppose that we
  have a $t_i$-evaluation~$\varphi_{i}$ defined on all formulas
  \begin{align}
    \Gamma_i =
    \{
    \restrict{F}{\sigma_{i-1}^*}
    \mid
    F \in \Gamma \text{~and~} \depth(F) \leq i
    \}
  \end{align}
  of original depth at most~$i$. Sample a full
  restriction~$\sigma_{i} \in \calD_k\bigl(\Sigma(n_{i},
  n_{i+1})\bigr)$ and extend~$\varphi_i$ to a
  $t_{i+1}$\nobreakdash-evaluation~$\varphi_{i+1}$ defined on all
  formulas~$\Gamma_{i+1}$ of original depth at most~$i+1$.

  Let us implement the above plan. We choose~$t_0 = 1$ and~$n_0 = n$,
  set~$s= 152 \log N$, and let~$t_i = s$
  and~$n_{i} = \lfloor n_{i-1}/4At_{i-1}\log^4 n_{i-1}\rfloor$
  for~$i \in [d]$. The constant~$A$ is chosen as in
  \cref{lemma:switch}.

  Initially, for~$i=0$, each formula is either a variable which is
  mapped by~$\varphi_0$ to the depth~$1$ decision tree evaluating it
  or a constant which is mapped to the appropriate depth~$0$ decision
  tree. For the inductive step we may assume that we have a
  $t_i$-evaluation~$\varphi_i$ for all formulas in~$\Gamma_i$.

  Sample~$\sigma_{i} \in \calD_k\bigl(\Sigma(n_{i},
  n_{i+1})\bigr)$. We need to extend the~$\varphi_i$ to a
  $t_{i+1}$-evaluation~$\varphi_{i+1}$ defined on all
  formulas~$\Gamma_{i+1}$ of original depth at most~$i+1$. Consider
  any such formula~$F \in \Gamma$. We define~$\varphi_{i+1}$ as
  follows.
  \begin{enumerate}
  \item If~$F$ is of depth at most~$\depth(F) \leq i$, then we
    let~$\varphi_{i+1}(\restrict{F}{\sigma_i^*}) =
    \restrict{\varphi_i(\restrict{F}{\sigma_{i-1}^*})}{\sigma_i}$. By
    \cref{lem:phiextend} the restricted
    $t_i$-evaluation~$\restrict{\varphi_i}{\sigma_i}$ is defined on
    such formulas.

  \item If~$F = \neg F'$ is of depth~$i+1$,
    then~$\varphi_{i+1} \bigl(\restrict{F}{\sigma^*_i}\bigr)$ is
    defined in terms
    of~$\varphi_{i+1} \bigl(\restrict{F'}{\sigma^*_i}\bigr)$ by
    negating all the leaf labels.

  \item If~$F = \bigvee_j F_j$ is of depth~$i+1$ and each~$F_j$ is of
    depth at most~$\depth(F_j) \leq i$, then we appeal
    to~\cref{lemma:switch} to obtain a decision tree~$T$ of
    depth~$\depth(T) \leq t_{i+1}$ representing~$\bigvee_j T_j$
    for~$T_j =  \varphi_i(\restrict{F_j}{\sigma^*_{i-1}})$.\label{it:switch}\end{enumerate}
  The only place where the extension might fail is in
  \cref{it:switch}. By \cref{lemma:switch} and our choice of
  parameters we see that the failure probability is bounded
  by
  \begin{align}
    \left(
    \frac{At_i\log^4 n_{i+1}}
    {n_i/n_{i+1}}
    \right)^{s/76}
    \leq
    \left(
    \frac{At_i\log^4 n_{i+1}}
    {2 At_i\log^4 n_{i}}
    \right)^{2 \log N}
    \leq
    N^{-2}
    \eqperiod
  \end{align}
  We union bound over at most~$N$ sub-formulas and thus succeed with
  probability~$1-N^{-1}$ in every step.

  This completes the induction and we thus obtain a refutation of the
  formula~$\tseitin(G_{n_d})$ with a~$t_{d}$-evaluation~$\varphi_{d}$.
  Note that~$n_d \geq n/\log^{d-1}(N)(c_2\log^{4} n)^{d}$ for some
  constant~$c_2>0$.  On the other hand we have~$t_d = 152 \log N$.
  Thus if~$\log N \leq n^{1/d}/c_1 \log^4 n$ for
  constant~$c_1 > 0$ large enough, then we get a contradiction to
  \cref{lem:noproof}. The claimed lower bound follows.
\end{proof}

We turn our attention to the main result of the paper.

\begin{theorem}\label{thm:main-line-length}
  For any Frege refutation of the Tseitin formula~$\tseitin(G_n)$ with
  odd charges at all nodes of the~$n \times n$ grid the following
  holds. If each line of the refutation is of size~$M$ and depth~$d$,
  then the number of lines in the refutation is
  \begin{align*}
  \exp
  \left(
    \Omega
    \left(
      \frac
      {n}
      {\big((\log n)^{O(1)} \log M\big)^{d}}
    \right)
  \right)
  \eqperiod
  \end{align*}
\end{theorem}

Note that the lower bounds obtained from \cref{thm:main-line-length}
are much stronger than the lower bounds obtained from
\cref{thm:main-size} if the line-size~$M$ is bounded. For example,
if~$M = O\bigl(n^{\polylog(n)}\bigr)$ and the depth~
$d = o\bigl(\frac{\log n}{\log \log n}\bigr)$, then we obtain a length
lower bound of~$\exp\bigl(n^{1-o(1)}\bigr)$. This lower bound is
essentially optimal since there is a resolution refutation of
length~$2^{O(n)}$.

The strategy of the proof is similar to the proof of
\cref{thm:main-size}: we again build a $t$-evaluation for a supposed
Frege proof. The main difference is that instead of creating a single
$t$-evaluation for the entire proof we in fact independently create
$t$-evaluations for each line. These $t$-evaluations turn out to be
functionally equivalent, as defined in
\cref{def:consistency-evaluation}. We obtain the claimed bounds by an
appeal to \cref{lem:noproof}.

Suppose we are given a Frege refutation of the Tseitin principle
defined over the~$n \times n$ grid consisting of~$N$ lines, where each
line is a formula of size~$M$ and depth~$d$. Denote by~$\Gamma^\nu$
the set of sub-formulas of line~$\nu\in [N]$ in the proof. We
construct a sequence of
restrictions~$\sigma_1, \sigma_2, \ldots, \sigma_d$ such that all
formulas of depth at most~$\eta \in [d]$ have functionally
equivalent~$t(\eta)$-evaluations if hit by the concatenation
$\sigma_\eta^*$ of the first~$\eta$ restrictions in the sequence,
where~$t(\eta)$ is some function dependent on~$\eta$ to be fixed
later. That is, for every line~$\nu$ we have a~$t(\eta)$-evaluation
$\varphi_\eta^\nu$ for all formulas in the set
\begin{align}
  \Gamma^\nu_\eta = \set{\restrict{F}{\sigma_\eta^*} \mid
  F \in \Gamma^\nu \wedge \depth(F) \le \eta},
\end{align}
and all these~$t(\eta)$-evaluations are functionally equivalent. In
addition to these~$t(\eta)$-evaluations, for each line~$\nu$ we also
maintain a decision tree~$\calT_\eta(\nu)$. We maintain the property
that~$\calT_\eta(\nu)$ is a~$t$-common partial decision tree for
all~$t(\eta)$-evaluations~$\varphi^\nu_\eta(\Gamma_\eta^\nu)$ of
bounded depth.

These common partial decision trees~$\calT_\eta(\nu)$ are useful to
extend the~$t(\eta)$-evaluations~$\varphi_\eta^\nu$ to larger
depths. In each such step, increasing~$\eta$, we apply for each
branch~$\tau$ from~$\calT_\eta(\nu)$ the following multi-switching
lemma to the set of decision
trees~$\restrict{\varphi^\nu_\eta(\Gamma_\eta^\nu)}{\tau}$ of
depth at most~$t(\eta)$. We then extend~$\calT_\eta(\nu)$ in each
leaf~$\tau$ by the the common partial decision tree from the lemma to
obtain $\calT_{\eta+1}(\nu)$ of slightly larger depth.

\begin{restatable}[Multi-switching Lemma]{lemma}{multiSwitch}
  \label{lemma:multiswitch}
  There are absolute constants~$A, c_1, c_2, n_0 > 0$ such that for
  integer~$n \geq n_0$ the following holds.
  Let~$k, M,n',s,t \in \N^+$ satisfy~$n/n' \geq At\log^{c_1} n$,
  $k = n'^2(1\pm 0.01)C\log n'$ be odd, and~$t \leq s \leq n'/32$.
  For~$m_1, \ldots, m_M \in \N^+$ and any decision trees~$T_i^j$ of
  depth at most~$t$, where~$j \in [M]$ and~$i \in [m_j]$, it holds
  that if~$\sigma\sim\calD_k\bigl(\Sigma(n,n')\bigr)$, then the
  probability
  that~$\big(\restrict{\bigvee_{i=1}^{m_j}
    T_i^j}{\sigma}\big)_{j=1}^M$ cannot be represented by
  an~$\ell$-common partial decision tree of depth~$s$ is bounded by
  \begin{align*}
    M^{s/\ell} \left(\frac{At \log^{c_1} n}{n/n'}\right)^{s/c_2}
    \eqperiod
  \end{align*}
\end{restatable}

We defer the proof of \cref{lemma:multiswitch} to
\cref{sec:multi-switch}. In the following we explain how
\cref{thm:main-line-length} follows from \cref{lemma:multiswitch}.

We apply \cref{lemma:multiswitch} with mostly the same parameters. Let
us fix these. We choose~$\ell = t = \log M$, let~$n_0 = n$ and
set~$n_{\eta} = \lfloor n_{\eta-1}/A_1\cdot t \cdot\log^{c_1}
n_{\eta-1} \rfloor$ for~$\eta \in [d]$ and a sufficiently large
constant~$A_1$. The parameter~$s$ depends on~$\eta$ and is fixed
to~$s = s_\eta = 2^{\eta-1} \log N$. With these parameters in place we
can finally also
fix~$t(\eta) = \sum_{i \le \eta} s_i + \log M \le 2^\eta \log N + \log
M$.

\begin{lemma}\label{lem:induction}
  Suppose that for every line~$\nu \in [N]$ we have functionally
  equivalent~$t(\eta-1)$-evaluations~$\varphi^\nu_{\eta-1}$ for
  formulas in~$\Gamma^\nu_{\eta-1}$ along with a~$t$-common partial
  decision tree~$\calT_{\eta-1}(\nu)$
  for~$\varphi^\nu_{\eta-1}(\Gamma_{\eta-1}^\nu)$ of
  depth~$\sum_{i < \eta} s_i$.
  Suppose that~$t(\eta) \le n_{\eta}/16$.
  For~$\sigma_\eta \sim
  \calD_k\bigl(\Sigma(n_{\eta-1},n_{\eta})\bigr)$ with
  probability~$1 - N^{-1}$, for every line~$\nu \in [N]$ there are
  functionally equivalent
  $t(\eta)$-evaluations~$\varphi^\nu_{\eta}$ for formulas
  in~$\Gamma^\nu_\eta$ and a $t$-common partial decision
  tree~$\calT_\eta(\nu)$
  for~$\varphi^\nu_\eta(\Gamma_{\eta}^\nu)$ of
  depth~$\sum_{i \le \eta} s_i$.
\end{lemma}

\begin{proof}
  Let us first extend the common partial decision trees and then
  explain how to obtain the evaluation~$\varphi^\nu_\eta$ for
  different lines~$\nu \in [N]$.

  The interesting formulas of original depth~$\eta$ to consider are
  the ones with a top~$\vee$ gate. Let us fix a line~$\nu \in [N]$ and
  consider all
  sub-formulas~$\set{F^j = \bigvee_{i=1}^{m_j}F^j_i}_{j=1}^{M_\nu}$ of
  line~$\nu$ of original depth~$\eta$ with a top~$\vee$ gate under the
  restriction~$\sigma_{\eta-1}^*$. As the original depth of every
  formula~$F_i^j$ is at most~$\depth(F_i^j) \leq \eta-1$, all these
  formulas are in the domain of the evaluation~$\varphi^\nu_{\eta-1}$.
  Let us further fix a branch~$\tau$ in~$\calT_{\eta-1}(\nu)$ and
  recall that all decision
  trees~$\restrict{\varphi^\nu_{\eta-1}(F_i^j)}{\tau}$ are of depth at
  most~$t$.

  For every~$\nu \in [N]$ and branch~$\tau$ of~$\calT_{\eta-1}(\nu)$
  we apply \cref{lemma:multiswitch} to the set of
  formulas~$\restrict{F_i^j}{\tau}$ with associated
  trees~$\restrict{\varphi^\nu_{\eta-1}(F_i^j)}{\tau}$ of depth at
  most~$t$. The probability of failure of a single application is
  bounded by
  \begin{align}
    M^{s_\eta/\ell}
    \left(
    \frac
    {At \log^{c_1} n_{\eta - 1}}
    {n_{\eta - 1}/n_\eta}
    \right)^{s_\eta/c_2}
    &\leq
    M^{s_\eta/\log M}
    \left(
    \frac
    {At \log^{c_1} n_{\eta - 1}}
    {A_1t \log^{c_1} n_{\eta - 1}}
    \right)^{s_\eta/c_2}\\
    &\leq
      2^{-4s_\eta}
      =
    N^{-2^{\eta+1}} \eqcomma
  \end{align}
  assuming that the constant~$A_1$ is large enough. As we invoke
  \cref{lemma:multiswitch} at
  most~$N \cdot 2^{\sum_{i < \eta} s_i} \le N^{2^{\eta}}$ times, by a
  union bound, with probability at least~$1 - N^{-1}$, there is a full
  restriction~$\sigma_\eta$ such that for every line~$\nu \in [N]$ and
  every branch~$\tau \in \calT_{\eta-1}(\nu)$ we get a $t$-common
  partial decision tree of depth at most~$s_\eta$ for the
  formulas~$(\restrict{F^j}{\tau\sigma_\eta})_{j=1}^{M_\nu}$. Let us
  denote this common decision tree by~$\calT(\nu, \tau)$ and attach it
  to~$\calT_{\eta-1}(\nu)$ at the leaf~$\tau$ to
  obtain~$\calT_\eta(\nu)$.
  The trees~$\calT_\eta(\nu)$ are of depth at
  most~$\sum_{i \le \eta} s_i$ as required.

  Let us explain how to define the evaluation~$\varphi^\nu_\eta$ for a
  fixed line~$\nu \in [N]$. Consider any formula~$F$
  in~$\Gamma_\eta^\nu$.
  \begin{itemize}

  \item If~$F$ is of depth less than~$\eta$, then~$F$ is in the domain
    of~$\varphi^\nu_{\eta-1}$ and we can appeal to
    \cref{lem:phiextend}.

  \item If~$F = \lnot F'$ is of depth~$\eta$,
    then~$\varphi^\nu_{\eta}(F)$ is defined
    from~$\varphi^\nu_{\eta}(F')$ by negating the labels at the
    leaves.

  \item For~$F = \bigvee_i F_i$ of depth~$\eta$ we use the previously
    constructed common partial decision trees. We
    define~$\varphi^\nu_\eta(F)$ to be the decision tree whose
    first~$\sum_{i \le \eta} s_i$ levels are equivalent
    to~$\calT_\eta(\nu)$ followed by~$t$ levels unique to~$F$ obtained
    from the multi-switching lemma.
  \end{itemize}
  
  Let us check that the decision trees~$\calT_\eta(\nu)$ are indeed
  $t$-common partial decision trees
  for~$\varphi^\nu_\eta(\Gamma_\eta^\nu)$. By construction this
  clearly holds for formulas of depth~$\eta$ with a top~$\vee$
  gate. Since~$\calT_\eta(\nu)$ is equivalent to~$\calT_{\eta-1}(\nu)$
  on the upper levels, and restrictions only decrease the depth of
  decision trees, by the initial assumptions this also holds for
  formulas of depth less than~$\eta$. As the $t(\eta)$-evaluations of
  formulas of depth~$\eta$ with a top~$\lnot$-gate are defined in
  terms of formulas of depth less than~$\eta$, we also see
  that~$\calT_\eta(\nu)$ is a $t$-common partial decision tree for
  such formulas.
  
  Last we need to check that each~$\varphi^\nu_\eta$ is a
  $t(\eta)$-evaluation plus that these are pairwise functionally
  equivalent.

  By \cref{lem:phiextend} all the properties hold for formulas of
  depth less than~$\eta$. Let us verify the $t(\eta)$-evaluation
  properties for formulas of depth~$\eta$. That is, we need to check
  that constants (variables) are mapped to the corresponding decision
  trees of depth 0 (of depth 1), that axioms are assigned to
  $1$-trees, that the negation of a formula is assigned to the same
  decision tree as the formula is except that the leaf-labels are
  negated, and that if a formula~$F$ is the or of some sub-formulas,
  then the decision tree that~$F$ is mapped to represents the or of
  the decision trees that the sub-formulas are mapped to. Let us check
  these properties.

  \Cref{prop:depth-0} is immediate as~$\eta > 0$. As we only consider
  locally consistent decision trees as defined in
  \cref{def:local-consistency-dec-trees}, \cref{prop:1-tree} also
  follows. Further, \cref{prop:not} is satisfied by
  construction. \Cref{prop:or} can be established by checking the
  property for each branch~$\tau$ in~$\calT_{\eta-1}(\nu)$ separately;
  for a fixed~$\tau$ we see by \cref{lemma:multiswitch} that this
  indeed holds.
  
  Finally we need to establish that two
  $t(\eta)$-evaluations~$\varphi^\nu_\eta$ and~$\varphi^{\nu'}_\eta$
  are functionally equivalent for formulas of depth~$\eta$. By the
  inductive hypothesis isomorphic formulas with a top~$\lnot$-gate are
  functionally equivalent. Hence we are left to check functional
  equivalence for isomorphic formulas of depth~$\eta$ with a
  top~$\vee$ gate.
  
  Let~$F = \bigvee_i F_i$ and~$F' = \bigvee_i F_i'$ be two isomorphic
  formulas from~$\Gamma_{\eta}^\nu$ and~$\Gamma_\eta^{\nu'}$
  respectively. For the sake of contradiction suppose
  that~$\restrict{\varphi^\nu_\eta(F)}{\tau} = 1$
  but~$\restrict{\varphi^{\nu'}_\eta(F')}{\tau} = 0$ for some
  assignment~$\tau$. In the following we use
  that~$t(\eta) \le n_\eta/16$ to argue that there are locally
  consistent branches as claimed.

  By \cref{prop:1-tree} we know that for some~$F_i$ it holds
  that~$\restrict{\varphi^\nu_\eta(F_i)}{\tau} = 1$. Since the
  formulas~$F$ and~$F'$ are isomorphic formulas we know that there is
  an~$F'_j$ such that~$F_i$ and~$F'_j$ are isomorphic formulas. As
  such formulas have functionally equivalent decision trees (by
  induction and \cref{lem:phiextend}) we get
  that~$\restrict{\varphi^{\nu'}_\eta(F'_j)}{\tau} = 1$. But this
  cannot be as by \cref{prop:or} of a $t(\eta)$-evaluation this
  implies that~$\restrict{\varphi^{\nu'}_\eta(F')}{\tau} = 1$. This
  establishes that the different $t(\eta)$-evaluations are
  functionally equivalent, as required.
\end{proof}

With all pieces in place we are ready to prove
\cref{thm:main-line-length}.

\begin{proof}[Proof of \cref{thm:main-line-length}]
  Suppose we are given a proof of
  length~$N = \exp(n/((\log n)^c \log M)^{d})$, for some
  constant~$c > 0$. We may assume
  that~$M \le \exp(n^{1/d - 1/d(d-1)})$, as otherwise we can apply
  \cref{thm:main-size}.
  
  In order to create the functionally equivalent
  $t(\eta)$-evaluations~$\varphi^\nu$ for each line~$\nu \in [N]$ we
  consecutively apply \cref{lem:induction} $d$ times.
  We start with the evaluation~$\varphi^\nu_0$ which maps constants to
  the appropriate depth~$0$ decision tree and variables to the
  corresponding depth~$1$ decision trees. The common partial decision
  trees~$\calT_0(\nu)$ are all empty.

  After applying \cref{lem:induction} $d$ times we are left with a
  $t(d)$-evaluation for the proof. We need to ensure that~$t(d)$ is
  upper bounded by the dimension of the final
  grid:~$t(d) \le 2^{d}\log N + \log M$, while the final side length
  of the grid is~$n \cdot (2A_1 (\log n)^{c_1} \log M)^{-d}$. For our
  choice of~$N$ and the assumption on~$M$ this indeed holds and by
  \cref{lem:noproof} the theorem follows.
\end{proof}

\section{Switching Lemma: Proof Outline \& Further Preliminaries}
\label{sec:singleswitch-overview}

In this section we revisit full restrictions and define some
bookkeeping objects that are used in the proof of the switching
lemma. The actual proof of \cref{lemma:switch} is carried out in
\cref{sec:singleswitch}.

In order to motivate the following definitions we give a very
high-level proof outline of \cref{lemma:switch} in the following
section.

\subsection{High Level Proof Outline}
\label{sec:proof-outline-high-level}

We are given~$m$ decision trees~$T_i$ of depth at
most~$\depth(T_i)\leq t$ that query the edges of the~$n \times n$
grid. We sample a full
restriction~$\sigma \sim \calD_k\bigl(\Sigma(n,n')\bigr)$ and want to
argue that the probability that there is no decision tree of depth~$s$
representing~$\bigvee_{i=1}^{m}\restrict{T_i}{\sigma}$ is
exponentially small in~$s$.

We bound this probability by constructing the so-called \emph{extended
  canonical decision tree}~$\calT$ that
represents~$\bigvee_{i=1}^{m}\restrict{T_i}{\sigma}$ and bounding the
probability that~$\calT$ is of depth~$\depth(\calT) > s$.

For now we can think of~$\calT$ being constructed like the canonical
decision tree: proceed in stages. In each stage a branch~$\tau$
of~$\calT$ is extended by querying the variables of the first
$1$-branch~$\psi$ in the decision
trees~$\restrict{T_1}{\sigma\tau}, \restrict{T_2}{\sigma\tau}, \ldots,
\restrict{T_m}{\sigma\tau}$. Once queried we check in each new leaf of
the tree whether we traversed the path~$\psi$. If so, then we label
the leaf with a~$1$ and otherwise we continue with the next stage. If
there are no $1$-branches left, then we label the leaf with a~$0$.

It is not so hard to see that this process results in a decision
tree~$\calT$ that
represents~$\bigvee_{i = 1}^m \restrict{T_i}{\sigma}$: for each
leaf~$\tau$ of $\calT$ that is labeled~$1$ it holds that there is
an~$i \in [m]$ such that~$\restrict{T_i}{\sigma\tau} = 1$ and if the
branch~$\tau$ is labeled $0$, then for all~$i \in [m]$ we have
that~$\restrict{T_i}{\sigma\tau}$ is a $0$-tree. It remains to argue
that the decision tree~$\calT$ is of depth at most~$s$ except with
probability exponentially small in~$s$.

We analyze this event using the labeling technique of
Razborov~\cite{Razborov1995}. The idea of this technique is to come up
with an (almost) bijection~$F$ mapping restrictions~$\sigma$ that give
rise to an extended canonical decision tree~$\calT$ of
depth~$\depth(\calT) \geq s$ to a space of restrictions~$\Sigma^*$
that is much smaller than the space~$\Sigma$ from which we sampled the
full restriction~$\sigma$. If we manage to exhibit such an~$F$ to a
space~$\Sigma^*$ that is an exponential in~$s$ factor smaller
than~$\Sigma$, then the claimed upper bound on the failure probability
follows.

Let us sketch the construction of such an (almost) bijection. Consider
a full restriction~$\sigma$ and decision trees~$T_1, \ldots, T_m$ that
give rise to an extended canonical decision tree~$\calT$ of
depth~$s$. Let~$\tau \in \calT$ be a $0$-branch of length~$s$ and
denote by~$\psi_1, \ldots, \psi_g$ the $1$-branches used in the stages
that constructed the branch~$\tau$. We know that~$\tau$ does not agree
with any of~$\psi_1, \ldots, \psi_g$ as otherwise~$\tau$ would end in
a $1$-leaf. Let~$\tau_1, \ldots, \tau_g \subseteq \tau$ be the partial
assignments to the variables of the corresponding~$\psi_j$, that is,
the domain~$\dom(\tau_j) = \dom(\psi_j)$ of these assignments for
all~$j \in [g]$ are equal.

Razborov~\cite{Razborov1995} maps~$\sigma$ to~$F(\sigma) = \sigma^*$:
the restriction~$\sigma^*$ is~$\sigma$ composed
with~$\psi_1, \ldots, \psi_g$, that is, we somehow ``add'' the
restrictions~$\psi_j$ to~$\sigma$ such that the branch~$\psi_1$ is
traversed by the resulting restriction~$\sigma^*$. The crucial insight
is that with the help of the shallow decision trees~$T_1, \ldots, T_m$
the mapping~$F$ can be cheaply inverted by ``re-doing'' the
construction of~$\tau$: we proceed in stages~$j = 1, \ldots, g$. At
the beginning of stage~$j$ we assume that we have the
restriction~$\sigma^*_{\geq j}$ which is~$\sigma$ composed
with~$\tau_1, \ldots, \tau_{j-1}, \psi_j, \ldots, \psi_g$ such that
the branch~$\psi_j$ is traversed. Since~$\sigma^*_{\geq j}$ by
construction traverses~$\psi_j$ we can identify~$\psi_j$ for free: it
is the first $1$-branch in~$T_1, \ldots, T_m$ traversed
by~$\sigma^*_{\geq j}$. We intend to recover~$\tau_j$ so that we can
``remove''~$\psi_j$ from~$\sigma^*_{\geq j}$ and ``add''~$\tau_j$ to
obtain~$\sigma^*_{\geq j+1}$. Since~$\psi_j$ is a branch of depth at
most~$t$, using only~$\log t$ bits of external information per
variable we can point out all the variables where~$\psi_j$
and~$\tau_j$ differ.

Thus using in total at most~$s\log t$ bits we seem to be able to
reconstruct~$\sigma$ from~$\sigma^*$. If we could further argue that
the space of restrictions~$\Sigma^*$ that~$F$ maps to is a factor
$(n/n')^{-s}$ smaller than~$\Sigma(n,n')$, then we could bound the
failure probability by
\begin{align}
  \left(\frac{t}{n/n'}\right)^{s} \eqperiod
\end{align}
This completes the high-level proof outline.

It is not immediately clear how to implement the above proof outline
for the Tseitin contradictions~$\tseitin(G_n)$ defined over
the~$n \times n$ grid. One of the ideas of \cite{jhtseitin} is, given
a full restriction~$\sigma$, to try to reduce the number of centers
in~$\sigma_0$ with an even charge, where~$\sigma_0$ is the assignment
used in the construction of~$\sigma$ (see
\cref{sec:full-restriction}). This approach does not work immediately:
there are \emph{more} assignments to the Tseitin formula with an even
charge at all chosen centers except two (chosen freely) than there are
assignments~$\sigma_0$ with an even charge at all chosen centers.

We follow \cite{jhtseitin} and define \emph{partial
  restrictions}~$\rho$ which have a large number of centers with an even
charge. We can think of~$\restrict{\tseitin(G_n)}{\rho}$ as an
intermediate formula between~$\tseitin(G_n)$ and the
formula~$\restrict{\tseitin(G_n)}{\sigma}$ restricted by a full
restriction~$\sigma$. These partial restrictions have the property
that the number of restrictions decreases as the number of centers
with an even charge decreases. This allows us to implement the above
proof outline.

\paragraph{Organization.} In \cref{sec:assoc-center} we define the
concept of an associated center which allows us to recover a center
from a discovered variable on a branch~$\psi_j$. In
\cref{sec:partial-restriction} we then introduce the notion of a
\emph{partial restriction}.  Finally, in \cref{sec:information}, we
introduce the main bookkeeping objects of the proof: \emph{information
  pieces}.

\subsection{Associated Centers}
\label{sec:assoc-center}

Recall that each path connecting two centers in adjacent sub-squares
consists of 5 non-empty \emph{segments}: The first and last segment
are within the central area and of length 1, the middle segment is
contained in-between the central areas on the designated row or
column, and segments two and four pass from the central areas to the
area in between. Illustrations can be found in
\cref{fig:path,fig:associated-center}.
\begin{figure}
  \centering
  \includegraphics{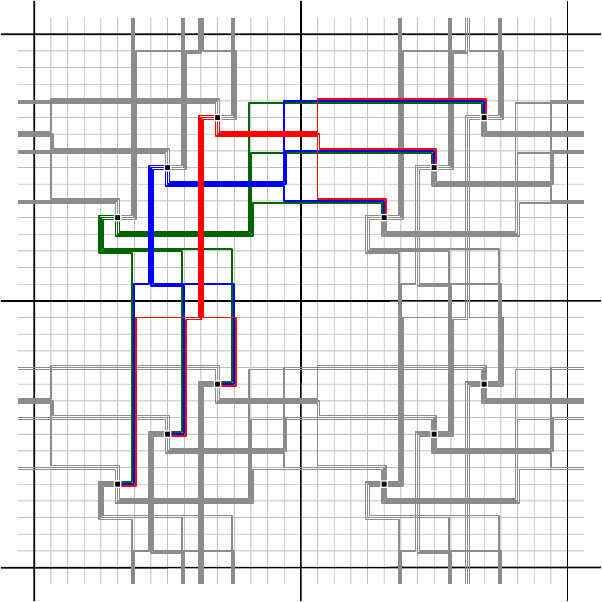}
  \caption{The paths from the top-left sub-square going right and down
    are highlighted}
  \label{fig:associated-center}
\end{figure}
The key property of these paths is stated in the below lemma.

\begin{lemma}
  \label{lemma:path}
  If an edge $e$ lies on multiple paths as described, then all these
  paths have a common endpoint.
\end{lemma}
\begin{proof}
  Consider the set of paths $\calP$ connecting the centers of a fixed
  sub-square $s$ to centers in adjacent sub-squares. Denote by
  $P \in \calP$ a path that connects a center $c \in s$ to a center
  $c'$ in the sub-square below or to the right of $s$. Recall that
  in between two consecutive centers there are at least two
  rows/columns with no centers.

  Consider the 5 segments of $P$. If some path $P' \in \calP$ shares
  an edge with $P$ on segments $1$, $2$ or $3$, then $P$ and $P'$
  share $c$ as a common endpoint. Similarly, if $P'$ shares an edge on
  segments $4$ or $5$ with $P$, then they have $c'$ as a common
  endpoint. The statement follows by symmetry.
\end{proof}

We let the \emph{associated center} of an edge~$e$ denote the endpoint
common to all paths that contain~$e$. This notion is naturally
extended to variables: the \emph{associated center} of a
variable~$x_e$ is the associated center of the edge~$e$.

Let us introduce some notation. Given a fixed center~$v$ we say that a
path~$P$ connecting~$v$ to some other center~$u$ goes to the~$\delta$,
where $\delta$ is one of the directions \emph{left}, \emph{right},
\emph{up} or \emph{down}, if~$v$ lies in the sub-square to
the~$\delta$ of~$u$. Let us stress that this notion of direction is
relative to the fixed center~$v$: if we fix~$u$ and then considered
the same path~$P$, then~$P$ has the opposite direction
of~$\delta$. Finally note that if we fix the associated center of an
edge~$e$, then all the paths passing through~$e$ have the same
direction.

Recall from the proof outline provided in
\cref{sec:proof-outline-high-level} that during the recovery of the
assignment~$\tau_j$ it is possible to cheaply identify variables~$x_e$
on the short branch~$\psi_j$. The notion of an associated center
allows us to pass from such a variable~$x_e$ to a center. This enables
us to do a counting argument in terms of centers. A more detailed
proof outline is provided in \cref{sec:switch-overview}.

\subsection{Partial Restrictions and Pairings}
\label{sec:partial-restriction}
\label{sec:partial-restrictions}

Recall that we sample a full restriction
$\sigma \sim \calD_k\bigl(\Sigma(n,n')\bigr)$ by first sampling
$k$~\emph{alive} centers such that each sub-square contains
$(1\pm 0.01)k/n'^2$ alive centers and then fixing the alive centers
with the lowest numbered row in each sub-square to be the
\emph{chosen} centers~$\chosen$ of~$\sigma$. We then sample an
assignment~$\sigma_0$ from the space of solutions to the Tseitin
formula with $0$-charges at chosen centers and $1$-charges at all
other nodes. We define~$\sigma$ from~$\sigma_0$ as done in
\cref{eq:def-sigma}.

In the following we define a so-called \emph{partial
  restriction}~$\rho$. A partial restriction is in some sense one half
of a full restriction~$\sigma$: we split a full restriction~$\sigma$
into a partial restriction~$\rho$ and a \emph{pairing}~$\pi$. The
partial restriction~$\rho$ is an (affine) restriction similar
to~$\sigma$ but it maps to more variables: while~$\sigma$ maps to
variables~$y_P$ for~$P$ a chosen path, the partial restriction~$\rho$
maps to variables~$z_P$ where~$P$ is any path connecting two
\emph{alive} centers. The pairing~$\pi$ is then the (affine)
restriction such
that~$\restrict{\restrict{\tseitin(G_n)}{\rho}}{\pi} =
\restrict{\tseitin(G_n)}{\sigma}$. Let us define~$\rho$ and~$\pi$ more
formally by describing an alternative way of sampling a full
restriction~$\sigma\sim \calD_k\bigl(\Sigma(n,n')\bigr)$.

As before we start by sampling uniformly at random~$k$ \emph{alive
  centers} from the set of subsets of centers of size~$k$ satisfying
that each sub-square contains~$(1\pm 0.01)k/n'^2$ alive
centers. Sample~$\rho_0$ from the space of solutions to the Tseitin
formula with $0$-charges at alive centers and $1$-charges at all other
nodes. Since we have an odd number~$k$ of alive centers, by
\cref{lemma:evenok}, this is indeed possible. We have variables~$z_P$
for each path~$P$ connecting two alive centers (an \emph{alive path}),
denote by~$\oplus$ the exclusive or, and define~$\rho$ from~$\rho_0$
by
\begin{align}
  \rho(x_e) =
  \begin{cases}
    \rho_0(x_e)
    &\text{if $e$ is not  on an alive path,}\\
    \bigoplus_{P \ni e} z_P
    &\text{if $e$ is on some alive path(s) and $\rho_0(x_e) = 1$,}\\
    \neg \bigoplus_{P \ni e} z_P
    &\text{if $e$ is on some alive path(s) and $\rho_0(x_e) = 0$.}\\
  \end{cases}
  \label{eq:def-rho}
\end{align}
The assignment given by~$\rho_0$ to variables not on alive paths are
called the \emph{final values}.

It remains to define the pairing~$\pi$ such that~$\rho$ combined
with~$\pi$ gives a full restriction~$\sigma$. The alive center with
the lowest numbered row in each sub-square is called the \emph{chosen
  center} (as before) and the other alive centers are called
\emph{non-chosen centers}. Similarly we denote paths that are alive
but not chosen the \emph{non-chosen paths} and let centers that are
not alive simply be the \emph{dead centers}.

The main task of the pairing~$\pi$ is to ensure that the non-chosen
centers have an odd charge in the final full restriction. We can thus
think of~$\pi$ as an assignment to the non-chosen paths such that each
non-chosen center has an odd number of incident non-chosen paths that
are set to 1. For reasons to become apparent later on it is convenient
to have small $1$-components in~$\pi$. Let a star of size~$4$ be the
graph with a central node of degree~$3$ and three nodes of degree~$1$
connected to the central node by an edge.

\begin{definition}[graphical pairing]
  \label{def:pairing}
  A \emph{graphical pairing} $\pi_0$ is a graph supported on the
  non-chosen centers. Each component of $\pi_0$ is either a single
  edge or a star of size $4$. The centers of a component are in
  distinct but adjacent sub-squares.
\end{definition}

The following lemma follows from the proof of~\cite[Lemma
4.3]{jhtseitin}. For completeness we provide a proof in
\cref{sec:omitted-proofs}.

\begin{restatable}[{\cite[Lemma 4.3]{jhtseitin}}]{lemma}{matchinglemma}
  \label{lem:matching}
  For large enough integer $a \in \N$ the following holds. If each
  sub-square has $(1\pm 0.01)a$ alive centers, then there is a
  graphical pairing $\pi_0$.
\end{restatable}

For each choice of alive centers we fix a graphical
pairing~$\pi_0$. Let us stress that there is \emph{no} randomness in
the choice of~$\pi_0$ once we have sampled a partial
restriction~$\rho$.

We have variables $y_P$ for each chosen path $P$ connecting two chosen
centers. We define the \emph{pairing} $\pi$ from a graphical pairing
$\pi_0$ by
\begin{align}
  \pi(z_P) =
  \begin{cases}
    1
    &\text{if $P$ is non-chosen and in $\pi_0$ as an edge,}\\
    0
    &\text{if $P$ is non-chosen and not in $\pi_0$,}\\
    y_P
    &\text{if $P$ is a chosen path.}\\
  \end{cases}
  \label{eq:def-pi}
\end{align}
See \cref{fig:rho} for an illustration. With $\pi$ and $\rho$ defined
we let $\sigma(x_e) = \pi\bigl(\rho(x_e)\bigr)$.

\begin{figure}
  \centering
  \includegraphics{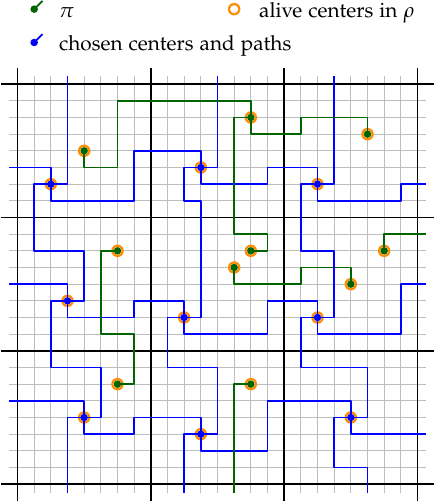}
  \caption{A part of the grid with $\pi$ and the chosen centers and
    paths highlighted}
  \label{fig:rho}
\end{figure}

We claim that sampling $\sigma$ through $\rho$ and $\pi$ as outlined
above is equivalent to sampling it directly as explained in
\cref{sec:distribution}. This is readily verified: we may assume that
we sample the same set of alive centers. Since $\pi$, which only
depends on the set of alive centers, defines a bijection between
assignments $\rho_0$ and $\sigma_0$ as sampled in the respective
processes, we see that the two distributions obtained from these
processes are indeed equivalent.

The sampling of~$\rho$ is the main probabilistic event that we analyze
in the proof of the switching lemma. Since~$\sigma$ can be defined in
terms of~$\rho$ and~$\pi$ we write~$\sigma = \sigma(\rho, \pi)$
whenever we want to stress this fact.

Instead of sampling the~$k$ alive centers uniformly from all subsets
of centers of size~$k$ satisfying that each sub-square
has~$(1\pm 0.01)k/n'^2$ alive centers, we can instead simply
sample~$k$ centers uniformly form the set of centers and then
sample~$\rho$ as above. Denote by~$\spaceRho(k,n,n')$ the space of
partial restrictions obtained by choosing~$k$ alive centers uniformly
from the set of centers and write~$\spaceRhoReg(k,n,n')$ for the space
of partial restrictions with~$k$ alive centers such that each of
the~$n'\times n'$ sub-squares contains~$(1\pm 0.01)k/n'^2$ many alive
centers.

\begin{lemma}
  \label{lem:no-alive}
  For~$n,n'\in \N$ large enough there is a constant $C$ such that for
  odd $k \geq Cn'^2 \log n'$ it holds that
  $\spaceRhoReg(k,n,n') \geq (1-1/n')\spaceRho(k,n,n')$.
\end{lemma}

\begin{proof}
  Sample~$\rho \sim \calU\bigl(\spaceRho(k,n,n')\bigr)$ uniformly and
  consider the random variables~$X_{(i,j)}$ defined for each
  sub-square~$(i,j)$ counting the number of alive centers of~$\rho$ in
  the sub-square~$(i,j)$. Each such~$X_{(i,j)}$ is the sum of the
  indicator random variables $Y_{(i,j)}^\nu$ for $\nu \in [\Delta]$
  that are $1$ if and only if the $\nu$th center of the
  sub-square~$(i,j)$ is alive in~$\rho$. Note that these random
  variables are negatively correlated since we choose exactly~$k$
  alive centers overall. Since the Chernoff bounds continue to hold
  for negatively correlated random variables it holds
  that~$X_{(i,j)} \in (1\pm 0.01)k/n'^2$ except with probability
  $1/n'^3$ for a large enough constant $C> 0$. The claim follows by a
  union bound over all sub-squares.
\end{proof}

\begin{lemma}\label{lem:final-counting}
  Let $n, n' \in \N^+$ be large enough such that
  $n \geq 20 n' C\log n'$ and suppose that $k = Cn'^{2} \log n'$. For
  integer~$s \geq 1$ it holds that
  \begin{align*}
    \frac{
    \lvert\spaceRho(k-2s,n,n')\rvert
    }
    {
    \lvert\spaceRhoReg(k,n,n')\rvert
    }
    \leq
    \left(
    \frac
    {
    13 C \log n'
    }
    {
    n/n'
    }
    \right)^{2s}
     \eqperiod
  \end{align*}
\end{lemma}

The whole proof of the switching lemma hinges on this exponential in
$s$ factor~$\left(\frac{13 C \log n'}{n/n'}\right)^{2s}$: it allows us
to implement the proof plan as outlined in
\cref{sec:proof-outline-high-level}
with~$\Sigma = \spaceRhoReg(k,n,n')$
and~$\Sigma^* = \bigcup_{i=\Omega(s)}^{(k-1)/2}\spaceRho(k-2i,n,n')$.

\begin{proof}
  According to \cref{lem:no-alive} it holds that
  $\lvert\spaceRhoReg(k,n,n')\rvert \geq
  (1-1/n)\lvert\spaceRho(k,n,n')\rvert$. It thus suffices to show that
  $\lvert \spaceRho(k-2s,n,n')\rvert \leq \left(\frac{12 C \log
      n'}{n/n'}\right)^{2s} \cdot \lvert\spaceRho(k,n,n')\rvert$ to
  establish the claim.  Recall from \cref{lemma:evenok} that the
  number of solutions of the Tseitin formula $\tseitin(G_n,\alpha)$
  only depends on the parity of the sum of the charges
  $\sum_{v\in V(G)} \alpha_v$ and $n$: the space of restrictions
  $\spaceRho(k-2s,n,n')$ is of size $2^{r_n}\binom{m}{k-2s}$ where
  $m = n'^2\Delta$ is the number of centers and $r_n$ is as in
  \cref{lemma:evenok}. It thus holds that
  \begin{align}
    \frac
    {\lvert\spaceRho(k-2s,n,n')\rvert}
    {\lvert\spaceRho(k,n,n')\rvert}
    =
    \frac
    {2^{r_n}\binom{m}{k-2s}}
    {2^{r_n}\binom{m}{k}}
    =
    \prod_{i=0}^{2s-1}
    \frac {k-i} {m-k+i} 
    &\leq
      \left(
      \frac
      {k}
      {m-k}
      \right)^{2s}\\
    &= 
      \left(
      \frac
      {C \log n'}
      {\Delta - C \log n'}
      \right)^{2s} \\
    &\leq
      \left(\frac{2 C\log n'}{\Delta}\right)^{2s} \eqcomma
  \end{align}
  using the assumption that $n \geq 20 n' C\log n'$ hence
  $\Delta \geq 2C \log n'$. The claimed bound follows from the
  fact~$\Delta \geq n/6n'$.
\end{proof}

Let us establish some nomenclature. As the original grid is also a
graph with edges we from now on use the term \emph{grid-edges} to
refer to edges in the original grid. We only consider paths in the
original grid and keep the short term \emph{path} for these. An
\emph{edge} refers to an alive path, that is, an \emph{edge} is a
connection between two alive centers and corresponds to an alive path
(possibly a chosen path) in the original grid. Edges between chosen
centers are \emph{new grid edges} and we say that two chosen centers
are neighbors if they lie in adjacent sub-squares.

A partial restriction is usually denoted by~$\rho$ and since we mostly
discuss partial restrictions we simply call them \emph{restrictions}
while we use the term \emph{full restrictions} when that is what we
have in mind.

\subsection{Information Pieces}
\label{sec:information}

As an intermediate between~$\rho$ and a full restriction~$\sigma$ we
have~$\rho$ along with some information in the form of the existence
or absence of edges (alive paths) incident to alive centers. We have
the following definition.

\begin{definition}[information piece]
  \label{def:information-piece}
  An \emph{information piece} is either
  \begin{enumerate}
  \item an edge $\set{v,w}$ where $v,w$ are centers in adjacent
    sub-squares, or
  \item of the form $(v, \delta, \bot)$ for a center $v$ and a
    direction~$\delta$, that is, $\delta$ is either \emph{left},
    \emph{right}, \emph{up} or \emph{down}.
  \end{enumerate}
  The former says that there is an edge from $v$ to $w$ while the
  latter says that there is no edge from $v$ in direction $\delta$.
\end{definition}

We can think of information pieces as always being on alive
centers. Suppose we are given a restriction~$\rho$ together with an
information piece $\set{v,w}$ for alive centers~$v$ and~$w$. The
tuple~$(\rho, \set{v,w})$ corresponds to the restriction obtained
from~$\rho$ by additionally restricting
\begin{enumerate}
\item the variable $z_{P_{vw}}= 1$, where $P_{vw}$ is the alive path
  connecting the centers $v$ and $w$, and
\item setting variables $z_{P'} = 0$ for any alive path
  $P' \neq P_{vw}$ which lies in the same direction as $P_{vw}$ does
  from either $v$ or $w$.
\end{enumerate}
Let us stress that the second point above is by choice: it would be
perfectly fine to allow multiple paths set to $1$ in one direction
incident to a single center. In the following we do not consider such
assignments.

Note that there is some asymmetry in information pieces: if we are
instead given $\rho$ along with~$(v, \text{down}, \bot)$, then this
corresponds to the restriction obtained from $\rho$ where we
additionally restrict all variables $z_P = 0$, where $P$ is an alive
path going down from $v$.

We make the intuition of what restriction one obtains from $\rho$ with
some information pieces more formal in \cref{def:force}. It is
convenient to have the abstraction of an information piece since it
allows us to think of certain alive paths as restricted while only
knowing one endpoint. We also use sets of information pieces.

\begin{definition}[information set]
  \label{def:information-set}
  An \emph{information set} $I$ is a collection of information pieces.
  The \emph{support} of $I$, denoted by $\supp(I)$, is the set of
  centers mentioned in these pieces.
\end{definition}

A partial assignment to new grid-edges naturally corresponds to a set
of information pieces: an assignment of $0$ to a new grid-edge
$P_{uv}$ corresponds to two non-edge information pieces in the
appropriate directions at the chosen centers $u$ and $v$. An
assignment of $1$ corresponds to an information piece in the form of
an edge between the two chosen centers $u,v$ connected by the
path~$P_{uv}$.

\begin{definition}[local consistency for information sets]
  \label{def:local-consistency-info}
  An information set $I$ is \emph{locally consistent} if
  \begin{enumerate}
  \item the set~$I$ does not have two different information pieces in
    one direction from the same center, and
  \item if~$I$ has information in all four directions from a center
    $v$, then it has an odd number of edges incident to $v$.
  \end{enumerate}
  Two information sets $I$ and $J$ are \emph{pairwise locally
    consistent} if the union $I \cup J$ is locally consistent.
\end{definition}

We use the term \emph{locally consistent} both for sets of information
pieces and partial assignments. Local consistency for assignments
requires an odd number of 1s incident to any node if it assigns all
the incident variables. This corresponds exactly to the local
consistency property of information pieces. Hence a locally consistent
assignment gives rise to a locally consistent information set. Note,
though, that the converse is not true since a locally consistent
assignment needs to satisfy further properties (see
\cref{def:consistent}).

Jointly with $\rho$ an information set forces the values of some more
variables as follows.

\begin{definition}[forcing]
  \label{def:force}
  Let $\rho$ be a restriction, denote by $\rho_0$ the assignment used
  in the construction of $\rho$, and let $I$ be an information set. A
  variable $x_e$ is \emph{forced} by $(\rho, I)$ if and only if either
  \begin{enumerate}
  \item the associated center~$v$ of~$x_e$ is dead in~$\rho$, or
  \item the information set~$I$ has information pieces in all
    directions incident to~$v$.
  \end{enumerate}
  In the former case the variable $x_e$ is always forced to
  $\rho_0(x_e)$, while in the later case it is forced to
  $\neg\rho_0(x_e)$ if the information piece $i\in I$ incident to~$v$
  in the direction of~$x_e$ is
  \begin{enumerate}
  \item an edge $i = \set{v,w}$, and
  \item the grid-edge~$e$ lies on the alive path~$P$ connecting $v$ to
    $w$.
  \end{enumerate}
  Otherwise, as in the first case, the variable~$x_e$ is forced
  to~$\rho_0(x_e)$.
\end{definition}

\begin{remark}
  Readers familiar with~\cite{jhtseitin,jhkr} should note that our
  notion of forcing differs from the notion used in previous work: a
  variable with an alive center is considered forced only if the
  information set~$I$ contains information pieces in \emph{all}
  directions incident to the associated
  center. In~\cite{jhtseitin,jhkr} a variable~$x_e$ is considered
  forced if~$I$ contains an information piece incident to the
  associated center in the direction of~$x_e$.
\end{remark}

There are other situations where the value of a variable might be
determined by $\rho$ and $I$ such as the lack, or scarcity, of live
centers in a sub-square. We do not use such information in the
reasoning below. We need a notion of a closed information set.

\begin{definition}[closed information set]
  An information set~$I$ is \emph{closed at a node~$u$} if~$I$ is
  locally consistent and has information pieces in all four directions
  incident to~$u$. The information set~$I$ is \emph{closed} if~$I$ is
  closed at all nodes $u \in \supp(I)$.
\end{definition}

The definition implies that if an information set~$I$ is closed, then
for any $u\in \supp(I)$ and direction~$\delta$ where there is not an
element of~$\supp(I)$ we have a non-edge $(u,\delta,\bot)$. When
considered as a graph such an information set is an odd-degree graph
(with degrees one and three) on the centers of $\supp(I)$. See
\cref{fig:closed-info} for an illustration of a closed information
set. Going forward we will usually think of each connected component
of the pairing~$\pi$ as a closed information set.

\begin{figure}
  \centering
  \includegraphics{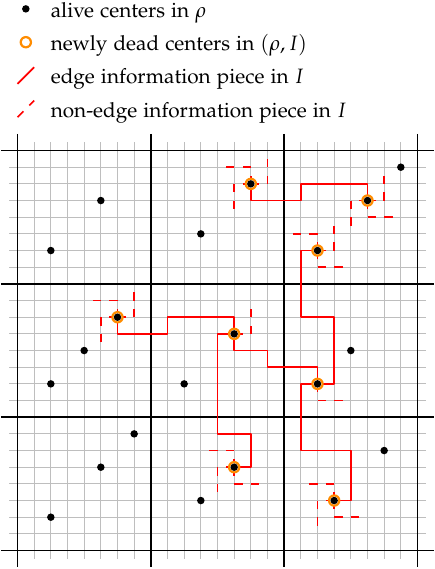}
  \caption{A closed information set $I$ and the newly dead centers in
    $\supp(I)$}
  \label{fig:closed-info}
\end{figure}

Consider an information set~$I$ supported on alive centers that is
closed on~$U$ and suppose that on centers~$v \in \supp(I) \setminus U$
the set~$I$ has an even number of incident edges. Note that the
variables forced by~$(\rho, I)$, can be described by a restriction
with the centers~$U$ killed: negate the values of any variable on any
path in~$I$ and then forget that the centers in~$U$ were alive. More
precisely, if $\rho_0$ denotes the assignment used in the construction
of $\rho$, then the restriction $\rho^*$ defined by
\begin{align}
  \label{eq:rho-star}
  \rho^*(x_e)
  &=
    \begin{cases}
      b
      &\text{if~$(\rho, I)$ forces~$x_e$ to~$b \in \set{0,1}$, and}\\
      \rho(x_e)
      &\text{otherwise,}
    \end{cases}
\end{align}
assigns the same variables to the same constants as $(\rho, I)$
forces. Thus, if we let such an information set~$I$ operate on a
restriction~$\rho$ we obtain a restriction with fewer live centers
where the killed centers are precisely the centers in~$U$.

\section{Proof of the Switching Lemma}
\label{sec:singleswitch}

This section is dedicated to the proof of the switching lemma restated
for convenience.

\singleSwitch*

The proof of \cref{lemma:switch} closely follows the argument
of~\cite{jhtseitin} with some simplifications and changes. In the
following section we give a detailed proof outline. The details of the
proof are fleshed out in
\cref{sec:ext-can-dec-tree,sec:ext-can-dec-tree-prop,sec:encode-rho,sec:encode}.

\subsection{Detailed Proof Outline}
\label{sec:switch-overview}

Let us recall the setup. We have a full
restriction~$\sigma = \sigma(\rho, \pi)$ as defined in
\cref{sec:restrictions} that consists of a restriction~$\rho$ and a
pairing~$\pi$ as introduced in \cref{sec:partial-restriction}. The
restriction~$\rho$ has~${(1 \pm 0.01)C \log(n)}$ alive centers in each
sub-square for some (large) constant~$C>0$.
For~${i \in [m]}$ we have decision trees~$T_i$ of depth at
most~$\depth(T_i)\leq t$ querying grid-edges. We want to bound the
probability that there is no decision tree of depth~$s \geq t$
representing~$\bigvee_{i=1}^{m}\restrict{T_i}{\sigma}$.

In the following we construct a decision tree~$\calT$ that
represents~$\bigvee_{i=1}^{m}\restrict{T_i}{\sigma}$ which is with
high probability over the choice of~$\rho$ of depth at most~$s$. Let
us stress that in contrast to the decision trees~$T_i$ that query
grid-edges, this decision tree~$\calT$ queries edges, that is, path
variables.

The decision tree~$\calT$ is constructed in a similar manner to the
construction of a canonical decision tree: we proceed in stages where
in each stage a branch~$\tau$ of~$\calT$ is extended by querying
variables related to the first $1$-branch~$\psi$ in the
trees~$\restrict{T_1}{\sigma\tau}, \restrict{T_2}{\sigma\tau}, \ldots,
\restrict{T_m}{\sigma\tau}$. For now it is not so important what the
``related variables of~$\psi$'' precisely are and we can simply think
of these as the variables on the branch~$\psi$. Once all these
variables have been queried we check in each new leaf of the tree
whether we traversed the path~$\psi$. If so, then we label the leaf
with a~$1$ and otherwise continue with the next stage. If there are no
$1$-branches left, then we label the leaf with a~$0$.

As argued in \cref{sec:proof-outline-high-level} it is quite immediate
that the above process indeed results in a decision tree~$\calT$ that
represents~$\bigvee_{i = 1}^m \restrict{T_i}{\sigma}$. It remains to
argue that~$\calT$ is with high probability of depth at most~$s$.

We analyze this event using the labeling technique of Razborov
\cite{Razborov1995}. The idea of this technique is to come up with an
(almost) bijection from restrictions~$\rho$ that give rise to a
decision tree~$\calT$ of depth larger than~$s$ to a set of
restrictions that is much smaller than the set of all restrictions. In
a bit more detail, given such a bad~$\rho$, we create a
restriction~$\rho^*$ with fewer live centers such that with a bit of
extra information we can recover the restriction~$\rho$
from~$\rho^*$. As~$\rho^*$ has roughly~$s$ fewer live centers
than~$\rho$ we obtain our statement by \cref{lem:final-counting}.

Let us explain in a bit more detail how to construct~$\rho^*$ from
a~$\rho$ that gives rise to a decision tree~$\calT$ of
depth~$\depth(\calT) \geq s$. To this end we first need to slightly
refine the construction process of~$\calT$. Namely, we need to discuss
what the ``related variables of a branch~$\psi$'' are. Instead of
thinking of this as a set of variables we rather want to think of it
as an information set~$J$, as introduced in
\cref{sec:information}. The information set~$J$ is a minimal set that
forces, along with the already collected information set on the
branch~$\tau$, the branch~$\psi$. Once we identified such a set~$J$,
we then query all necessary variables to see whether we agree with~$J$
(along with some further variables).

Recall that we are trying to explain how to construct a
restriction~$\rho^*$ from a restriction~$\rho$ that gives rise to a
decision tree~$\calT$ of large depth. Fix a long
branch~$\tau \in \calT$ and consider the sets~$J_1, J_2, \ldots, J_g$
identified in the different stages of the construction of the long
branch~$\tau$. For this proof overview, let us assume that each~$J_j$
is closed and that the support of these information sets are pairwise
disjoint. Let us stress that this is a slight simplification. Assuming
this holds, note that the union~$J^* = \cup_{i=1}^g J_j$ is also
closed and recall from \cref{sec:information} that all variables
forced by~$(\rho, J^*)$ can be described by a restriction where the
centers in~$\supp(J^*)$ are killed. This defines the
restriction~$\rho^*$: it is the restriction that forces all variables
forced by~$(\rho, J^*)$. Assuming that the support of~$J^*$ is large
we see that~$\rho^*$ has much fewer alive centers. Using
\cref{lem:final-counting} we obtain an upper bound on the failure
probability, assuming that we can easily recover~$\rho$ from~$\rho^*$.

It remains to argue that~$\rho$ can be recovered from from~$\rho^*$
with little extra information. The idea is to remove the set~$J_j$,
starting with~$j=1$, one-by-one from~$\rho^*$. To do this cheaply we
use the decision trees~$T_1, \ldots, T_m$. Recall that the information
set~$J_1$ determines all variables on the first
$1$-branch~$\psi_1$. This implies in particular that~$\rho^*$
traverses the branch~$\psi_1$. Hence identifying~$\psi_1$ is for free:
it is the first $1$-branch in~$T_1, \ldots, T_m$ traversed by~$\rho^*$
(since the set~$J_1$ is pairwise disjoint from all later
sets~$J_j$). Once we identified the branch~$\psi_1$ we want to recover
the first part of the long branch~$\tau$ so that we can repeat this
argument with~$J_2$. As~$\psi_1$ is of length at most~$t$, using
only~$\log t$ bits per variable, we indicate which variables are
different on~$\tau$ from~$J_1$. This lets us cheaply recover~$\tau$
along with the centers killed by~$J_1$. Repeating this argument~$g$
times allows us to recover the restriction~$\rho$.

This completes the proof overview. We allowed ourselves some
simplifications and left out a number of details.

\paragraph{Organization.} The proof of \cref{lemma:switch} spans the
following four sections.
In \cref{sec:ext-can-dec-tree} we define the extended canonical
decision tree~$\calT$ and in the subsequent
\cref{sec:ext-can-dec-tree-prop} we prove some crucial properties of
these decision trees.
In \cref{sec:encode-rho} we explain how \cref{lemma:switch} follows
from the encoding argument as outlined above. The proof of the
encoding lemma is the final part of the proof of the switching lemma
and is given in \cref{sec:encode}.

\subsection{Extended Canonical Decision Trees}
\label{sec:ext-can-dec-tree}

Sample a full restriction
$\sigma = \sigma(\rho, \pi) \sim \calD_k\bigl(\Sigma(n,n')\bigr)$ as
defined in \cref{sec:restrictions} and denote by $\chosen$ the chosen
centers of $\sigma$. Recall that $\rho$ has $(1 \pm 0.01)C \log n$
many alive centers in each sub-square. Let~$T_1, \ldots, T_m$ be
decision trees of depth at most~$t$ querying grid-edge variables of
the~$n \times n$ grid.

We intend to construct the \emph{extended canonical decision tree}
$\calT$ that represents $\bigvee_{i=1}^m \restrict{T_i}{\sigma}$. Note
that, in contrast to the decision trees $T_i$ that query variables of
the unrestricted formula, that is, they query grid-edges, the decision
tree $\calT$ queries variables of the \emph{restricted} formula, that
is, path variables $y_P$ where $P$ is a chosen path (a path connecting
two chosen centers in adjacent sub-squares).

\paragraph{Intuition.} Before formally defining the extended canonical
decision tree let us give an intuitive (and flawed) outline of the
construction. We would like to first construct a decision tree
$\widetilde \calT$ representing~$\bigvee_{i=1}^m T_i$ restricted
by~$\rho$ (instead of~$\sigma$) in the usual manner: proceed in stages
and in each stage extend a branch~$\tau$ of~$\widetilde \calT$ by
querying variables related to the first $1$-branch~$\psi$ in the
trees~$\restrict{T_1}{\rho\tau}, \restrict{T_2}{\rho\tau}, \ldots,
\restrict{T_m}{\rho\tau}$. For the moment, as in
\cref{sec:switch-overview}, we may think of the related variables as
the variables on the branch~$\psi$. Once we have queried all these
variables we check in each newly created leaf of~$\widetilde\calT$
whether we traversed~$\psi$: if so, then the leaf becomes a
$1$-leaf. Otherwise we proceed with the next stage. If no further
$1$-branches are left, then the leaf is labeled~$0$.

Note that such a decision tree~$\widetilde \calT$ not only queries
variables~$y_P$ for chosen paths~$P$ but queries variables associated
with \emph{any} path connecting two alive centers: while~$\calT$ knows
all of~$\sigma$ and hence only needs to query variables associated
with chosen paths, the tree~$\widetilde \calT$ is unaware\footnote{Let
  us remark that the pairing~$\pi$ is a function of the
  restriction~$\rho$. Hence given~$\rho$ the pairing~$\pi$ is known
  and hence this description makes formally little sense. For
  intuition, however, it is a quite insightful view.} of the
pairing~$\pi$ and hence does not know the identity of the chosen
centers. Once we have~$\widetilde \calT$ we would like to
define~$\calT$ (the decision tree representing
$\bigvee_{i=1}^m \restrict{T_i}{\sigma}$) as the restriction
of~$\widetilde \calT$ by~$\pi$, that is,
$\calT = \restrict{\widetilde \calT}{\pi}$.

Instead of analyzing the event that~$\calT$ has a branch of length~$s$
we would like to analyze the event that~$\widetilde \calT$ has a
branch~$\tilde \tau$ of length~$s$ which is pairwise locally
consistent with~$\pi$, that is, the restricted
branch~$\restrict{\tilde \tau}{\pi}$ is a traversable branch
in~$\calT$. If we manage to upper bound the probability of the latter
event, then this upper bound clearly also applies to the former.

The actual construction differs from the above intuitive
description. We do \emph{not} take the detour via the decision
tree~$\widetilde \calT$. Instead we directly define the decision tree
$\calT$ since we want to treat chosen path variables in a slightly
different manner than a path variable whose corresponding path simply
connects alive centers. The formal construction follows.

\paragraph{Setup.} We need to introduce two sets that guide the
construction of the extended canonical decision tree~$\calT$. We start
with~$\calT$ being the empty decision tree.  We extend the decision
tree~$\calT$ in \emph{stages}. In each stage we extend~$\calT$ at some
branch. By the end of a stage each branch~$\tau$ of~$\calT$ is
associated with
\begin{enumerate}
\item a subset of the alive centers~$S = S(\tau, \sigma)$ called the
  \emph{exposed centers}, and
\item an information set~$I = I(\tau, \sigma)$ supported on alive
  centers.
\end{enumerate}
Initially the sets~$S(\emptyset, \sigma)$ and~$I(\emptyset, \sigma)$
are empty. Throughout the construction of~$\calT$ we maintain the
following invariants.
\begin{invariant}\label[invariants]{inv:ext-can-dec-tree}
  Throughout the creation of~$\calT$ the following properties of~$S$
  and~$I$ are maintained.
  \begin{enumerate}
  \item No element is ever removed from~$S$ or~$I$. In other words,
    the sets~$S$ and~$I$ only become larger throughout the creation of
    a branch~$\tau$ of~$\calT$. 
    \label[inv]{inv:monotone}

  \item The information set~$I$ is locally consistent and closed
    on~$S$. \label[inv]{inv:local-consistency}
    
  \item If a center of a connected component of the pairing~$\pi$ is
    in the set of exposed centers~$S$, then the entire component is
    in~$S$. \label[inv]{inv:non-chosen-entire-component}
    
  \item The part of the information set~$I$ on the non-chosen
    centers~$S \setminus \chosen$ is a subset of the connected
    components of the pairing~$\pi$ in the form of a closed
    information set.\label[inv]{item:non-chosen-S}

  \item For every exposed chosen center~$v \in S\cap\chosen$ all the
    variables~$y_P$ incident to~$v$ are queried by the
    branch~$\tau$. The answers to these queries is precisely the
    information in~$I$: $1$-answers are recorded in the form of an
    edge while the $0$-answers are recorded as a non-edge in the
    appropriate direction.
    \label[inv]{item:chosen-S}
  \end{enumerate}
\end{invariant}

Let us stress that information about the pairing~$\pi$ comes from the
restriction~${\sigma(\rho, \pi)}$ and hence in order to maintain
\cref{item:non-chosen-S} we do not need to query a variable
in~$\calT$. This is in contrast to \cref{item:chosen-S}: querying a
variable~$y_P$ associated with a chosen path~$P$ causes a query in the
decision tree~$\calT$.
There is another subtle difference between \cref{item:non-chosen-S}
and \cref{item:chosen-S}: on the non-chosen centers we only have
information pieces in the information set~$I$ that are incident to the
exposed centers~$S$. On the other hand~$I$ may contain information
pieces incident to chosen centers that are not exposed, that is,
information pieces incident to chosen centers that are not
in~$S \cap \chosen$. Finally note that the information set~$I$ never
contains a path between a chosen center and a non-chosen center.

We need one last definition before we can formally define the extended
canonical decision tree~$\calT$. Recall that the construction of the
decision tree~$\calT$ commences in stages. In each stage we extend a
branch~$\tau$ of~$\calT$ by querying variables associated with a
$1$-branch~$\psi$ of a decision tree~$T_i$ pairwise locally consistent
with the branch~$\tau$ and the full restriction~$\sigma$. Hence there
is a unique minimum partial assignment to the variables of~$\psi$
pairwise locally consistent with~$\tau$ and~$\sigma$ to reach the
corresponding leaf. In the following we define the analogue of this
minimum partial assignment in terms of information pieces.

Intuitively a \emph{possible forcing information}~$J$ is a minimal
information set that jointly with the information
set~$I(\tau, \sigma)$ and the partial restriction~$\rho$
forces\footnote{According to \cref{def:force} a variable is forced if
  and only if the associated center is dead or we have information
  pieces in all directions at its associated center.} all variables on
the branch~$\psi$ to take the values given by this partial
assignment. We require some further properties as summarized in the
following definition.

\begin{definition}[possible forcing information]
  \label{def:forcing-info}
  Let~$\psi$ be a branch, let~$I$ be an information set, let~$S$ be a
  set of centers, and denote by~$\sigma(\rho,\pi)$ a full restriction.
  A \emph{possible forcing information} for~$\psi$ is a minimal
  information set~$J$ that satisfies the following.
  \begin{enumerate}

  \item The information sets~$I$ and~$J$ are pairwise locally
    consistent.\label[prop]{item:disjoint}

  \item If an associated center~$u$ of a variable on the branch~$\psi$
    is not in~$S$, then~$J$ is closed at~$u$.
    \label[prop]{item:closed}

  \item All variables on~$\psi$ are forced by~$(\rho, I \cup J)$ such
    that the leaf of~$\psi$ is reached.\label[prop]{item:determined}

  \item The part of the information set~$J$ on the non-chosen centers
    is a subset of the connected components of the pairing~$\pi$ in
    the form of a closed information set.\label[prop]{item:non-chosen-pi}

  \end{enumerate}
\end{definition}

Let us stress that a possible forcing information~$J$ never contains
an edge between a chosen center and a non-chosen center.
Note that a possible forcing information~$J$ may not be unique for a
given branch~$\psi$. If there are several sets as described above,
choose one in a fixed but otherwise arbitrary manner. While the choice
is not essential, we do need to establish that whenever some decision
tree~$T_i$ can still reach a 1-leaf, then there is a possible forcing
information~$J$. We postpone this to the following section (see
\cref{lemma:forced0}) and for now assume that such an information
set~$J$ exists whenever we have a branch~$\psi$ as described. With
these definitions in place we are ready to formally define the
extended canonical decision tree~$\calT$.

\begin{figure}
  \centering
  \includegraphics[page=4]{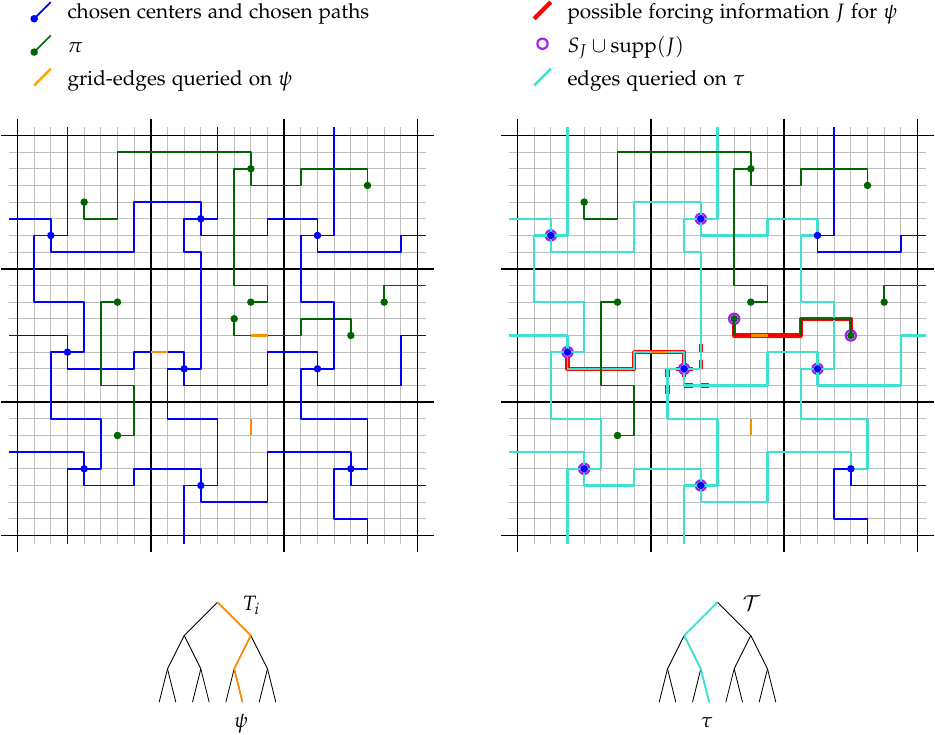}
  \caption{Depiction of the first stage in the construction of the
    extended canonical decision tree~$\calT$. For clarity we omitted
    the non-edges in~$J$ on non-chosen centers.}
  \label{fig:ext-can-dec-tree}
\end{figure}

\paragraph{Construction.} We provide pseudo-code (\cref{alg:cdt}) in
case of ambiguity in the following verbal
description. Initially~$\calT = \emptyset$ is the empty tree. In each
stage we fix a branch~$\tau$ in~$\calT$ and go over the decision
trees~$T_1, T_2, \ldots, T_m$ one by one. Suppose we consider the
decision tree~$T_i$. Let~$\psi$ be the first (in some arbitrary but
fixed order) $1$-branch of~$T_i$ such that
\begin{enumerate}
\item the branch~$\psi$ and the full restriction~$\sigma$ are pairwise
  locally consistent, and
\item the assignment\footnote{As defined in
    \cref{prop:def-consistency-beta} of
    \cref{def:local-consistency-full}.}~$\psi_\sigma$ induced by $\psi$
  on the smaller~$n'\times n'$ grid is locally consistent with $\tau$.
\end{enumerate}
If there is no such branch, then proceed with the decision
tree~$T_{i+1}$. If there is no such branch~$\psi$ for any decision
tree~$T_i$, then label the leaf~$\tau$ of~$\calT$ by $0$ and continue
with a different branch~$\tau'$ of~$\calT$ until all leaves of~$\calT$
are labeled. For the remainder of this section let us assume that
there is such a $1$-branch~$\psi$. Denote by~$J$ a possible forcing
information for the \emph{forceable branch}~$\psi$.

Let~$I_\pi$ be the information set from~$\pi$ incident to non-chosen
centers in~$\supp(J)$ and denote by~$S_J$ all chosen centers at
distance at most 1 from~$\supp(J)\cap\chosen$ with respect to the
smaller $n'\times n'$ grid.

Extend the decision tree~$\calT$ at the leaf~$\tau$ by querying all
variables incident to~$S_J$. Let~$\tau'$ be a newly created extension
of~$\tau$. Denote by~$I_{\tau' \setminus \tau}$ the information set
consisting of the answers to the queries on~$\tau' \setminus \tau$ and
the locally implied answers, that is, $1$-answers are recorded as an
edge and $0$-answers as a non-edge incident to the exposed center.

Update the bookkeeping objects 
\begin{enumerate}
\item $S(\tau', \sigma) =
  S(\tau, \sigma)
  \cup
  S_J
  \cup
  \supp(J)$ and
\item
  $I(\tau', \sigma) =
  I(\tau, \sigma)
  \cup
  I_{\tau'\setminus \tau}
  \cup
  I_\pi$.
\end{enumerate}
Finally check whether the information set~$I(\tau', \sigma)$ traverses
the forceable branch~$\psi$ of~$T_i$. Since all variables on~$\psi$
have their associated center in~$S$ and~$I$ is closed on~$S$ this can
be easily checked. If the forceable branch~$\psi$ is indeed followed,
then label the leaf~$\tau'$ with a $1$. Otherwise, that is, if the
forceable branch is not followed, then proceed with the next stage.

This completes the description of the creation of the extended
canonical decision tree~$\calT$
for~$\bigvee_{i=1}^m\restrict{T_i}{\sigma}$. If in the construction
of~$\calT$ a branch~$\tau$ is of length at least $s$, then we stop the
construction of~$\calT$. In the following we bound the probability of
ever reaching a stage such that the construction is stopped.

It is straightforward to check that the invariants hold after every
completed stage; we only need to recall that we enforce that all the
newly created branches~$\tau'$ are \emph{locally consistent} as
assignments on the smaller grid (see \cref{def:local-consistency-full}
and the discussion thereafter). Hence the information sets created are
locally consistent, that is, all the information
sets~$I(\tau', \sigma)$ are locally consistent. The other invariants
hold by construction.

\subsection{Some Properties of Extended Canonical Decision Trees}
\label{sec:ext-can-dec-tree-prop}

We need to show that the extended canonical decision tree~$\calT$
represents~$\bigvee_{i = 1}^m \restrict{T_i}{\sigma}$. This will be a
direct consequence of the postponed claim that if it is possible to
reach a 1-leaf in some decision tree $T_i$, then there is a possible
forcing information~$J$.

Observe that at any point when forming the extended canonical decision
tree the information~$I$ comes partly from the pairing~$\pi$ and from
queries already done in the decision tree~$\calT$ with
answers~$\tau$. Remember that~$\sigma = \sigma(\rho, \pi)$ includes
all the information from the pairing~$\pi$.

\begin{lemma}
  \label{lemma:forced0}
  If there is a $1$-branch $\psi$ in $\restrict{T_i}{\sigma}$ that is
  pairwise locally consistent with $\tau$, then there is a possible
  forcing information $J$ for $\psi$.
\end{lemma}

\begin{proof}
  Extend~$\psi$ (an assignment to the $n' \times n'$ grid) to an
  assignment $\psi^+ \supseteq \psi$ such that
  \begin{enumerate}
  \item the assignments~$\psi^+$ and~$\tau$ are pairwise locally
    consistent, and
  \item the assignment~$\psi^+$ is complete on all associated
    centers~$u$ of variables on~$\psi$ except on the associated
    centers~$u$ on which~$\tau$ is already complete.
  \end{enumerate}

  According to \cref{lemma:extendone} such an extension $\psi^+$
  exists, assuming that $\lvert\supp(\psi)\rvert \leq n'/32$ and
  $\lvert \supp(\tau)\rvert \leq n'/4$.
  
  Let $\psi_0$ be the branch in~$T_i$ that gives rise to~$\psi$. Note
  that~$\psi_0$ is an assignment to the $n\times n$ grid. Let us
  construct a possible forcing information~$J$ such that~$\psi_0$ is
  followed.

  Add information pieces next to chosen centers as indicated
  by~$\psi^+$. The information pieces next to non-chosen centers are
  the relevant information pieces from $\pi$. 
  
  We claim that the information set~$J$ that contains the
  just-mentioned information pieces is a valid possible forcing
  information for~$\psi_0$ if we ignore the minimality condition. Let
  us check the other properties of \cref{def:forcing-info}. Let $I$ be
  the information set corresponding to $\tau$.

  Note that since~$\tau$ and~$\psi^+$ are locally consistent, so are
  $I$ and $J$. Hence \cref{item:disjoint} follows and
  \cref{item:closed,item:determined,item:non-chosen-pi} follow by
  construction.
  Hence any minimal subset of~$J$ with the above properties
  constitutes a forcing information for~$\psi_0$. This completes the
  proof of the lemma.
\end{proof}

As an immediate corollary we have that the decision tree~$\calT$ is
indeed a legitimate choice
for~$\bigvee_{i=1}^m \restrict{T_i}{\sigma}$.

\begin{corollary}\label{cor:represent}
  The extended canonical decision tree~$\calT$
  represents~$\bigvee_{i=1}^m \restrict{T_i}{\sigma}$.
\end{corollary}

We have three auxiliary lemmas regarding forcing information and the
size of the set of exposed centers~$S$.

\begin{lemma}\label{clm:size-J}
  In each stage at most~$4\lvert\supp(J)\rvert$ centers are added to
  the set of exposed centers~$S$.
\end{lemma}
\begin{proof}
  We add~$\supp(J) \cup S_J$ to the set of exposed centers~$S$. Recall
  that~$S_J$ is the set of chosen centers at distance at most $1$ from
  $\supp(J) \cap \chosen$. It suffices to argue that every chosen
  center in~$\supp(J)$ is adjacent to at least one other chosen center
  in~$\supp(J)$: if this holds, then for every chosen center
  in~$\supp(J)$ we add at most~$3$ chosen centers
  to~$S_J \setminus \supp(J)$.
  
  By minimality of~$J$ each non-exposed chosen
  center~$u \in \supp(J) \cap \chosen$ is either an associated center
  of a variable on~$\psi$ or is adjacent to such an associated
  center. In the latter case we are done. If~$u$ is an associated
  center, then the forcing information~$J$ contains information pieces
  in all directions incident to~$u$ and hence, since~$J$ is locally
  consistent, there is at least one edge incident to~$u$
  in~$J$. The claim follows.
\end{proof}

\begin{corollary}\label{clm:size-s}
  In each stage at most $16t$ centers are added to the set of exposed
  centers $S$.
\end{corollary}

\begin{proof}
  A forceable branch~$\psi$ is of length at most~$t$ as the decision
  trees~$T_i$ are of depth at most~$t$. For each variable~$x_e$
  on~$\psi$ there are at most~$4$ chosen centers in $\supp(J)$ if the
  associated center of~$x_e$ is chosen and at most~$4$ non-chosen
  centers if the associated center is
  non-chosen. Hence~$\lvert\supp(J)\rvert \leq 4t$. The statement
  follows by \cref{clm:size-J}.
\end{proof}

\begin{lemma}\label{lem:J-disjoint}
  Consider a branch~$\tau \in \calT$ along with the possible forcing
  information~$J_1, J_2, \ldots$ used in the different stages of the
  construction of $\tau$. For~$j \neq j'$ it holds that~$\supp(J_j)$
  and~$\supp(J_{j'})$ are disjoint.
\end{lemma}

\begin{proof}
  Denote by~$1, 2, \ldots$ the stages used in the construction of the
  branch~$\tau$ and let~$S_{j-1}^*$ be the set of exposed centers~$S$
  at the beginning of stage~$j$. Suppose $j' < j$.

  By minimality of~$J_j$ it holds that~$S_{j-1}^* \cap \supp(J_j)$
  only contains chosen centers that are adjacent to a non-exposed
  chosen center.
  Since in stage~$j'$ all centers at distance at most $1$
  from~$\supp(J_{j'}) \cap \chosen$ are added to the set of exposed
  centers it holds that (1) $\supp(J_{j'}) \subseteq S^*_{j-1}$ and
  that (2) all chosen centers adjacent to~$\supp(J_{j'}) \cap \chosen$
  are in~$S^*_{j-1}$.  Hence~$\supp(J_j)$ cannot contain a center
  from~$\supp(J_{j'})$. The claim follows.
\end{proof}

\subsection[Encoding \texorpdfstring{$\rho$}{ρ}]{Encoding $\bm{\rho}$}
\label{sec:encode-rho}

We want to bound the number of restrictions~$\rho$ that give rise to
an extended canonical decision tree~$\calT$ of
depth~$\depth(\calT) \geq s$. Fix such a restriction~$\rho$ along with
the extended canonical decision tree~$\calT$ and a branch~$\tau$ of
length at least~$s$. Denote by $1, 2, \ldots$ the stages in which the
branch~$\tau$ was constructed.

Since the branch~$\tau$ is of length at least~$s$ there is a first
stage~$g$ such that by the end of stage~$g$ at least~$s/4$ centers are
exposed: only variables incident to exposed chosen centers are queried
and each exposed center causes at most $4$ queries on the
branch~$\tau$. In other words, if we let~$\tau_g \subseteq \tau$ be
the branch constructed by the end of stage~$g$, then the set of
exposed centers~$S^*_g = S(\tau_g, \sigma)$ is for the first time of
size~$\lvert S^*_g \rvert \geq s/4$. We analyze the event of ever
reaching such a stage $g$.

Note that~$\lvert S^*_g\rvert < s/4 + 16t$ by \cref{clm:size-s}
and~$g \leq s/4$ as in each stage at least one center is added to the
set of exposed centers~$S$. For~$j \in [g]$ we let the forceable
branch of stage~$j$ in the decision tree~$T_{i_j}$ be denoted
by~$\psi_j$, let~$J_j$ be the corresponding possible forcing
information and let~$\tau_j \subseteq \tau_g$ be the branch in~$\calT$
created by the end of stage~$j$. Let~$S^*_j = S(\tau_j, \sigma)$,
denote the information set added to~$I$ in stage~$j$ by~$I_j$, and
let~$I_j^* = I(\tau_j, \sigma)$, or
equivalently~$I_j^* = \cup_{i=1}^{j} I_i$, be the information set
gathered during the first~$j$ stages.

Let $\cl_j \supseteq J_j$ be the information set obtained from $J_j$
by adding non-edge information pieces in direction~$\delta$ incident
to chosen centers~$v \in \supp(J_j) \cap \chosen$ if
\begin{enumerate}
\item the information set~$J_j$ has an odd number of edge information
  pieces incident to~$v$, and
\item there is no information piece in~$J_j$ incident to~$v$ in
  direction~$\delta$.
\end{enumerate}
Note that since~$\supp(J_j) = \supp(\cl_j)$ it holds by
\cref{lem:J-disjoint} that~$\supp(\cl_j)$ and~$\supp(\cl_{j'})$ are
disjoint for~$j \neq j'$. This implies in particular that the
information pieces in~$\cl_j$ and~$\cl_{j'}$ are not in direct
contradiction.

Let~$\cl^* = \bigcup_{j=1}^{g} \cl_j $ and define~$\rho^*$ to be the
restriction as defined in \cref{eq:rho-star} that forces the same
variables as~$(\rho, \cl^*)$ does. Observe that alive centers~$v$
of~$\rho$ with an odd number of edge information pieces incident in
some information set~$J_j$ are now dead in~$\rho^*$. The centers alive
in~$\rho$ but dead in~$\rho^*$ are the \emph{disappearing centers}.

Let~$a_j$ be the number of associated centers of variables on~$\psi_j$
that are also in~$\supp(J_j) \setminus S^*_{j-1}$. Note that by
\cref{item:closed} of \cref{def:forcing-info} all these associated
centers are disappearing centers. Denote by~$b_j$ the number of
additionally disappearing centers in~$\cl_j$, and
let~$a = \sum_{j=1}^g a_j$, $b = \sum_{j=1}^g b_j$.
We have the following relation between these parameters.

\begin{lemma}\label{lem:abc}
  It holds that~$b \leq 3a$, and~$s/64 \leq a$.
\end{lemma}
\begin{proof}
  For each stage $j\in [g]$ it holds
  that~$\supp(J_j)$ contains at
  most~$4a_j$ centers: each chosen center
  $u\in\supp(J_j)$ is either an associated center of a variable
  on~$\psi_j$ and in~$u \in\supp(J_j) \setminus S^*_{j-1}$ or
  $u$ is adjacent to such an associated center. For each of the former
  there are at most 3 of the latter
  since~$J_j$ is locally consistent. Similarly, since
  $\pi$ has components of size 4, each non-chosen associated center
  causes at most 3 other non-chosen centers to be included in
  $\supp(J_j)$.
  Since $a_j + b_j \leq \lvert \supp(\cl_j) \rvert = \lvert \supp(J_j)
  \rvert \leq 4a_j$ we obtain the desired inequality $b \leq 3a$.

  It remains to establish the inequality $s/64 \leq a$.  By
  \cref{clm:size-J} and the just-established inequality
  $\lvert \supp(J_j) \rvert \leq 4a_j$ we get that
  $\lvert S^*_g\rvert \leq 4\sum_{i=1}^g \lvert \supp(J_j) \rvert \leq
  16a$. Since~$\lvert S^*_g \rvert \geq s/4$ we obtain the desired
  inequality.
\end{proof}

In the following section we prove the next lemma stating that a
restriction $\rho$ that causes the extended canonical decision tree to
have a path of length at least $s$ can be encoded using few bits,
given $\rho^*$ and $T_1, \ldots, T_m$. Put differently, the mapping
from $\rho$ to $\rho^*$ can be inverted with a bit of extra
information. Recall that~$\Delta = \Theta(n/n')$ is the number of
centers in each sub-square.

\begin{lemma}\label{lemma:encode}
  There is a constant $A > 0$ such that the following holds. Suppose
  we are given decision trees $T_1, \ldots, T_m$ of depth at most $t$
  and $\rho^*$.
  Then
  \begin{align*}
    a \log t + b \log \Delta + s \cdot A
  \end{align*}
  many bits suffice to encode $\rho$.
\end{lemma}

Before diving into the proof of \cref{lemma:encode} let us verify that
\cref{lemma:switch} indeed follows.

\begin{proof}[Proof of \cref{lemma:switch}]
  We analyze the probability that a $\rho$ chosen uniformly
  from~$\spaceRhoReg(k,n,n')$ gives rise to an extended canonical
  decision tree of length at least~$s$. Let $m=n'^2\Delta $ be the
  total number of centers and recall that the odd integer~$k$ is the
  total number of alive centers.
  
  According to \cref{lemma:encode}, for some absolute constant~$A$,
  the number of restrictions $\rho$ that give rise to an extended
  canonical decision tree of depth at least~$s$ can be upper bounded
  by the number of ways to choose a restriction~$\rho^*$ with~$k-a-b$
  alive centers times~$t^a \Delta^{b} A^{s}$. Using
  \cref{lem:final-counting} we can bound the probability of sampling
  such a~$\rho$ by
  \begin{align}\label{eq:final1}
    \sum_{a,b}
    \frac{
    \lvert
    \spaceRho(k-a-b,n,n')
    \rvert
    \cdot
    t^a \Delta^{b} A^{s}
    }
    {
    \lvert
    \spaceRhoReg(k,n,n')
    \rvert
    }
    &\leq
      \sum_{a,b}
      \left(
      \frac{A_0\log n'}{\Delta}
      \right)^{a+b}
      \cdot
      t^a \Delta^{b} A^{s}\\
    &\leq
      \sum_{a}
      A_1^s
      \left(
      \frac{t\log n'}{\Delta}
      \right)^a
      \cdot
      \sum_{b} \log^b n'\\
    &\leq
      \sum_{a}
      A_2^s\left(\frac{t \log^4 n'}{\Delta}\right)^{a}
      \eqcomma
      \label{eq:final2}
  \end{align}
  for appropriate constants~$A_0, A_1$ and~$A_2$. The final inequality
  relies on the bound~$b \leq 2a$ from \cref{lem:abc}.  As
  \cref{lem:abc} further guarantees that~$a \geq s/64$ and since the
  sum in \cref{eq:final2} is a geometric series the claimed bound on
  the probability of the extended canonical decision tree reaching
  depth at least~$s$ follows.
\end{proof}

\subsection{Proof of \cref{lemma:encode}}
\label{sec:encode}

\paragraph{Outline.} On a very high level, we want to remove the
information set~$\cl^*$ from the partial restriction~$\rho^*$. We
commence in stages. In each stage we remove a single information
set~$\cl_j$ from~$\rho^*$ by utilizing the shallow decision
trees~$T_1, \ldots, T_m$ and some extra information. We need some
further notation and a simple observation to give a detailed proof
outline.

For convenience let $I_0^* = \emptyset$ and for $i > j$ denote
by~$I^{i}_j \subseteq I_j$ the information set obtained from~$I_j$ by
removing any information piece that occurs (identically) in some
possible forcing information~$J_{j'}$ for $j' \geq i$. In other words
we let~$I_j^i = I_j \setminus \bigl(\bigcup_{j' = i}^{g}J_{j'}\bigr)$
and hence it holds
that~$I_j^{j+1} \subseteq I_j^{j+2} \subseteq \cdots \subseteq
I_j^{g+1} = I_j$. Let $I^{*-}_{j-1} = \bigcup_{i=1}^{j-1} I_i^{j}$. By
the previous observation we have that ~$I^{*-}_{g} = I^*_g$.

For $i < j$ denote by~$\cl_j^i \subseteq \cl_j$ the information set
obtained from~$\cl_j$ by removing any information piece in direct
contradiction with~$I^*_{i}$: remove information
pieces~$(v, \delta, \bot)$ if and only if~$I^*_{i}$ contains an edge
incident to~$v$ in direction~$\delta$. Since~$J_j \subseteq \cl_j$ is
pairwise locally consistent with $I^*_{j-1}$ it holds
that~$J_j \subseteq \cl_j^{j-1}\subseteq \cl_j^{j-2} \subseteq \cdots
\subseteq \cl_j^{0} = \cl_j$.
Let~$\cl_{\geq j}^{*-} = \bigcup_{i=j}^g \cl_i^{j-1}$ and note
that~$K_{\geq 1}^{*-} = \bigcup_{i=1}^g \cl_i^{0} = \bigcup_{i=1}^g
\cl_i = K^*$.

Denote by~$\rho_{j}^*$ the restriction obtained from composing $\rho$
with the information~$I^{*-}_{j-1} \cup \cl_{\ge j}^{*-}$ as done in
\cref{eq:rho-star}. Note that $\rho^*_{j}$ forces the same variables
as~$(\rho, I^{*-}_{j-1} \cup \cl_{\ge j}^{*-})$ forces, observe
that~$\rho_1^* = \rho^*$, and note that~$\rho_{g+1}^*$ forces the same
variables as~$(\rho, I^*_g)$ forces.

\begin{lemma}
  \label{lem:traverse-psi}
  The restriction $\rho^*_j$ traverses the forceable branch~$\psi_j$
  of stage $j$.
\end{lemma}
\begin{proof}
  Recall from \cref{def:forcing-info}, \cref{item:determined}, that
  $(\rho, I_{j-1}^* \cup J_j)$ forces all the variables on the
  forceable branch $\psi_j$ of stage $j$ such that $\psi_j$ is
  traversed. As the information set~$\cl_j^{j-1}$ extends the forcing
  information~$J_j$ and is not in direct contradiction
  with~$I^*_{j-1}$ we observe that
  also~$(\rho, I_{j-1}^* \cup \cl^{j-1}_j)$ traverses $\psi_j$.
  Furthermore, since ~$\cl_{\geq j}^{*-}$ extends $\cl_j^{j-1}$ and is
  also not in direct contradiction with~$I_{j-1}^*$, we conclude
  that~$(\rho, I_{j-1}^* \cup \cl_{\geq j}^{*-})$,
  equivalently~$\rho^*_j$, traverses the forceable branch~$\psi_j$.
\end{proof}

\Cref{lem:traverse-psi} allows us to pursue the following high-level
plan. We proceed in stages~$j = 1, 2, \ldots, g$. At the beginning of
each stage~$j$ we assume that we know the restriction~$\rho^*_{j}$ and
the information set~$I_{j-1}^{*-}$. Since~$I_0^{*-} = \emptyset$
and~$\rho_{1}^*= \rho^*$ we have the necessary information to start
with stage~$j=1$. Furthermore, if we can complete these $g$ stages we
obtain the restriction~$\rho^*_{g+1}$ from which, along
with~$I^{*-}_g = I^*_g$, we can recover the sought-after
restriction~$\rho$: since~$\rho^*_{g+1}$ forces the same variables
as~$(\rho, I^*_g)$ forces we can ``remove''~$I^*_g$
from~$\rho^*_{g+1}$ to obtain~$\rho$.

Let us consider a stage $j \in [g]$. By \cref{lem:traverse-psi} the
restriction~$\rho^*_{j}$ traverses the forceable branch~$\psi_j$. Let
us assume, for now, that~$\psi_j$ is the first $1$-branch
traversed. This assumption allows us to identify~$\psi_j$ for
free. Once identified we can use the branch~$\psi_j$ to cheaply
recover a good fraction of the support of the forcing
information~$J_j$: since the branch~$\psi_j$ is of length at most~$t$
we can point out the variables on the branch~$\psi_j$ forced by~$J_j$
at cost~$\log t$ each. From these variables we can recover their
unique associated centers for \emph{free}. Each of these associated
centers~$u$ is a disappearing center: by \cref{item:closed} of
\cref{def:forcing-info} the set~$J_j$ is closed at~$u$ and since~$J_j$
is locally consistent (\cref{item:disjoint} of
\cref{def:forcing-info}) it has an odd number of edges incident
to~$u$.

To find the other centers in~$\supp(J_j)$ and to recover the structure
of~$J_j$, i.e., whether there are edges or non-edges in-between
centers of~$\supp(J_j)$, we rely on external information: we read
$\log \Delta$ bits per disappearing center and a constant number of
bits per potential edge. Once we have recovered~$J_j$ it is
straightforward to obtain~$\cl^{j-1}_j$. ``Remove''~$\cl^{j-1}_j$
from~$\rho_{j}^*$ and add all information pieces from~$J_j$
to~$I^{*-}_{j-1}$ that are common to~$J_j$ and~$I^*_{j-1}$. By adding
all these information pieces to~$I^{*-}_{j-1}$ we obtain the
information set~$\bigcup_{i=1}^{j-1} I_i^{j+1}$. Before we can proceed
to stage~$j+1$ we need to recover~$I^{j+1}_j$ so that we can create
the information set~$I^{*-}_{j} = \bigcup_{i=1}^{j} I_i^{j+1}$.

The support of the information set~$I_j$ consists of the
centers~$S_j = S^*_j \setminus S^*_{j-1}$ added to the set of exposed
centers and some further centers at distance $1$ from these newly
exposed centers~$S_j$. Most of the centers in~$S_j$ are readily
identified as follows. Recall that the set~$S_j$ consists
of~$\supp(J_j)$ along with the chosen centers at distance at most 1
from~$\supp(J_j) \cap \chosen$. For each chosen center in~$S_j$ at
distance 1 from~$\supp(J_j) \cap \chosen$ we read one bit of extra
information to determine whether it has disappeared (since these may
be in the support of some forcing information~$\cl_{j'}$
for~$j' > j$). If it has not disappeared, then it is readily verified
as the alive center with the lowest numbered row. Otherwise we can
afford to read~$\log \Delta$ extra information per such center to
identify it. This identifies the exposed centers in the support
of~$I_j$.

Recall that there may be some edge information pieces in~$I_j$ between
an exposed chosen center~$u \in S_j$ and a chosen
center~$v \not\in S^*_j$ that is not exposed. Either
\begin{enumerate}
\item the edge $\set{u,v}$ is shared with a forcing
  information~$J_{j'}$ for $j' > j$ and is thus not in $I_j^{j+1}$,
  or\label[case]{case:shared}\item the information piece~$\set{u,v}$ contradicts some~$\cl_{j'}$
  for $j' > j$, or\label[case]{case:contradict}\item the chosen center~$v$ is in no support~$\supp(\cl_{j'})$ for
  $j' > j$ and is thus alive.\label[case]{case:alive}
\end{enumerate}
In \cref{case:shared} we do not have to identify the chosen center~$v$
since we will recover it in stage~$j'$.
In \cref{case:alive} the chosen center~$v$ is identified for free
since it is the alive center with the lowest numbered row.
For \cref{case:contradict} note that since~$J_{j'}$ and~$I_j$ are
locally consistent (\cref{item:disjoint} of \cref{def:forcing-info}),
the contradiction is due to an information piece
in~$K_{j'} \setminus J_{j'}$. This implies that the center~$v$ is
\emph{not} an associated center of stage~$j'$: by \cref{item:closed}
of \cref{def:forcing-info} the forcing information~$J_{j'}$ has
information pieces in all directions incident to such centers. We may
thus identify~$v$ by reading~$\log \Delta$ bits of extra information
to then ``fix'' the contradiction by removing the non-edge and adding
the edge~$\set{u,v}$.

All that remains is to recover the structure of~$I^{j+1}_j$. We read a
constant amount of extra information per potential edge. This
identifies~$I^{j+1}_j$. Proceed with stage $j+1$.

Let us tally the amount of extra information read. We have~$a \log t$
bits per disappearing center that is also an associated center of a
variable on a forceable branch, for other disappearing centers we
pay~$b \log \Delta$, and finally for the structure of the different
information sets we need another~$A \lvert S^*_g\rvert$ bits for some
constant $A$. Thus in total, as claimed, we need at
most~$a \log t + b \log \Delta + A \lvert S^*_g \rvert$ bits to
recover~$\rho$. This completes the proof overview.

\paragraph{Setup.} Unfortunately there are some complications. Recall
that the forceable branch~$\psi_j$ is defined to be the next
\emph{pairwise locally consistent} branch with~$\tau_{j-1}$ and
$\sigma$ in stage~$j$ of the construction of the decision
tree~$\calT$. In particular, the branch~$\psi_j$ along with
$\tau_{j-1}$ needs to be pairwise locally consistent as assignments on
the smaller $n' \times n'$ grid. At this point it is not clear how to
determine whether a given branch satisfies this property. Hence the
first branch traversed by~$\rho^*_j$ is not necessarily the forceable
branch~$\psi_j$ of stage $j$. We need a way to cheaply tell that a
given branch is not the forceable branch. To this end\footnote{We also
  use signatures to ensure that only branches forced by~$J_j$ are
  considered and not branches forced by information pieces
  in~$K^{j-1}_j\setminus J_j$.} we introduce signatures.

\begin{definition}[signature]\label{def:signature}
  Let~$v$ be a center in the support of~$J_j$. The \emph{signature} of
  such a center~$v$ consists of~$9$ bits.
  \begin{enumerate}
  \item The first bit is~$1$ if and only if~$v$ is a chosen center.
  \item The following four bits indicate in what directions~$J_j$ has
    information pieces incident to~$v$.
  \item The final four bits indicate for each direction whether there
    is an edge in~$J_j$ incident to~$v$.
  \end{enumerate}
\end{definition}

\begin{remark}
  Note that the stage~$j$, although mentioned in the definition, is
  \emph{not} part of the signature (as it was in \cite{jhtseitin}).
  This change is mandated by our desire to get a tighter bound which
  requires that the signature of a single center is of constant size.
  Furthermore note that the signature of~$v$ does not include the
  identity of~$v$. In the following we only read signatures in
  combination with a fixed center~$v$; we assume that the signatures
  are ordered on the auxiliary information as we process these
  centers.
\end{remark}

Since the information sets~$J_j$ and~$J_{j'}$ are disjoint (see
\cref{lem:J-disjoint}) and since~$\supp(J_j) = \supp(K_j)$ it holds
that each disappearing center in the support of~$\cl^*$ has a unique
signature. Also, since~$(\rho, I_{j-1}^* \cup J_j)$ forces all
variables on the forceable branch~$\psi_j$, every variable on~$\psi_j$
that is not forced by~$(\rho, I_{j-1}^*)$ has an associated center
in~$J_j$ with a signature. Finally note that only nodes in the support
of~$\cl^*$ require signatures. As such we can afford to read all the
signatures: at most~$9\lvert \supp(\cl^*) \rvert \leq 200s$ bits need
to be read.

Note that a chosen center~$v$ along with its signature~$\sign$ defines
a partial assignment to the incident path variables: the first set of
four bits of~$\sign$ indicates the direction and the final four bits
indicate the assignment to the incident path-variables.

The idea is that signatures indicate how the next forceable branch
should look like: they indicate where the branch is supposed to have
variables and what the corresponding assignment on the smaller grid
looks like. If a $1$-branch~$\psi$ contradicts this information, then
we say that~$\psi$ is in \emph{conflict} with these signatures. The
formal definition follows.

\begin{definition}[conflict]\label{def:signature-conflict}
  Let~$I$ be an information set, denote by~$\psi$ a branch, and
  let~$E$ be a set of tuples~$(v, \sign)$ each consisting of a
  center~$v$ along with its signature~$\sign$. Denote
  by~$E_\psi \subseteq E$ the subset of tuples~$(v,\sign)$ such
  that~$v$ is a chosen center (i.e., the first bit of~$\sign$ is~$1$)
  and there is a variable~$x_e$ on~$\psi$ such that~$v$ is the
  associated center of~$x_e$. The set~$E$ is in \emph{conflict}
  with~$\psi$ and~$I$ if and only if either
  \begin{enumerate}
  \item there is an associated center~$v$ of a variable on~$\psi$ such
    that the partial assignment induced by~$I$ along with the
    signature of~$v$ is not defined in all directions incident to~$v$,
    or
    
  \item the partial assignment on chosen path variables obtained
    from~$I$ jointly with the assignments defined by the tuples
    in~$E_\psi$ is not pairwise locally consistent as assignments on
    the smaller grid.
  \end{enumerate}
\end{definition}

\paragraph{Reconstruction.} Let us explain how signatures are used to
recover the forceable branch~$\psi_j$ of stage~$j$. See
\cref{alg:reconstruct} for pseudo-code of the procedure described in
the following.
Throughout the procedure we maintain the following objects. A
counter~$j= 1, 2, \ldots, g$ of the current stage to be reconstructed,
the restriction~$\rho_{j}^{*}$, the information set~$I_{j-1}^{*-}$,
the exposed centers~$S^*_{j-1}$, and a set~$E$ of (prematurely
identified) disappearing centers along with their
signatures. Initially we set $j=1$, $\rho_{1}^* = \rho^*$, and
$S^*_0 = I^{*-}_0 = E = \emptyset$. We proceed as follows.
\begin{enumerate}
\item Find the next $1$-branch~$\psi$ traversed by the
  restriction~$\rho_{j}^{*}$.\label[step]{step-1}

\item If~$\psi$ and~$I_{j-1}^{*-}$ is in conflict with~$E$, then go to
  \cref{step-1}.

\item Read a bit~$b$ of extra information to determine whether there
  is a disappearing center that is the associated center of a variable
  on~$\psi$.\label[step]{step-3}

\item If $b=1$, then read an integer~$i$ of magnitude at
  most~$t$. This identifies the associated center~$v$ of the~$i$th
  variable on~$\psi$ as a disappearing center. Read the
  signature~$\sign$ of~$v$ and add~$(v, \sign)$ to~$E$. If~$E$ is in
  conflict with~$\psi$ and~$I_{j-1}^{*-}$, then go to
  \cref{step-1}. Otherwise go to \cref{step-3}.

\item If $b=0$, then~$\psi$ is the forceable branch of
  stage~$j$. Recover~$J_j$, the information set~$\cl_j^{j-1}$ as well
  as~$I^{j+1}_j$ essentially as discussed in the proof overview
  (details provided below). Update $\rho_{j}^{*}$ to $\rho_{j+1}^{*}$,
  $I_{j-1}^{*-}$ to $I_j^{*-}$ and $S^*_{j-1}$ to $S^*_j$, remove all
  associated centers of $\psi$ from $E$, and set $j=j+1$. Ensure
  that~$E$ contains all the signatures~$(v,\sign)$ of chosen
  centers~$v \in S^*_{j-1} \cap \supp(\cl_{\ge j}^{*-})$. If
  $\lvert S^*_{j-1}\rvert \ge s/4$, then terminate. Otherwise go to
  \cref{step-1}.\label[step]{step-5}
\end{enumerate}
This completes the description of the procedure used to recover the
forceable branch of stage~$j$. 

The following lemma guarantees that the above procedure is correct,
that is, that it identifies the forceable branch of stage~$j$.

\begin{lemma}\label{lem:found-force}
  Let~$E$ be the set of all tuples~$(v, \sign)$ that consist of a
  disappearing center~$v \in \supp(\cl_{\ge j}^{*-})$ or an exposed
  chosen center in~$v \in S^*_{j-1} \cap \supp(\cl_{\ge j}^{*-})$
  along with their signature~$\sign$. If~$\psi$ is the first
  $1$-branch traversed by~$\rho_{j}^*$ such that~$E$ is \emph{not} in
  conflict with~$I^{*-}_{j-1}$ and~$\psi$, then~$\psi$ is the
  forceable branch~$\psi_j$ of stage~$j$.
\end{lemma}

\begin{proof}We need to establish that $E$ is in conflict with $I^{*-}_{j-1}$ and
  all $1$-branches $\psi$ occurring before $\psi_j$. Towards
  contradiction suppose otherwise: let $\psi$ be a $1$-branch before
  $\psi_j$ such that
  \begin{enumerate}
  \item the restriction~$\rho_{j}^*$ forces all variables on~$\psi$
    such that the respective leave is reached, and
  \item the set~$E$ is not in conflict with~$I^{*-}_{j-1}$ and~$\psi$.
  \end{enumerate}
  Note that also the restriction defined by composing~$\rho$
  with~$I^*_{j-1} \cup \bigcup_{j'=j}^{g} J_{j'}$ would
  traverse~$\psi$: a center~$u$ that does not have information pieces
  in all directions incident
  in~$I^*_{j-1} \cup \bigcup_{j'=j}^{g} J_{j'}$ does not force any
  incident variables. In this case the assignment induced by
  $I^{*-}_{j-1}$ along with the signature~$\sign$ of~$u$ is not
  defined in all directions incident to~$u$. Hence since~$E$ is not in
  conflict with~$I^{*-}_{j-1}$ and~$\psi$ it holds that all variables
  forced by $\cl_{\ge j}^{*-}$ on~$\psi$ are incident to
  centers~$u\in\supp(J_{j'})$ for~$j' \geq j$ with information pieces
  in all directions incident to~$u$ present in~$I_{j-1}^*\cup J_{j'}$.
  
  Let us construct a possible forcing information $J'_j$ that could
  have been used in stage $j$ of the construction of the extended
  canonical decision tree to force the branch $\psi$. On the
  non-chosen centers the set~$J'_j$ contains the pieces of~$\pi$
  needed to force all variables on~$\psi$. On the chosen centers the
  set~$J'_j$ consists of information pieces as given by the partial
  assignments defined by signatures~$(v,\sign) \in E$ such that there
  is a variable on~$\psi$ whose associated center is~$v$. These
  information pieces are pairwise locally consistent
  with~$I_{j-1}^{*-}$ as~$E$ is not in conflict with~$I^{*-}_{j-1}$
  and~$\psi$. Furthermore, these force the input to traverse~$\psi$ as
  these information pieces are the same as used
  in~$I^*_{j-1} \cup \bigcup_{j'=j}^{g} J_{j'}$.
\end{proof}

\Cref{lem:found-force} shows that the above procedure, once it reaches
\cref{step-5}, indeed identifies the forceable branch~$\psi_j$ of
stage~$j$. It remains to argue that the information
sets~$J_j, K_{j}^{j-1}$ and~$I^{j+1}_j$ can be recovered. We closely
follow the argument presented in the proof outline and there is
pseudo-code in \cref{sec:algo-switch} in case of ambiguity in the
following verbal discussion.

We are given the restriction~$\rho^*_j$, the information
set~$I_{j-1}^{*-}$ and all the signatures of the disappearing centers
that are also associated centers of variables on the forceable
branch~$\psi_j$ of stage~$j$. We need to construct the
restriction~$\rho^*_{j+1}$ and the information set~$I^{*-}_j$.

We start with the reconstruction of the forcing information~$J_j$,
then explain how to obtain~$\cl^{j-1}_j$ from~$J_j$ and finally argue
that we can recover~$I^{j+1}_j$ from~$J_j$.

The unique associated centers of the branch~$\psi_j$ identify a good
part of the support of~$J_j$. For each associated
center~$u \in \supp(J_j)$ that is also \emph{chosen} we read the up to
three incident centers used to make~$J_j$ closed at~$u$. Each is read
at cost~$\log \Delta$ unless it has already been identified. If it has
already been identified, then it is the center in the appropriate
sub-square whose first bit in the signature is~$1$; we can identify it
at cost at most~$\log t \leq \log \Delta$. Reading one bit per
information piece on chosen centers in~$J_j$ we obtain the structure
of $J_j$ on the chosen centers.

This identifies all of~$J_j$ on the chosen centers. On the non-chosen
centers we may need to complete some connected components from the
pairing~$\pi$. We encode the structure of the connected component
using a constant number of bits. Each node in such a connected
component can again be recovered at cost at most~$\log \Delta$. This
identifies all of~$J_j$.

We obtain~$\cl^{j-1}_j$ from~$J_j$ by adding non-edge information
pieces incident to a center~$u \in \supp(J_j)$ in direction~$\delta$
if and only if
\begin{enumerate}
\item the center~$u$ has an odd number of edges incident in~$J_j$, and
\item there are no information pieces in~$I^{*-}_{j-1} \cup J_j$ in
  direction~$\delta$ from~$u$.
\end{enumerate}

``Undo'' the information set~$\cl^{j-1}_j$ from~$\rho_{j}^*$ by
flipping the assignment along all the edges in~$\cl^{j-1}_j$ and add
all information pieces from~$J_j$ to~$I^{*-}_{j-1}$ that are common
between~$J_j$ and~$I^*_{j-1}$. This results in the information
set~$\bigcup_{i=1}^{j-1} I_i^{j+1}$. It remains to
recover~$I^{j+1}_j$.

Recall that the support of the information set~$I_j$ consists of the
centers~$S_j = S^*_j \setminus S^*_{j-1}$ added to the set of exposed
centers and some further chosen centers at distance $1$
from~$S_j$. The non-chosen centers in~$S_j$ are readily identified:
these are~$\supp(J_j) \setminus \chosen$. It remains to identify the
chosen centers in~$S_j$ at distance precisely~$1$
from~$\supp(J_j) \cap \chosen$. Check~$E$ whether any of the remaining
chosen centers has already been identified (we know the relevant
sub-squares; check whether there is a tuple~$(v,\sign)\in E$ such
that~$v$ is in one of these sub-squares and the first bit of~$\sign$
is~$1$). These are identified at cost~$\log t \leq \log \Delta$. For
each of the remaining sub-squares we read~$1$ bit of extra information
to determine if the chosen center is alive. If so, then it is the
alive center with the lowest numbered row. Otherwise we
read~$\log \Delta$ extra information to identify it. Note that since
such a center was exposed in stage~$j$ it cannot be a disappearing
associated center of some stage~$j' > j$ and we can thus afford to
read these~$\log \Delta$ bits. This identifies the exposed centers in
the support~$\supp(I_j)$.

Recall that there may be some edge information pieces in~$I_j$ between
an exposed chosen center~$u \in S_j$ and a chosen
center~$v \not\in S^*_j$ that is not exposed. Either
\begin{enumerate}
\item the edge $\set{u,v}$ is shared with a forcing
  information~$J_{j'}$ for $j' > j$ and is thus not in $I_j^{j+1}$,
  or\label[case]{case:shared-1}\item the information piece~$\set{u,v}$ contradicts some~$\cl_{j'}$
  for $j' > j$, or\label[case]{case:contradict-1}\item the chosen center~$v$ is in no support~$\supp(\cl_{j'})$ for
  $j' > j$ and is thus alive.\label[case]{case:alive-1}
\end{enumerate}
In \cref{case:shared-1} we do not have to identify the chosen
center~$v$ since we will recover it in stage $j'$. In
\cref{case:alive-1} the chosen center~$v$ is identified for free since
it is the alive center with the lowest numbered row. In
\cref{case:contradict-1} we identify~$v$ at cost~$\log \Delta$ (unless
it has already been found and is in~$E$). Since in
\cref{case:contradict-1} the chosen center~$v$ does not have
information pieces incident in all directions in~$J_j$ it is not an
associated center (by \cref{item:closed} of \cref{def:forcing-info})
and we can hence afford to pay~$\log \Delta$ bits for its
identification.

It remains to recover the structure of~$I^{j+1}_j$. We read a constant
amount of extra information per potential edge. This
identifies~$I^{j+1}_j$.
This concludes the discussion how one recovers~$J_j$
and~$I^{*-}_j$. Let us tally the external information needed.

Recall that $a_j$ is the number of disappearing centers
outside~$S^*_{j-1}$ that are an associated center of a variable of
$\psi_j$, $b_j$ is the number of other disappearing centers
of~$\cl_j$, and we let~$a = \sum_{j=1}^g a_j$
and~$b= \sum_{j=1}^g b_j$. The following summarizes the amount of
external information needed.

\begin{itemize}
\item The disappearing centers discovered as an associated center of a
  forceable branch contribute~$a \log t$ bits.

\item The other disappearing centers contribute at
  most~$b \log \Delta$ bits.

\item For each center in~$\supp(\cl^*)$ we may have to read the
  signature. These are at most~$9 \lvert S^*_g \rvert$ bits.
  
\item All other centers are discovered at constant cost; we only
  read~$A_1 \lvert S^*_g \rvert$ bits for some constant $A_1$.

\item For the graph structure of~$J_j$ and~$I_j^{*-}$ we need another
  $A_2\lvert S^*_g\rvert$ bits for some constant $A_2$.

\item There is a constant $A_3 > 0$ such that in the above procedure
  at most $s + 16t + s/4 = s \cdot A_3$ bits $b$ are read:
  \begin{itemize}
  \item there are at most $s+16t$ bits that are $1$ as each time a
    disappearing variable is discovered, and this is bounded by
    \cref{clm:size-s}, and
  \item at most $s$ bits that are $0$ as a stage is ended each time and
    $g \le s/4$.
  \end{itemize}
\end{itemize}
Since~$\lvert S^*_g\rvert \leq s/4 + 16t \leq s \cdot A_4$ for some
other constant~$A_4$ this completes the proof of
\cref{lemma:encode}. As \cref{lemma:encode} is the last missing piece
of the proof of the switching lemma, we thereby also establish
\cref{lemma:switch}.

\section{Proof of the Multi-Switching Lemma}
\label{sec:multi-switch}

The purpose of this section is to prove the multi-switching lemma,
restated here for convenience.

\multiSwitch*

The proof of \cref{lemma:multiswitch} very much follows the proof of
\cref{lemma:switch}. The first section gives a short proof overview,
followed by \cref{sec:cpdt} that formally explains how to construct a
common partial decision tree. In \cref{sec:multi-proof} we finally
prove \cref{lemma:multiswitch}.

\subsection{Proof Overview}

The high level proof outline is as follows.
Consider~$\bigvee_{i=1}^{m_1} T^1_{i}, \bigvee_{i=1}^{m_2} T^2_{i},
\ldots, \bigvee_{i=1}^{m_M} T^M_{i}$ in order. If a disjunction of
these decision trees~$\bigvee_{i=1}^{m_j} T^j_{i}$ is not turned in to
a decision tree of depth~$\ell$, then find a branch~$\lambda^j$ in the
extended canonical decision tree of~$\bigvee_{i=1}^{m_j} T^j_{i}$ of
length at least~$\ell$. Put the variables on~$\lambda^j$ in the common
partial decision tree, query these variables as well as some extra
variables and recurse.

As in the proof of the switching lemma we consider any full
restriction~$\sigma = \sigma(\rho, \pi)$ for which
\cref{lemma:multiswitch} fails. We then turn~$\rho$ into a restriction~$\rho^*$ such that the mapping
can be inverted with little extra information to argue that there are
few full restrictions for which \cref{lemma:multiswitch} fails.

The construction of~$\rho^*$ is analogous to the construction used in
the proof of the switching lemma: fix a long branch~$\tau$ in the
$\ell$-common partial decision tree and consider the long
branches~$\lambda^1, \lambda^2, \ldots$ from the respective extended
canonical decision trees used to construct~$\tau$. Each
such~$\lambda^j$ was constructed with the help of possible forcing
information~$J^j_1, J^j_2, \ldots$. Close up each such~$J^j_i$ as
before to obtain~$K^j_i$ and
let~$\rho^* = (\rho, \bigcup_{i,j} K^j_i)$.

We recover~$\rho$ from~$\rho^*$ by iteratively recovering the possible
forcing information~$J^j_1, J^j_2, \ldots$ as in the standard
switching lemma, to obtain the information sets~$I^j_1, I^j_2, \ldots$
describing the long branch~$\lambda^j$. We then read a bit of extra
information to recover the sub-assignment~$\tau^j$ of~$\tau$ to the
variables assigned by~$\lambda^j$.

There is one minor complication to handle: the recovered information
sets~$\bigcup_{i} I^j_i$ of a branch~$\lambda^j$ may be
\emph{in}consistent with future forcing information: a set~$I^j_i$ may
be inconsistent on chosen centers with some~$J^{j'}_{i'}$
for~$j' > j$. This is so because the set~$J^{j'}_{i'}$ is consistent
with~$\tau^j$ but not necessarily with the branch~$\lambda^j$.

We handle this complication by not just querying the variables
of~$\lambda^j$ in the common partial decision tree but to query all
variables with an associated center at distance at most~$1$ from the
exposed chosen centers of~$\lambda^j$. Analogous to how we proved in
\cref{lem:J-disjoint} that different possible forcing information have
disjoint support, it can be shown that future~$J^{j'}_{i'}$ have a
support disjoint of the support of some~$I^j_i$ with~$j < j'$. This
implies in particular that~$J^{j'}_{i'}$ cannot contain information
pieces in direct contradiction with~$I^j_i$.

\subsection{Common Partial Decision Trees}
\label{sec:cpdt}

Let us explain how to construct the $\ell$-common partial decision
tree~$\calT$
of~$\restrict{\bigvee_{i=1}^{m_1} T^1_{i}}{\sigma}, \ldots,
\restrict{\bigvee_{i=1}^{m_M} T^M_{i}}{\sigma}$. Start with~$\calT$
empty. We proceed in \emph{rounds}. In each round we consider a leaf
$\tau$ of $\calT$ such that there is a~$\bigvee_{i=1}^{m_j} T^j_{i}$
that cannot be represented by a depth~$\ell$ decision tree
under~$\sigma$ and~$\tau$. Extend~$\calT$ at~$\tau$ as follows.

Let~$j$ be minimum such
that~$\restrict{\bigvee_{i=1}^{m_j} T^j_{i}}{\sigma \tau}$ cannot be
represented by a depth~$\ell$ decision tree. Create the extended
canonical decision tree~$T^j$
of~$\restrict{\bigvee_{i=1}^{m_j} T^j_{i}}{\sigma \tau}$ essentially
as in \cref{sec:ext-can-dec-tree} -- more details follow. Denote
by~$\lambda$ a branch of length at least~$\ell$ in~$T^j$ and
extend~$\calT$ at~$\tau$ by querying all variables on~$\lambda$.
Modulo the precise definition of the extended canonical decision tree
used this describes the entire creation process of an~$\ell$-common
partial decision tree.

Let us discuss how to construct the extended canonical decision trees
in the above procedure. The only difference to the definition in
\cref{sec:ext-can-dec-tree} is that we initialize the set of exposed
centers~$S$ and the information set~$I$ used in the creation of the
extended canonical decision tree with information from previous
rounds. Let us explain this in more detail.

Throughout the creation of the~$\ell$-common partial decision
tree~$\calT$ we maintain the following sets for each leaf~$\tau$ of
$\calT$:
\begin{enumerate}
\item a set of exposed centers $S = S(\tau, \sigma)$, and
\item a set of information pieces $I = I(\tau, \sigma)$.
\end{enumerate}

Initially the set of exposed centers~$S$ and the information set~$I$
are empty~$S(\emptyset, \sigma) = T(\emptyset,\sigma)
=\emptyset$. Throughout the creation of~$\calT$ and the various
extended canonical decision trees we maintain the same invariants as
in \cref{inv:ext-can-dec-tree}.

In each round of the construction of~$\calT$, when building the
extended canonical decision tree~$T^j$ of~$\bigvee_{i=1}^{m_j}T^j_i$,
we initialize the sets~$S$ and~$I$ used in the creation of~$T^j$ with
the corresponding objects maintained for the creation of the common
partial decision tree~$\calT$. Other than that the creation of~$T^j$
follows \cref{sec:ext-can-dec-tree}: in each stage we find a new
forceable branch~$\psi_i^j$, the corresponding forcing
information~$J_i^j$, and add all nodes in~$\supp(J_i^j)$ to $S$ along
with the chosen centers adjacent to~$\supp(J_i^j) \cap \chosen$.

To find out whether the forceable branch~$\psi_i^j$ is followed we get
information sets~$I_i^j$ consisting of pieces from~$\pi$ and answers
from~$T^j$.

We continue with the next stage until at least~$\ell/4$ centers have
been added in this round to the set of exposed
centers~$S(\lambda^j, \sigma)$ for some leaf~$\lambda^j$ of~$T^j$. We
know that this happens
as~$\restrict{\bigvee_{i=1}^{m_j}T^j_i}{\sigma\tau}$ could not be
decided by a decision tree of depth~$\ell$ (recall that~$\tau$ is the
leaf of~$\calT$ we are considering).

For such a long branch~$\lambda^j \in T^j$ denote by~$S_{\lambda^j}$
all chosen centers at distance at most~$1$ from the newly exposed
chosen centers, that is, $S_{\lambda^j}$ consists of all chosen
centers at distance at most~$1$
from~$\bigl( S(\lambda^j, \sigma) \setminus S(\tau, \sigma)\bigr) \cap
\chosen$. Extend the common partial decision tree~$\calT$ at~$\tau$ by
querying all variables incident to~$S_{\lambda^j}$. For each newly
created leaf~$\tau' \in \calT$ we need to update the set of exposed
centers~$S$ and the information set~$I$:
let~$S = S(\lambda^j, \sigma) \cup S_{\lambda^j}$ and set~$I$
according to the answers on~$\tau' \setminus \tau$ while also
including the same information about~$\pi$ as is present
in~$I(\lambda^j, \sigma)$. This completes the description of the
construction of the common partial decision tree~$\calT$.

Note that at the end of each round the information pieces in~$I$ are
determined by the branch~$\tau$ of~$\calT$ and the matching~$\pi$. The
information pieces~$I_1^j, I_2^j, \ldots$ on the chosen centers used
to determine~$\lambda^j$ in~$T^j$ are ``forgotten''. These answers
were only used to find the long branch~$\lambda^j$.

Clearly the above process creates an~$\ell$-common partial decision
tree~$\calT$. We need to analyze the probability that we obtain a tree
of depth at least~$\depth(\calT) \geq s$.

\subsection{Encoding Cost}
\label{sec:multi-proof}

Once we have set up the machinery the proof parallels the proof of the
standard switching lemma. We need to verify that it works and no new
complications arise.

As in \cref{sec:encode-rho} we extend each possible forcing
information~$J_i^j$ to information sets~$\cl_i^j$ by adding non-edges
incident to nodes in~$\supp(J_i^j)$ that have an odd number of edges
incident in~$J_i^j$. The restriction~$\rho^*$ is obtained by applying
these~$\cl_i^j$ to~$\rho$. We need to specify the information needed
to invert this mapping, that is, the information needed to
recover~$\rho$ from~$\rho^*$.

The inversion commences in rounds. Each round essentially corresponds
to the inversion process of the standard switching lemma (see
\cref{sec:encode}). The following extra information is read per round.

\begin{itemize}

\item We require~$\log M$ bits to obtain an index~$j \in[M]$ that
  identifies~$\bigvee_{i=1}^{m_j}T^j_i$ being processed.

\item The following information read is identical to the inversion
  process of~$\bigvee_{i=1}^{m_j}T^j_i$ as in the standard switching
  lemma.

\item Once we have the long branch~$\lambda^j$ of~$T^j$ we need to
  recover the chosen centers~$v$ at distance at most~$1$ from the
  newly exposed chosen centers. We read at most~$\log \Delta$ bits for
  each such~$v$. Since there are at most linear in~$s$ many such nodes
  we can afford this -- the cost is absorbed by the
  constants~$c_1$ and $c_2$.
  
\item Finally we read the difference in values of variables queried in
  the decision tree~$T^j$ and the same variables in the common
  decision tree~$\calT$. Note that analogous to
  \cref{case:shared-1,case:contradict-1,case:alive-1} there will be
  some information pieces which will be recovered in future stages.
\end{itemize}

The inversion process of each round runs parallel to the inversion for
the standard switching lemma. We recover the information pieces used
in the single formula process and use the knowledge of the difference
to turn these into information pieces for the common decision tree.

Of the additional extra information needed (that is, the
index~$j \in [m]$ for each round, the identity of the additional
chosen centers, and the differences in values) only the index~$j$
cannot be absorbed by the constants~$A, c_1$ and~$c_2$. This extra
information causes the factor~$M^{s/\ell}$ in the bound of
\cref{lemma:multiswitch}. This completes the discussion of the proof
of \cref{lemma:multiswitch}.

\section{Conclusion}\label{sec:conclusion}

Of course our bounds are not exactly tight so there is always room for
improvement.  We could hope to get truly exponential lower bounds for
a bounded depth Frege proof, i.e., essentially bounds $2^n$ where $n$
is the number of variables.  Since any formula given by a small CNF
has a resolution proof this is the best we could hope for. As our
formulas have $O(n^2)$ variables we are off by a square. If one is to
stay with the Tseitin contradiction one would need to change the graph
and the first alternative that comes to mind is an expander graph.  We
have not really studied this question but as our current proof relies
heavily on properties of the grid, significant modifications are
probably needed.

This brings up the question for which probability distributions of
restrictions it is possible to prove a (multi) switching lemma.
Experience shows that this is possible surprisingly often.  It seems,
however, that it needs to be done on a case by case basis.  It is
probably too much to ask for a general characterization, but maybe it
could be possible to prove switching lemmas that cover several of the
known cases.

\section*{Acknowledgments}

We are indebted to Susanna F. de Rezende for suggesting the change in
the restriction that allows us to obtain better
parameters. Furthermore we are grateful to Mrinal Ghosh, Bj{\"o}rn
Martinsson and Aleksa Stankovi\'c for helpful discussions on the topic
of this paper.

Last but not least we would like to thank the anonymous referees for
the time they took to read our manuscript. The reviews we received
were incredibly detailed, contained a lot of useful suggestions, and
helped us improve the presentation of this paper considerably.

\bibliographystyle{alpha}
\bibliography{references}

\appendix
\crefalias{section}{appendix}
\section{Omitted Proofs}
\label{sec:omitted-proofs}

\subsection{\Cref{sec:prelim}}

\evenok*

\begin{proof}
  Let us first argue that the formula~$\tseitin(G_n, \alpha)$ is
  satisfiable if~$\sum_v \alpha_v$ is even. Towards contradiction
  suppose this is not the case and consider an assignment~$\beta$ that
  satisfies the maximum number of linear constraints. Note that the
  number of such violated linear constraints is even since the sum of
  the constraints
  \begin{align}
    \sum_v \sum_{e \ni v} x_e = 2\sum_e x_e
  \end{align}
  is always even. Hence if for a node~$v$ it holds
  that~$\sum_{e \ni v} \beta_e \neq \alpha_v$, then there is another
  node~$u \neq v$ such that~$\sum_{e \ni u} \beta_e \neq
  \alpha_u$. Suppose there are two such nodes~$u$ and~$v$.

  Since the graph~$G_n$ is connected we may consider a path~$P$
  connecting~$u$ to~$v$. By negating the assignment~$\beta$ along the
  path~$P$ we obtain the assignment~$\beta'$ that satisfies the
  constraints at~$u$ and~$v$. Furthermore, the parity of the edges
  incident to other nodes remains the same since we negated an even
  number of variables incident to every other node. This contradicts
  the assumption that~$\beta$ is an assignment that satisfies the
  maximum number of linear constraints.

  For the other part of the lemma we only need to recall that the
  number of satisfying assignments to a satisfiable system of linear
  equations only depends on the dimension of the system and not on the
  right hand side. This establishes the claim.
\end{proof}

\extendone*

\begin{proof}
  Denote the support of~$\alpha$ by~$U = \supp(\alpha)$ and
  let~$V \supseteq U$ be the set of nodes that consists of~$U$ and the
  endpoints of~$e$. Since~$\alpha$ is locally consistent we may
  consider an extension~$\beta \supseteq \alpha$ that satisfies all
  the constrains on the nodes in~$\closure(U)$.

  Denote by~$\gamma \supseteq \beta$ an extension of~$\beta$ to all
  variables incident to the nodes in~$\closure(V)$ that satisfies the
  maximum number of constraints on the nodes
  in~$\closure(V)$. Suppose~$\gamma$ violates the linear constraint of
  some node~$v \in \closure(V)$.

  Note that~$v \in \closure(U)^c$ since~$\beta \subseteq \gamma$
  satisfies all the constraints on~$\closure(U)$. Since the
  component~$\closure(U)^c$ is connected we may consider a path~$P$
  that starts at~$v$, ends in the giant
  component~$\closure(V)^c \subseteq \closure(U)^c$, and does not pass
  through any node in~$\closure(U)$. Negate the assignment~$\gamma$
  along the edges of~$P$ to obtain~$\gamma'$. Note that the
  assignment~$\gamma'$ satisfies the constraint on~$v$,
  extends~$\beta$, and causes no new violations
  on~$\closure(V) \setminus \closure(U)$ since it negates an even
  number of variables incident to all nodes
  in~$\closure(V) \setminus \set{v}$. This is in contradiction to the
  initial assumption that~$\gamma$ satisfies the maximum number of
  constraints on the nodes in~$\closure(V)$.

  We conclude that there is an extension~$\gamma \supseteq \alpha$
  that is locally consistent and contains the variable~$x_e$ it its
  domain. The statement follows.
\end{proof}

\noproof*

\begin{proof}
  We proceed by induction over the number of derivation
  steps. Consider the formula $F$ derived on line $\nu$ of the
  proof. Towards contradiction suppose that the $t$-evaluation
  $\varphi^\nu$ of line $\nu$ does not assign $F$ to a $1$-tree. That
  is, the decision tree~$\varphi^\nu(F)$ contains a branch $\tau$ that
  ends in a $0$-leaf. Since axioms are mapped to $1$-trees
  (\cref{prop:1-tree} of \cref{def:evaluation}) the formula $F$ has to
  be derived by a rule in the Frege system as listed in
  \cref{sec:frege}. The idea is to consider each rule separately and
  show that the $0$-branch $\tau$ causes a $0$-branch in one of the
  $t$-evaluations of a line used to derive $F$. This contradicts the
  inductive hypothesis.

  The assumption that all decision trees are of depth less than $n/16$
  ensures by virtue of \cref{lem:local-consistent-branch-restricted}
  that it is always possible to find a locally consistent branch in
  any decision tree restricted by $\tau$. We do a case distinction on
  the rule used to derive $F$.
  
  \paragraph{Exculuded Middle.} We have $F = p \vee \neg p$. By
  \cref{def:evaluation}, \cref{prop:not} the decision tree
  $T_p = \phi^\nu(p)$ is equal to $T_{\neg p} = \varphi^\nu(\neg p)$
  except that the labels of the leaves are negated. Since
  $\varphi^\nu(F)$ represents $T_p \vee T_{\neg p}$ by \cref{prop:or}
  of \cref{def:evaluation} the two restricted decision trees
  $\restrict{T_p}{\tau}$ and $\restrict{T_{\neg p}}{\tau}$ are both
  $0$-trees. This cannot be since the two restricted trees are the
  same except that the labels at the leaves are negated.
  
  \paragraph{Expansion Rule.} We have $F = q \vee p$. Let
  $T_p = \varphi^\nu(p)$. By \cref{prop:or} of \cref{def:evaluation}
  the decision tree $\varphi^\nu(F)$ represents
  $\varphi^\nu(q) \vee T_p$. This implies in particular that the
  decision tree $\restrict{T_p}{\tau}$ is a $0$-tree.

  Denote by $\nu' < \nu$ the line used to derive $F$. Note that $p$ is
  the formula on line $\nu'$ and that by the inductive hypothesis the
  decision tree~$T'_p = \varphi^{\nu'}(p)$ is a $1$-tree. Hence
  $\restrict{T'_p}{\tau}$ is also a
  $1$-tree. By \cref{lem:loc-cons-b-tree} this contradicts the assumed
  functional equivalence of the $t$-evaluations
  $\varphi^{\nu}$ and $\varphi^{\nu'}$.

  \paragraph{Contraction Rule.} We have $F = p$. Consider the formula
  $p \vee p$ on line $\nu' < \nu$ used to derive $F$. Let
  $T'_p = \varphi^{\nu'}(p)$. Since $\varphi^\nu$ and $\varphi^{\nu'}$
  are functionally equivalent, it holds that $\restrict{T'_p}{\tau}$ is a
  $0$-tree. Hence $\restrict{\varphi^{\nu'}(p \vee p)}{\tau}$ is a
  $0$-tree by \cref{prop:or} of \cref{def:evaluation}. This
  contradicts the inductive hypothesis which asserts that each line
  before $\nu$ is mapped to a $1$-tree.

  \paragraph{Association Rule.} We have $F = (p \vee q) \vee
  r$. Consider the formula $F' = p \vee (q \vee r)$ on line
  $\nu' < \nu$ used to derive $F$. Since $F$ and $F'$ are isomorphic
  by \cref{def:consistency-evaluation}
  the two decision trees $\varphi^\nu(F)$ and
  $T_{F'} = \varphi^{\nu'}(F')$ are functionally equivalent. This implies
  that $\restrict{T_{F'}}{\tau}$ is a $0$-tree. This is in direct
  contradiction to the inductive hypothesis.

  \paragraph{Cut Rule.} We have $F=(q \lor r)$. Let
  $T_q = \restrict{\varphi^\nu (q)}{\tau}$ and
  $T_r = \restrict{\varphi^\nu(r)}{\tau}$. Since by \cref{prop:or} of
  \cref{def:evaluation} the decision tree $\varphi^\nu(F)$ represents
  $T_q \vee T_r$ it holds that $\restrict{T_q}{\tau}$ and
  $\restrict{T_r}{\tau}$ are both $0$-trees.

  Suppose $p \lor q$ was derived on line $\nu' < \nu$ and
  $\neg p \lor r$ was derived on line $\nu'' < \nu$. Since the
  $t$-evaluation of line $\nu$ is functionally equivalent with both
  the $t$-evaluations of lines $\nu'$ and $\nu''$ the decision trees
  $\restrict{\varphi^{\nu '}(q)}{\tau}$ and
  $\restrict{\varphi^{\nu''}(r)}{\tau}$ are $0$-trees by
  \cref{lem:loc-cons-b-tree}.

  If any branch $\tau'$ in $\restrict{\varphi^{\nu '}(p)}{\tau}$ ends
  in a leaf labeled $0$, then the decision
  tree~$\restrict{\varphi^{\nu '}(p \lor q)}{\tau\cup\tau'}$ must be a
  $0$-tree by \cref{prop:or} of \cref{def:evaluation} and using
  \cref{lem:local-consistent-branch-restricted}.

  But $\restrict{\varphi^{\nu '}(p \lor q)}{\tau\cup\tau'}$ cannot be
  a $0$-tree by the inductive assumption. Hence
  $\restrict{\varphi^{\nu '}(p)}{\tau}$ is a $1$-tree. By repeating
  the above argument on line $\nu''$ with formulas $\neg p$ and $r$ we
  obtain that also $\restrict{\varphi^{\nu ''}(\neg p)}{\tau}$ is a
  $1$-tree. This is in direct contradiction to the assumed functional
  equivalence of the $t$-evaluations on lines $\nu'$ and $\nu''$.

  This completes the case distinction. The statement follows.
\end{proof}

\subsection{\cref{sec:singleswitch-overview}}

\matchinglemma*

\begin{proof}
  Consider the graph~$H$ defined on the set of non-chosen centers with
  an edge between any two such centers if they are in adjacent
  sub-squares. We want to show that there exists a graphical
  pairing~$\pi_0$ in~$H$. In other words, we want to partition the set
  of nodes of~$H$ such that each sub-graph induced by such a partition
  is either a single edge or a star of size 4.

  Let $m = \lceil 0.26 a\rceil$. The graphical pairing~$\pi_0$ will
  have either~$m$ or~$m+1$ edges between any two adjacent
  sub-squares. Since every node in~$H$ will have odd degree in~$\pi_0$
  the parity of the number of edges leaving a fixed sub-square is
  determined. When determining whether $\pi_0$ has~$m$ or~$m+1$ edges
  in-between two adjacent sub-squares we need to take the parity of
  the number of edges leaving each sub-square into account.

  Consider the following (satisfiable) Tseitin formula. For every pair
  of adjacent sub-squares~$s_1, s_2$ we introduce a
  variable~$y_{\set{s_1,s_2}}$ and introduce the constraint that the
  four variables incident to a single sub-square have the same parity
  as the number of edges leaving it. Note that since the number of
  non-chosen centers is even ($k$ as well as the number of chosen
  centers is odd) this is indeed a satisfiable Tseitin formula by
  \cref{lemma:evenok}. Take an assignment~$\beta$ satisfying said
  formula and determine that there are~$m+\beta_{\set{s_1,s_2}}$ many
  edges in~$\pi_0$ between the adjacent sub-squares~$s_1$ and~$s_2$.

  Consider any sub-square~$s$ and let~$b$ denote the number of
  non-chosen centers in it. We determined that there
  are~$4m + \sum_{e\ni s} y_e$ edges leaving the sub-square~$s$. This
  determines the number of degree 3 centers in~$s$ to
  \begin{align}
    c = \frac{4m + \sum_{e\ni s} y_e - b}{2} \eqperiod
  \end{align}
  Since the parity of~$\sum_{e\ni s} y_e$ and~$b$ are equal~$c$ is
  integer and because~$b \in (1\pm 0.01) a - 1$ it is also positive
  and bounded by~$c \leq 0.025a + 5$.

  Choose~$c$ non-chosen centers in the sub-square~$s$ to have degree 3
  in the graphical pairing~$\pi_0$ and connect them to non-chosen
  centers in adjacent sub-squares of designated degree 1. The
  remaining non-chosen centers can be paired up in such a manner that
  the number of edges between any two sub-squares is respected. This
  establishes the lemma.
\end{proof}

\pagebreak

\section{Switching Lemma Algorithms}
\label{sec:algo-switch}

\begin{algorithm}
  \caption{A Stage of the Construction of the Extended Canonical
    Decision Tree $\calT$ at $\tau$}
  \label{alg:cdt}
  \begin{algorithmic}[1]
    \Require{the sets~$S$ and~$I$}
    \Procedure{ExtendCanonicalDecisionTree}{$\calT, \tau, \sigma, T_1,
      \ldots, T_m$}

    \If{no $1$-branch in $T_1, \ldots, T_m$ is locally consistent with $\tau$
      and $\sigma$}

    \State $\tau \gets$ label $0$
    \State \textbf{return}

    \EndIf
    
    \State $\psi \gets$ first $1$-branch in $T_1, \ldots, T_m$ locally
    consistent with $\tau$ and $\sigma$

    \State $J \gets$ possible forcing information for $\psi$
    \Comment{Exists by \cref{lemma:forced0}}

    \State $I_\pi \gets$ all information of the connected components
    in $\pi$ of centers $v \in \supp(J) \setminus \chosen$
    \Statex\Comment{By \cref{def:forcing-info},
      \cref{item:non-chosen-pi}, it holds that $I_\pi \subseteq J$}

    \State $S_J \gets$ chosen centers at distance at most $1$
    from~$\supp(J) \cap \chosen$
    
    \Statex

    \State extend $\calT$ at $\tau$ by querying all variables incident
    to $S_J$

    \ForAll{$\tau' \gets$ locally consistent extension of $\tau$ in
      $\calT$}

    \State $I_{\tau' \setminus \tau} \gets$ information pieces from
    queries along $\tau' \setminus \tau$
    
    \State
    $S(\tau', \sigma) \gets S(\tau,\sigma) \cup S_J \cup \supp(J)$
    
    \State
    $I(\tau', \sigma) \gets
    I(\tau, \sigma) \cup
    I_\pi \cup
    I_{\tau' \setminus \tau}$
    
    \If{$I(\tau', \sigma)$ traverses $\psi$}
    \State $\tau' \gets$ label 1
    \EndIf
    \EndFor
    \EndProcedure
  \end{algorithmic}
\end{algorithm}

\begin{algorithm}
  \caption{recovers the partial restriction $\rho$ from $\rho^*$ given
    some extra information $X$}
  \label{alg:reconstruct}
  \begin{algorithmic}[1]
    \Procedure{Reconstruct}{$\rho^*, T_1, \ldots, T_m, s, t, X$}
    \State $j \gets 1$
    \State $\rho^*_{1} \gets \rho^*$
    \State $S^*_0,I^{*-}_0,E \gets \emptyset$
    \Statex
    \While{$|S^*_{j-1}| \leq s/4$}
    \State $\psi \gets$ next $1$-branch in $T_1, \ldots, T_m$
    traversed by $\rho_{j}^{*}$
    \While{$\psi$ and $I^{*-}_{j-1}$ \emph{not} in conflict with $E$}
    \State $\mathrm{discover} \gets$ next bit from $X$
    \Comment{associated center to discover on $\psi$?}
    
    \If{$\mathrm{discover}$}
    \State $i \gets$ next $\log(t)$ bits from $X$
    \State $v \gets$ associated center of $i$th variable on $\psi$
    \State $\sign \gets$ next $9$ bits from $X$
    \State $E \gets E \cup \set{(v,\sign)}$

    \Else
    \Comment{We found the forceable branch of stage $j$}
    \State $\psi_j \gets \psi$
    \State $\rho^*_j, I^{*-}_{j}, S^*_j \gets$
    \Call{RecoverForcingInformation}{$E, I^{*-}_{j-1}, \psi_j,
      \rho^*_{j-1}, S^*_{j-1}, X$}
    \State $E \gets E$ with used signatures removed and new ones added
    \State $j \gets j+1$
    \State \textbf{break}
    \EndIf

    \EndWhile
    \EndWhile
    \Statex
    \State $\rho \gets \rho^*_{j}$ with the assignment flipped along
    edges in~$I^{*-}_{j-1}$
    \State \textbf{return} $\rho$
    \EndProcedure
  \end{algorithmic}
\end{algorithm}

\begin{algorithm}
  \caption{Recover the objects from a single stage given the
    forceable branch~$\psi$}
  \label{alg:recoverForcingInfo}
  \begin{algorithmic}[1]
    \Procedure{RecoverForcingInformation}{$E, I^{*-}, \psi, \rho^*, S^*, X$}
    \State $E_{\psi} \gets $ set of $(v,\sign)\in E$ where $v$
    associated center of a variable on $\psi$
    \State $J \gets \emptyset$
    \ForAll{$(v, \sign_v) \in E_{\psi}$}
    \State $c, d_1, \ldots, d_4, e_1, \ldots, e_4 \gets \sign_v$
    \Comment{split the signature into single bits}
    \For{$i = 1, \ldots, 4$}
    \If{$d_i$}
    \Comment{there is a variable on $\psi$ in direction $i$ from $v$}
    \State $R_i \gets$ sub-square adjacent to $v$ in direction $i$
    \If{$e_i$}
    \State $u_i \gets$ \Call{GetPossiblyDeadCenter}{$R_i, E, c, X$}
    \State $J \gets J \cup \bigl\{\set{v,u_i}\bigr\}$
    \Else
    \State $J \gets J \cup \set{(v, i, \bot)}$
    \EndIf
    \EndIf
    \EndFor
    \EndFor
    \ForAll{$(v, \sign_v) \in E_{\psi}$}
    \Comment{it remains to recover connected components of~$\pi$}
    \State $c, d_1, \ldots, d_4, e_1, \ldots, e_4 \gets \sign_v$
    \If{$\lnot c$}
    \Comment{$v$ is not a chosen center}
    \State $\mathrm{cc} \gets$ next $\log(20)$ bits from $X$
    \Comment{encodes what kind of component $v$ is in $\pi$}
    \Statex
    \State $C \gets$ centers from $\mathrm{cc}$ that are connected by
    an edge in $J$ from $v$
    \Statex
    \ForAll{sub-squares $R$ in which $\mathrm{cc}$ has a center}
    \Comment{recover remaining centers}
    \If{no center in $C$ from $R$}
    \State $C \gets C \cup \phantom{}$\Call{GetPossiblyDeadCenter}{$R, E, c, X$}
    \EndIf
    \EndFor \Statex
    \ForAll{edges $(c_1, c_2)$ in $\mathrm{cc}$}
    \Comment{ensure that all edges are present}
    \State $w_1,w_2 \gets $ centers in $C$ corresponding to $c_1$ and $c_2$
    \State $J \gets J \cup \bigl\{\set{w_1,w_2}\bigr\}$
    \EndFor \Statex
    \ForAll{non-edges $(c,\delta,\bot)$ in $\mathrm{cc}$}
    \Comment{ensure that all non-edges are present}
    \State $w \gets $ center in~$C$ corresponding to $c$
    \State $J \gets J \cup \set{(w,\delta,\bot)}$
    \EndFor \EndIf
    \EndFor \Statex
    \State $K \gets$ \Call{RecoverK}{$I^{*-}, J$}
    \State $\rho^* \gets \rho^*$ with assignment flipped along edges in $K$
    \State $S \gets$ \Call{RecoverExposed}{$I^{*-}, J, E, X$}
    \State $S^+ \gets$ \Call{RecoverNonExposed}{$E, S, X$}

    \State $I \gets$ read from $X$ structure of $I$ on the centers $S \cup S^+$
    \State $I^{*-} \gets I^{*-} \cup I \cup \phantom{}$information from $J$ incident to nodes in~$\supp(I^{*-})$
    \State \textbf{return} $(\rho^*, I^{*-}, S^* \cup S)$
    \EndProcedure
  \end{algorithmic}
\end{algorithm}

\begin{algorithm}
  \caption{Recover $\cl$ from $J$}
  \label{alg:K}
  \begin{algorithmic}[1]
    \Procedure{RecoverK}{$I^{*-}, J$}
    \State $K \gets J$
    \For{$u \in \supp(J)$}
    \If{$u$ has odd number of edges incident in $J$}
    \State $\delta \gets$ direction in which $u$ has no information in
    $J \cup I^{*-}$
    \State $K \gets K \cup (u, \delta, \bot)$
    \EndIf
    \EndFor
    \State \textbf{return} $K$
    \EndProcedure
  \end{algorithmic}
\end{algorithm}

\begin{algorithm}
  \caption{Recover the exposed centers at distance $\leq 1$ from the
    chosen center in~$\supp(J)$}
  \label{alg:S}
  \begin{algorithmic}[1]
    \Procedure{RecoverExposed}{$I^{*-}, J, E, X$}
    \State $S \gets \supp(J)$
    \For{$u \in \supp(J) \cap \chosen$}
    \For{$\delta$ direction}
    \State $R \gets $ sub-square in direction $\delta$ of $u$
    \If{$R$ has no chosen center in $\supp(I^{*-}\cup J)$}
    \State $S \gets S
    \cup\phantom{}$\Call{GetPossiblyDeadCenter}{$R, E, 1, X$}
    \EndIf
    \EndFor
    \EndFor
    \State \textbf{return} $S$
    \EndProcedure
  \end{algorithmic}
\end{algorithm}

\begin{algorithm}
  \caption{Recover the non-exposed centers incident to the exposed
    chosen centers in~$S$}
  \label{alg:non-exposed}
  \begin{algorithmic}[1]
    \Procedure{RecoverNonExposed}{$E, S, X$}
    \State $S^+ \gets \emptyset$
    \For{$u \in S \cap \chosen$}
    \For{$\delta$ direction}
    \State $R \gets $ sub-square in direction $\delta$ of $u$
    \If{$R$ has no chosen center in $S$}
    \State $\mathrm{recover} \gets$ bit from $X$
    \If{$\mathrm{recover}$}
    \State $S^+ \gets S^+ \cup\phantom{}$\Call{GetPossiblyDeadCenter}{$R, E, 1, X$}
    \EndIf
    \EndIf
    \EndFor
    \EndFor
    \State \textbf{return} $S^+$
    \EndProcedure
  \end{algorithmic}
\end{algorithm}

\begin{algorithm}
  \caption{Get endpoint in sub-square $R$ potentially with the help of
    signatures from $E$}
  \label{alg:get-dead-center}
  \begin{algorithmic}[1]
    \Procedure{GetPossiblyDeadCenter}{$R, E, \mathrm{chosen}, X$}
    \State $\mathrm{known} \gets$ next bit from $X$
    \Comment{is the center in $R$ already in $E$ or still alive?}
    \If{$\mathrm{known}$}
    \If{$\mathrm{chosen}$}
    \State $u \gets$ center with lowest numbered row in~$R$ that is in~$E$ or alive
    \Else
    \State $s \gets $ number of signatures in~$E$ plus alive centers in~$R$
    \State $i \gets$ next $\log(s)$ bits from $X$
    \State $u \gets$ $i$th center in~$R$ that is in~$E$ or alive
    \EndIf
    \Else
    \State $i \gets$ next $\log(\Delta)$ bits from $X$
    \State $u \gets$ $i$th center in~$R$
    \EndIf
    \State \textbf{return} $u$
    \EndProcedure
  \end{algorithmic}
\end{algorithm}

\end{document}